\documentclass[12pt,a4paper]{article}
\usepackage{pst-node}
\usepackage{tikz-cd} 
\usepackage{amsthm,amsmath}
\usepackage{thm-restate}

\usepackage{algorithm}
\usepackage[noend]{algpseudocode}
\makeatletter
\renewcommand{\ALG@name}{Protocol}
\makeatother

\usepackage{amsfonts}
\usepackage{quantikz}

\usepackage{tikz}
\usetikzlibrary{cd,decorations.pathreplacing,calc,positioning,fit,decorations.pathmorphing,shapes.misc,decorations.markings,math,positioning}
\usepackage{graphics,color}
\usepackage{graphicx}
\usepackage{enumerate}
\usepackage{amssymb}
\usepackage{color}
\usepackage{amsthm,amsmath}
\usepackage{comment}
\usepackage{enumerate}

\usepackage{ulem}
\usepackage{mathtools}
\mathtoolsset{showonlyrefs}
\usepackage{cite}
\usepackage{authblk}
\usepackage{hyperref}
\usepackage{doi}
\hypersetup{
    colorlinks=true,
    linkcolor=blue,
    filecolor=blue,
    urlcolor=blue,
    citecolor=blue
}

\hypersetup{
    linktoc=all,
    linkcolor=black
}

\usepackage[left=0.75in,right=1.0in,top=0.5in,bottom=1.0in]{geometry}
\usepackage{subcaption}
\newcommand{\appref}[1]{\hyperref[#1]{{Appendix~\ref*{#1}}}}
\newcommand{\be}{\begin{eqnarray} \begin{aligned}}
\newcommand{\ee}{\end{aligned} \end{eqnarray} }
\newcommand{\benn}{\begin{eqnarray*} \begin{aligned}}
\newcommand{\eenn}{\end{aligned} \end{eqnarray*}}
\newcommand*{\textfrac}[2]{{{#1}/{#2}}}

\newcommand*{\cB}{\mathcal{B}}
\newcommand*{\cC}{\mathcal{C}}
\newcommand*{\cE}{\mathcal{E}}
\newcommand*{\cF}{\mathcal{F}}
\newcommand*{\cG}{\mathcal{G}}
\newcommand*{\cH}{\mathcal{H}}

\newcommand*{\cK}{\mathcal{K}}
\newcommand*{\cL}{\mathcal{L}}
\newcommand*{\cM}{\mathcal{M}}
\newcommand*{\cN}{\mathcal{N}}

\newcommand*{\tr}{\mathop{\mathrm{tr}}\nolimits}

\newcommand{\bc}{\begin{center}}
\newcommand{\ec}{\end{center}}

\newcommand{\abs}[1]{\left|#1 \right|}				

\newcommand{\e}{\mathrm{e}}

\newtheorem{theorem}{Theorem}[section]
\newtheorem{lemma}[theorem]{Lemma}

\newtheorem{corollary}[theorem]{Corollary}

\usepackage{amsfonts}

\def\01{\{0,1\}}
\newcommand{\ceil}[1]{\lceil{#1}\rceil}

\newcommand*{\ExpE}{\mathbb{E}}

\usepackage[OT2,T1]{fontenc}
\DeclareSymbolFont{cyrletters}{OT2}{wncyr}{m}{n}
\DeclareMathSymbol{\Sha}{\mathalpha}{cyrletters}{"58}

\newcommand*{\gkp}{\mathsf{GKP}}

\newcommand*{\round}[1]{\lfloor#1\rceil}
\newcommand*{\tGKP}{\mathsf{gkp}}

\newcommand*{\accepts}{\mathsf{accepts}}

\newcommand*{\acc}{\mathsf{acc}}
\newcommand*{\rej}{\mathsf{rej}}

\newcommand*{\vac}{\mathsf{vac}}
\newcommand*{\poly}{\mathsf{poly}}
\newcommand*{\her}{\mathsf{her}}

\newcommand*{\rhs}{RHS }

\usepackage{makecell}

\begin{document}

\title{The complexity of Gottesman-Kitaev-Preskill states}

\author[1,3]{Lukas Brenner}
\author[1,3]{Libor Caha}
\author[1,2,3]{\authorcr Xavier Coiteux-Roy}
\author[1,3]{Robert Koenig}
\affil[1]{School of Computation, Information and Technology, Technical University of Munich, Munich, Germany}
\affil[2]{School of Natural Sciences, Technical University of Munich, Munich, Germany}
\affil[3]{Munich Center for Quantum Science and Technology, Munich, Germany}
\maketitle

\begin{abstract}
We initiate the study of state complexity for continuous-variable quantum systems. Concretely, we consider a setup with bosonic modes and auxiliary qubits,  where available operations include Gaussian one- and two-mode operations, single- and two-qubit operations, as well as qubit-controlled phase-space displacements. We define the (approximate) complexity of a bosonic state by the minimum size of a circuit that prepares an $L^1$-norm approximation to the state.

We propose a new circuit which prepares an approximate Gottesman-Kitaev-Preskill (GKP) state~$\ket{\gkp_{\kappa,\Delta}}$. Here $\kappa^{-2}$ is  the variance of the envelope and~$\Delta^2$ is the variance of the individual peaks. 
 We show that the circuit accepts with constant probability and --- conditioned on acceptance --- the output state is polynomially close in $(\kappa,\Delta)$ to the state~$\ket{\gkp_{\kappa,\Delta}}$. 
 The size of our circuit is  linear in $(\log 1/\kappa,\log 1/\Delta)$. 
To our knowledge, this is the first protocol for GKP-state preparation with fidelity guarantees for the prepared state. 

We also show converse bounds, establishing that the  linear circuit-size dependence of our construction is optimal. This fully characterizes the complexity of GKP states.
\end{abstract}
\tableofcontents
\hypersetup{
    colorlinks=true,
    linkcolor=blue,
    filecolor=blue,
    urlcolor=blue,
    citecolor=blue
}

\section{Introduction}
\subsection{GKP states and their use}
Gottesman-Kitaev-Preskill (GKP) states are a key resource in continuous-variable (CV)  quantum information processing~\cite{GKP}. Originally introduced as basis states of an error-correcting code protecting against phase-space displacement noise, their primary use  has traditionally been in the context of quantum fault-tolerance. While the original paper~\cite{GKP} proposes the use of such states to  encode individual qubits into oscillators, a number of subsequent works make essential use of GKP states as building blocks in other fault-tolerance constructions: Examples include toric- and surface-GKP-codes (obtained by concatenating a qubit code with the GKP code)~\cite{ToricGKPCode,DesignedBiasGKPcodesresil,HighthresholdFTcomputation,NohChamberlandPhysRevA.101.012316} and oscillator-to-oscillator codes where GKP states are used in auxiliary oscillators to define encoding isometries~\cite{OscOscCodes}. Beyond the design of robust quantum memories,   several GKP-based schemes for achieving universal fault-tolerant quantum computation in CV systems subject to random displacement errors have been proposed, starting from the original work~\cite{GKP}, see e.g.,~\cite{HighthresholdFTcomputation}. Fault-tolerant  CV measurement-based computation schemes
based on GKP codes were proposed in~\cite{MenicucciCVCluster}. In this context, GKP states do not only act as a substrate for protecting quantum information, but also provide computational power: 
As shown in~\cite{Baragiolauniversality}, 
universal fault-tolerant quantum computation can be achieved using only Gaussian (linear optics) operations when a  supply of GKP states is provided.

Applications of GKP states outside the area of error correction include schemes for achieving maximal violations of CHSH-type Bell inequalities~\cite{WengeretalBellinequalityviolation,Etesseetal2014}, and procedures for process tomography of small displacements~\cite{DuivenvoordenSingleMode17} (themselves applicable to quantum fault-tolerance). Distributed sensing protocols relying on GKP codes were proposed in~\cite{Zhuang_2020,ZhouBradyZhuang22}.

\subsection{Prior work on GKP-state preparation} The versatility and central role of GKP states in these applications directly motivates the design and analysis of corresponding preparation procedures. Numerous proposals have been made in the past. We refer to~\cite{GKPreviewEikbushetcal} for a recent review, which includes a discussion of concrete physical setups. 

The primary distinction between  different proposals is  the type of non-Gaussian operation involved. In the seminal work~\cite{GKP}, the use of a (unitary) cubic coupling between two oscillators (with a Hamiltonian of the form $Q_1 (Q_2^2+P_2^2)$) was suggested. Protocols such as~\cite{Terhal_2016,Shi_2019,Weigand_2018} use a single auxiliary  qubit and a qubit-oscillator coupling (generated by a Hamiltonian of the form~$\sigma_z P$) allowing for qubit-controlled displacements. 
Such qubit-controlled displacements were also used in~\cite{CampagneEikbushetal20}, which provides a remarkable  proof-of-principle experimental demonstration of GKP-codes and associated fault-tolerance operations.

A different approach to creating GKP states is that of engineering  a ``GKP Hamiltonian'' whose ground states are (approximate) GKP code states, an approach already suggested in the seminal work~\cite{GKP}, see e.g.,~\cite{Ganeshan16, Doucot2012,Rymarzetal21}. Closely related to this are proposals  to realize such Hamiltonians as
effective evolution operators, e.g., by dynamical decoupling~\cite{Conrad21} or as effective (Floquet) Hamiltonians in  periodically driven systems, see e.g.,~\cite{sellem2024gkpqubitprotecteddissipation,KolesnikowFloquet24}. We refer to~\cite{GKPreviewEikbushetcal} for a more complete discussion of these different proposals, as well as experimental demonstrations.

\subsection{Complexity of bosonic states} \label{subsec: complexity}
We initiate the study of the complexity of  CV states, with a special focus on (approximate) GKP states.

 To define the notion of the complexity of a CV  quantum state formally, let us briefly review the established definition of  complexity for an $n$-qubit ``target'' state~$\ket{\Psi_\textrm{target}}\in (\mathbb{C}^2)^{\otimes n}$, 
see e.g.,~\cite{Aaronson2016TheCO}. For $\varepsilon\geq 0$, the (unitary state) complexity~$\cC^*_\varepsilon(\ket{\Psi_\textrm{target}})$ of~$\ket{\Psi_\textrm{target}}$   is defined  as the minimum size (i.e., number of operations) of a circuit~$U$  over a
gate set~$\cG$ which turns a product state~$\ket{0}^{\otimes (n+m)}$ consisting of~$n$ ``system'' qubits and $m$ auxiliary qubits into an $\varepsilon$-approximate version of~$\ket{\Psi_\textrm{target}}$ on the system qubits.
Closeness is measured in terms of the $L^1$-norm, i.e., the reduced density operator of the state~$U\ket{0}^{\otimes (n+m)}$ after tracing out the $m$~auxiliary qubits should satisfy
\begin{align}
    \left\|\tr_{m} U\proj{0}^{\otimes (n+m)}U^\dagger -\proj{\Psi_\textrm{target}}\right\|_1\leq \varepsilon\ .
\end{align}
 Here~$\cG$ is a set of allowed (unitary) operations that is computationally universal. For example, it could be chosen as the set of all single-qubit Hadamard 
 and $T$-gates, and all two-qubit~$\mathsf{CNOT}$ gates, but the exact choice of gate set~$\cG$ is typically irrelevant for complexity-theoretic considerations. This is because particular choices only affect overall constants and do not affect the (asymptotic) scaling.
 
 Beyond unitary state complexity,  similar natural notions of the complexity of a state can be defined by considering preparation protocols involving measurements, adaptivity (meaning that the applied gates depend on previously obtained measurement results), as well as post-selection (heralding).  Our definitions for the bosonic case mirror these concepts of complexity of an $n$-qubit state.

 \subsubsection{Allowed operations for hybrid qubit--boson systems}\label{sec: allowed operations}
In the following, we consider single-mode bosonic states~$\ket{\Psi_{\mathrm{target}}}\in L^2(\mathbb{R})$ for simplicity (but these considerations immediately generalize to multimode systems). To define a notion of complexity of such a state, we need to agree on the available resources, i.e., the auxiliary systems and operations used.  Here it is important to note that one typically considers infinite  families of states. This can be a countable family 
 $\{\ket{\Psi_n}\}_{n\in\mathbb{N}}\subset L^2(\mathbb{R})$ 
 or more generally a multiparameter family such as the family~$\{\ket{\gkp_{\kappa,\Delta}}\}_{(\kappa,\Delta)\in (0,\infty)^2}$ of GKP states we define below. In either case, our  notion of the complexity~$\cC^*(\ket{\Psi})$ of a state~$\ket{\Psi}$ should be such that  the function $n\mapsto \cC^*(\ket{\Psi_n})$ respectively $(\kappa,\Delta)\mapsto \cC^*(\ket{\Psi_{\kappa,\Delta}})$  quantifies 
 the number of operations needed to prepare these states for different elements in the family, yet with the same ($n$- respectively $(\kappa,\Delta)$-independent)  set of operations.   
In particular,  the considered set of operations should only consist of ``physically reasonable'' operations. For example, this means that 
only bounded-strength operations should be allowed, implying that highly squeezed states necessarily have high  complexity as their preparation requires  application of a large number of (squeezing)  unitaries from such a gate set. Our choice of operations will satisfy this requirement.

In our work, we are particularly interested in qubit--oscillator couplings as an experimental building block for designing GKP preparation procedures, similar to Refs.~\cite{Terhal_2016,Shi_2019,Weigand_2018}. We refer to Ref.~\cite{liu2024hybridoscillatorqubitquantumprocessors}, which provides a thorough review of the state-of-the-art of such operations, including, in particular, an extensive discussion of concrete physical realizations. In addition to such qubit--oscillator operations, we allow arbitrary single- and two-qubit 
operations, and (limited) bosonic Gaussian operations.  

In more detail, the systems we consider consist of one ``system'' oscillator (boson), $m$ auxiliary oscillators (bosons)  and $m'$~auxiliary qubits. That is, it is described by the Hilbert space~$L^2(\mathbb{R})\otimes L^2(\mathbb{R})^{\otimes m}\otimes (\mathbb{C}^2)^{\otimes m'}$. The set of operations we consider consists of the following:
\begin{enumerate}[(i)]
\item\label{it:preparation}
Preparation of the single-qubit computational basis state~$\ket{0}$ on a qubit, and of the vacuum state~$\ket{\vac}$ on any mode. (Here~$\ket{\vac}\in L^2(\mathbb{R})$ is the ground state of the harmonic oscillator Hamiltonian, i.e., a centered Gaussian state saturating Heisenberg's uncertainty relation for the position- and momentum-quadratures.)

\item \label{it:singletwoqubitunitary}
Arbitrary single- and two-qubit unitaries acting on any qubit, or any pair of qubits.

\item\label{it:constantstrenghquadratic}
Gaussian one- and two-mode unitaries generated by Hamiltonians, which are quadratic in the mode operators, and have bounded strength.
In more detail, for an $n$-mode bosonic system (here $n=m+1$), we consider Hamiltonians of the form
\begin{align}
H(A)&=\frac{1}{2} \sum_{j, k=1}^{2 n} A_{j, k} R_j R_k=:\frac{1}{2}R^T AR
\end{align}
where $R=(Q_1,P_1,\ldots,Q_n,P_n)$ denotes the vector of quadrature operators, and where $A=A^T\in \mathsf{Mat}_{2n\times 2n}(\mathbb{R})$ is a symmetric matrix with real entries. Denoting by $J=-J^T\in  \mathsf{Mat}_{2n\times 2n}(\mathbb{R})$ the symplectic form defined by the canonical commutation relations $[R_j,R_k]=iJ_{j,k}I$ (where $I$ denotes the identity on the Hilbert space), the action of the associated Gaussian  unitary~$U(A)=e^{iH(A)}$ on~$L^2(\mathbb{R})^{\otimes n}$ can be described by the sympletic group  element
\begin{align}
S(A)&=e^{A J}\in \mathrm{Sp}(2 n)=\left\{S \in \mathsf{Mat}_{2n\times 2n}(\mathbb{R}) \mid S J S^T=J\right\}\ .\label{eq:symplecticmatrix}
\end{align}
That is, the (Gaussian) unitary~$U(A)$ acts as
\begin{align}
U(A)^\dagger R_j U(A) &=\sum_{k=1}^{2n}S(A)_{j,k} R_k\qquad\textrm{ for }\qquad j\in [2n]:=\{1,\ldots,2n\}\ .\label{eq:symplecticactionunitarygaussian}
\end{align}
The unitary $U(A)$ is a one-mode unitary acting on a mode~$j\in [n]$ if $H(A)$ is a linear combination of monomials containing the position- and momentum operators~$Q_j$ and $P_j$ associated with the $j$-th mode only. Similarly,~$U(A)$ is a two-mode unitary acting on two modes~$(j,k)$, $j\neq k$ if and only if $H(A)$ only involves monomials that are products of the operators~$\{Q_j,P_j,Q_k,P_k\}$. 
We say that $U(A)$ has bounded strength if the operator norm~$\|A\|$ is bounded by a constant, which can be arbitrary but fixed (its choice does not  affect the scaling of the resulting notion of complexity). Concretely, we will say that~$U(A)$ is bounded if
\begin{align}
\|A\|&\leq 2\pi\ .\label{eq:anormbound}
\end{align}

\item \label{it:constantstrengthlinear}One-mode phase-space displacements 
of constant strength. For an $n$-mode (here: $n=m+1$) bosonic system with (vector of) quadrature operators
$R=(Q_1,P_1,\ldots,Q_n,P_n)$,
every  vector $d\in\mathbb{R}^{2n}$ defines 
a unitary
\begin{align}
D(d)=e^{i\sum_{j=1}^{2n} d_j R_j}\ .
\end{align}
called the  (Weyl) phase-space displacement by~$d$. Its action on mode operators is given by
\begin{align}
    D(d)^\dagger R_j D(d)&=R_j+d_jI\qquad\textrm{ for }\qquad j\in [2n]\ .
\end{align}
Correspondingly, the vector~$d$ will often be referred to as a displacement vector.
The unitary $D(d)$ is a single-mode displacement (acting on a mode~$j\in [n]$) if all entries of~$d$ other than~$d_{2j-1}$ (associated with~$Q_j$) and $d_{2j}$ (associated with~$P_j$) vanish. It is of bounded strength if the Euclidean norm~$\|d\|=\sqrt{\sum_{j=1}^{2n} d_j^2}$ of $d$ is bounded by a constant. In the following, we again arbitrarily choose this constant to be~$2\pi$, i.e., we say that~$D(d)$ is a bounded-strength phase-space displacement if and only if 
\begin{align}
    \|d\|\leq 2\pi\ .\label{eq:dnormbound}
\end{align}

\item \label{it:constantstrengthcontrolledlinear}Qubit-controlled single-mode displacements of bounded strength. For a system of~$n$ bosonic modes (here $n=m+1$) and $m'$ qubits, this is defined as follows.
Given an index~$r\in [m']$ of a qubit and a vector~$d\in \mathbb{R}^{2n}$, the
qubit-$r$-controlled phase-space displacement by~$d$ is the unitary
 \begin{align}
\mathsf{ctrl}_{r}D(d)&=
I_{L^2(\mathbb{R})^{\otimes n}}\otimes \proj{0}_r\otimes + D(d)\otimes \proj{1}_r\  .
\end{align}
Here, we write $\proj{0}_r$ for the~$m'$-qubit operator~$I^{\otimes r-1}\otimes\proj{0}\otimes I^{\otimes m'-r}$, and similarly for~$\proj{1}_r$.
We say that~
$\mathsf{ctrl}_{j}D(d)$ is single-mode
and of bounded strength if this is the case for the displacement operator~$D(d)$.
\item \label{it: measurements}Homodyne ($Q$-)quadrature-measurements of any  bosonic mode, and computational basis measurements on any qubit.
The former is defined in terms of the spectral decomposition of the quadrature operator, i.e., measuring the position (for $Q$) or momentum (for $P$). We assume that these measurements are destructive, i.e., we do not consider post-measurement states on the measured modes. 
\end{enumerate}
In the following, we refer to the operations from~\eqref{it:preparation}-\eqref{it: measurements} as elementary operations. Moreover, we define the set of all unitaries of the form~\eqref{it:singletwoqubitunitary}, \eqref{it:constantstrenghquadratic},~\eqref{it:constantstrengthlinear} and~\eqref{it:constantstrengthcontrolledlinear} as $\cG$.

We consider the use of each of these elementary operations as contributing one unit of complexity to any protocol. Since we further assume that any bosonic mode or qubit is initialized before further manipulations, this means that the overall complexity of any protocol will be lower bounded by the total  number of bosonic modes and qubits involved. (The latter quantity is sometimes referred to as the width of a circuit/protocol.)

Let us comment on the choice of the unitaries in the above list. We note that using
 the Solovay--Kitaev  algorithm (see e.g.,~\cite{dawsonnielson_solovaykitaev}), any single-qubit or two-qubit unitary~$U$ can be approximated to arbitrary precision by a product of generators from a finite universal gate set such as $\{H,T,\mathsf{CNOT}\}$.
 For convenience, we include the set of all single- and two-qubit unitaries in our set of operations, see~\eqref{it:singletwoqubitunitary} above. As a consequence, our complexity measure involves the number of one- and two-qubit operations (but not the number of individual gates from a finite gate set).
 
 Similar observations can be made about our bosonic operations: The bosonic and the qubit--boson unitaries~\eqref{it:constantstrenghquadratic},~\eqref{it:constantstrengthlinear} and~\eqref{it:constantstrengthcontrolledlinear} can be generated by the following set of generators:
\begin{enumerate}[(a)]
\item\label{it:constantstrengthgaussianunitaries}
One- and two-mode Gaussian unitary operations of constant strength. 
Specifically, we use  single-mode displacements of the form
\begin{align}
    e^{ia Q}, e^{iaP}\qquad\qquad\textrm{ with }\qquad a\in [-2\pi,2\pi]\ ,\label{eq:displacementsbasic}
\end{align}
single-mode squeezing operations of the form
\begin{align}
    S(z) &= e^{i \frac{z}{2}\left(Q P+P Q\right)}\qquad\textrm{ with }\qquad z\in [-2\pi,2\pi]\ ,\label{eq:squeezerbasic}
\end{align}
single-mode phase-space rotations (phase shifts), i.e.,
\begin{align}
    P(\phi) = e^{i\frac{\phi}{2} (Q^2 +P^2)}\qquad \textrm{with}\qquad\phi \in [-2\pi, 2\pi] \ , \label{eq: phase shift}
\end{align}
and two-mode beamsplitters acting on two modes~$j$ and $k$, that is,
\begin{align}
B_{j, k}(\omega) &= e^{i \omega\left(Q_j Q_k+P_j P_k\right)}
\qquad\textrm{ with }\qquad \omega \in [-2\pi,2\pi]\ .\label{eq:beamsplitterbasic}
\end{align}
Note that these gates generate the group of Gaussian unitaries on the bosonic modes. Indeed, the group of passive (i.e., number-preserving) Gaussian unitaries (corresponding to the orthogonal symplectic group) is generated by beamsplitters and phase shifters only~\cite{reckzeilinger}.
Note that,  by the Euler decomposition
(which factorizes every symplectic matrix as $S=O_1ZO_2$ with $O_1,O_2$ orthogonal symplectic and $Z$ diagonal symplectic, see~\cite{arvind_real_1995}), every Gaussian unitary is a product of passive Gaussian unitaries and single-mode squeezers.
Since passive Gaussian unitaries can be compiled into a sequence of phase-shifters and beamsplitters (see~\cite{reckzeilinger} for a description of a corresponding algorithm),  it follows that every Gaussian unitary can be realized by beamsplitters, phase-shifters and single-mode-squeezing operations.

\item  \label{it:constriesantqubitcontrolled}
Qubit-controlled phase-space displacements  of constant strength. These are unitaries of the form
\begin{align}
\renewcommand*{\arraystretch}{1.225}
\begin{matrix*}[l]
\mathsf{ctrl}e^{i\theta P}&:=\proj{0}\otimes I+\proj{1}\otimes \exp(i\theta P)\ , \qquad  \theta \in [-2\pi,2\pi]\,\\
\mathsf{ctrl}e^{i\theta Q}&:=\proj{0}\otimes I+\proj{1}\otimes \exp(i\theta Q)\ , \qquad  \theta \in [-2\pi,2\pi]\ .
\end{matrix*}\label{eq:controlleddisplacementconstant}
\end{align}

\end{enumerate}
In fact, the protocols we propose are expressed 
entirely in terms of the generators~\eqref{it:constantstrengthgaussianunitaries}--\eqref{it:constriesantqubitcontrolled} as well as single-qubit Clifford unitaries and the qubit--boson operation~\eqref{it:constantstrengthcontrolledlinear}.

The  constant-strength restriction in~\eqref{it:constantstrenghquadratic},~\eqref{it:constantstrengthlinear},~\eqref{it:constantstrengthcontrolledlinear}  is motivated by the fact that coupling strengths derived from physical interactions are typically constant (i.e., independent of the system size). In particular, this means that the evolution time required to realize such a unitary depends (typically linearly) on the involved parameter. A physically reasonable notion of complexity should take this into account. We achieve this by restricting our parameters to fixed  finite intervals. We note that the choice of constants
in~\eqref{eq:anormbound} and~\eqref{eq:dnormbound} (and similarly the choice of~$2\pi$ in~\eqref{eq:displacementsbasic}-\eqref{eq:controlleddisplacementconstant}) is arbitrary and does not affect the overall scaling we obtain in asymptotic settings.
 
\subsubsection{Unitary state complexity and circuit complexity}\label{sec: unitary state complexity}
Let us first consider the problem of preparing a  target state~$\ket{\Psi_{\textrm{target}}}\in L^2(\mathbb{R})$ by means of unitary operations. Concretely, we consider protocols of the following form on a system with Hilbert space~$L^2(\mathbb{R})\otimes L^2(\mathbb{R})^{\otimes m}\otimes (\mathbb{C}^2)^{\otimes m'}$, i.e., $m+1$~oscillators and $m'$~qubits.
\begin{enumerate}[(a)]
\item 
Each qubit is initialized in the computational basis state~$\ket{0}$, and each bosonic mode is initialized in the vacuum state~$\ket{\vac}$. After this step, the system is in the state~$\ket{\Psi}=\ket{\vac}\otimes\ket{\vac}^{\otimes m}\otimes\ket{0}^{\otimes m'}$.
\item 
A sequence of~$T\in\mathbb{N}$ unitaries~$U_1,\ldots,U_T\in\cG$  is applied.
Denoting by $U=U_T\cdots U_1$ 
the corresponding unitary, the state after this step is the state~$U\ket{\Psi}$.
\item 
The output  of the protocol is the reduced density operator of the first mode, i.e., $\rho=\tr_{m,m'} U\proj{\Psi}U^\dagger$, where
$\tr_{m,m'}$ denotes the partial trace over all $m$~auxiliary modes and $m'$~auxiliary qubits. 
The output state~$\rho$
 is supposed to be a good approximation to the target state~$\proj{\Psi_{\textrm{target}}}$ (as quantified below).
\end{enumerate}
A qubit--boson protocol for  preparing
a target state~$\ket{\Psi_{\textrm{target}}}$ of this form will be referred to a unitary state preparation protocol.  
Observe that the total number of operations 
from the set~\eqref{it:preparation}--\eqref{it: measurements} of allowed operations is equal to~$T+(m+1)+m'$ 
since $m+1$~vacuum states and~$m'$ qubits are prepared initially. We say that the protocol achieves error~$\varepsilon\geq 0$
if the output state is $\varepsilon$-close in $L^1$-distance to the desired target state~$\ket{\Psi_{\mathrm{target}}}$, i.e., if
\begin{align}
\left\|
\tr_{m,m'}U\proj{\Psi}U^\dagger- 
\proj{\Psi_\textrm{target}}
\right\|_1 &\leq \varepsilon\ .\label{eq:lowerboundapproximationcomplexity}
\end{align}
 We shall say that a state~$\ket{\Psi_{\mathrm{target}}}\in L^2(\mathbb{R})$
has unitary circuit complexity~$\cC^{*}_{\varepsilon}(\ket{\Psi_{\mathrm{target}}})$ for $\varepsilon\geq 0$ if 
there is a unitary state preparation protocol 
achieving error~$\varepsilon\geq 0$
such that
\begin{align}
    \cC^{*}_{\varepsilon}(\ket{\Psi_{\mathrm{target}}})&=T+(m+1)+m'\ ,
\end{align}
and the \rhs is minimal among all unitary state preparation protocols with this property. In other words, the reduced density matrix on the first mode of $U(\ket{\vac}\otimes\ket{\vac}^{\otimes m}\otimes \ket{0}^{\otimes m'})$ is $\varepsilon$-close in $L^1$-distance to the target state~$\ket{\Psi_\textrm{target}}$, and the protocol is resource-optimal in the sense that the number $T+(m+1)+m'$~of operations used is minimal.

To study this notion of (unitary) complexity of a state~$\ket{\Psi_{\textrm{target}}}$, we will often decompose a given unitary~$U\in L^2(\mathbb{R})$ into a product~$U=U_T\cdots U_1$ of unitaries~$U_1,\ldots,U_T$ belonging to our allowed gate set~$\cG$ (we only consider cases where this kind of factorization exists and is exact). We will denote by
\begin{align}
    \cC_{\cG}(U)=\min \{T\in\mathbb{N}\ |\ \exists U_1,\ldots,U_T\in\cG \textrm{ such that }U=U_T\cdots U_1\}\ \label{eq: def unitary compl gate}
\end{align}
the minimum number of unitaries needed in such a factorization.
Note that this is closely connected to the unitary state complexity of a state: If
$\|U\proj{\vac}U^\dagger-\proj{\Psi_{\mathrm{target}}}\|_1\leq \varepsilon$, then the unitary state complexity is bounded by~$\cC^*_\varepsilon(\ket{\Psi_{\textrm{target}}})\leq \cC_{\cG}(U)+1$. This is because the factorization of the unitary~$U$ into unitaries from~$\cG$ provides a protocol for preparing~$\ket{\Psi}$, and this protocol involves only the preparation of one vacuum state and application of~$\cC_{\cG}(\ket{\Psi_{\mathrm{target}}})$ gates from~$\cG$.
\subsubsection{Heralded state complexity} \label{sec: heralded state complexity}
Going beyond unitary preparation, we will also consider heralded state generation.
Here, the protocol additionally produces an output flag~$F\in \{\acc,\rej\}$ (i.e., a classical output bit)  indicating whether the protocol accepts (succeeds) or rejects (in the case of failure) the output. In more detail, we consider protocols of the following form:
\begin{enumerate}[(i)]
\item 
Each mode is initialized in the vacuum state $\ket{\vac}$ and each qubit in the computational state $\ket{0}$, i.e., the initial state is~$\ket{\Psi}=\ket{\vac}\otimes\ket{\vac}^{\otimes m}\otimes\ket{0}^{m'}$.
\item 
Subsequently, a unitary~$U$ as described before is applied, i.e., $U$ consists of gates from the set~$\cG$. We denote by~$T_1$ the number of such gates.\label{it:unitaryapplicationtmaf}
\item 
A measurement is then applied to the auxiliary modes and the qubits of the prepared state~$U\ket{\Psi}$. More precisely, a homodyne position-measurement is applied to every auxiliary bosonic mode, and a computational basis measurement to every  auxiliary qubit. 
Let us denote the corresponding POVM by~$\{E_\alpha\}_{\alpha\in\cM}$, where $\cM=\mathbb{R}^m\times \{0,1\}^{m'}$ denotes the set of measurement outcomes.
    \item
    Depending on the measurement outcome~$\alpha$,
    a flag~$F(\alpha)\in \{\acc,\rej\}$ is computed
    by means of an efficiently computable function~$F:\cM\rightarrow\{\acc,\rej\}$.
\item 
If $F(\alpha)=\acc$, a vector~$d(\alpha)=(d_Q(\alpha),d_P(\alpha))\in\mathbb{R}^2$ is computed from the measurement outcome by means of an efficiently computable function~$d:\cM\rightarrow\mathbb{R}^2$. 
The post-measurement state on the first mode is then displaced by application of  the unitary~$D(d(\alpha))=e^{i(d_Q(\alpha)Q-d_P(\alpha)P)}$, i.e., a phase-space displacement by the vector~$d(\alpha)$.

We note that this kind of ``correction'' operation by 
translation is commonly considered, e.g., in the context of Steane-type code state preparation for GKP-codes~\cite{GKP}.

\item The output of the protocol is the output flag~$F(\alpha)$ and, assuming that~$F(\alpha)=\acc$, the reduced density operator on the first mode. Conditioned on acceptance, we want the state of  the first mode to be close to the target state~$\ket{\Psi_{\textrm{taget}}}$. 
\end{enumerate}
We call a protocol for state preparation of this form a heralded state preparation protocol. 
The definition  implies that the (average) output state of the protocol conditioned on acceptance  is given by 
\begin{align}
    \rho_\acc = \frac{1}{p_\acc}\int_{F^{-1}(\{\acc\})} p(\alpha) D(d(\alpha)) \rho^{(\alpha)} D(d(\alpha))^\dagger d\alpha \label{eq: def rho acc}
\end{align}
where we introduce normalized states $\rho^{(\alpha)}$ as 
\begin{align}
    p(\alpha)\rho^{(\alpha)} = \tr_{m,m'}\left((I\otimes E_\alpha)U\proj{\Psi}U^\dagger \right)\qquad\textrm{with}\qquad p(\alpha)= \bra{\Psi}U^\dagger(I\otimes E_\alpha)U\ket{\Psi}
\end{align} and  where 
\begin{align}
    p_\acc =  \int_{F^{-1}(\{\acc\})} p(\alpha) d\alpha \label{eq:acceptanceprobabilitydefinitionm}
\end{align}
is the probability that the protocol accepts (i.e., that $F(\alpha)=\acc$). 
We say that the protocol produces a $\varepsilon$-approximation of the state~$\ket{\Psi_{\textrm{target}}}$ with acceptance probability~$p_{\acc}$ (see Eq.~\eqref{eq:acceptanceprobabilitydefinitionm}) if
\begin{align}
    \left\|\rho_\acc- \proj{\Psi_{\mathrm{target}}}\right\|_1&\leq \varepsilon\ .
\end{align}

The number of operations 
from the set~\eqref{it:preparation}--\eqref{it: measurements}  applied in this protocol 
is determined by the number~$T_1$ of unitaries applied to implement~$U$ (see~\eqref{it:unitaryapplicationtmaf}), and the unitaries applied to implement the ``correction''~$D(d(\alpha))$ after obtaining the measurement result~$\alpha$. We consider the worst case, i.e., we  will use the  minimal number of elementary unitaries from $\cG$ required to implement~$D(d(\alpha))$, maximized over all possible shifts~$d(\alpha)$ associated with different measurement outcomes $\alpha\in F^{-1}(\{\acc\})$. That is, we use the quantity (cf.~\eqref{eq: def unitary compl gate}) 
\begin{align}
    T_2 = \sup_{\alpha\in F^{-1}(\{\acc\})}\cC_{\cG}(D(d(\alpha)))\, . \label{eq: def T_2}
\end{align}
We establish upper and lower bounds on the quantity
$\cC_{\cG}(D(d))$, for arbitrary~$d\in\mathbb{R}^2$ in Section~\ref{sec:coherentstatecomplexity}.

Taking into account the~$(m+1)+m'$ single-mode and single-qubit preparations, the maximal number~$T_1+T_2$ of unitary operations, and the $m+m'$~single-mode and single-qubit measurements, the maximal (i.e., worst-case over measurement outcomes) number of operations from~\eqref{it:preparation}--\eqref{it: measurements} in such a protocol is therefore
\begin{align}
    T_1 + T_2 + (2m+1) + 2m' \ .\label{eq:sumexpressionmtotaloperations}
\end{align} 
(Here, we do not consider  the complexity of the classical computation used to evaluate the functions~$F$ and $d$, but simply assume that this is efficient. Indeed, this is the case for the protocols we consider.)

We say that a state~$\ket{\Psi_{\mathrm{target}}}\in L^2(\mathbb{R})$
has heralded state complexity~$\cC^{*}_{p, \varepsilon}(\ket{\Psi_{\mathrm{target}}})$ for $\varepsilon\geq 0$ and $p>0$ if there is a heralded state preparation protocol which prepares
a $\varepsilon$-approximation to the state~$\ket{\Psi_{\mathrm{target}}}$
with probability at least $p_{\acc}\geq p$, whose number of operations (see Eq.~\eqref{eq:sumexpressionmtotaloperations}) is equal to
\begin{align}
    \cC^{*}_{p, \varepsilon}(\ket{\Psi_{\mathrm{target}}})&=T_1 + T_2 +(2m+1)+2m'\ ,
\end{align}
and has the property that the \rhs of this equation is minimal among all heralded state preparation protocols with this property.

\subsection{Our contribution: Complexity bounds for (approximate) GKP states}
We are interested in characterizing the complexity of states in the two-parameter family~$\{\ket{\gkp_{\kappa, \Delta}}\}_{\kappa,\Delta>0}\subset L^2(\mathbb{R})$ of states defined by 
\begin{align}
\gkp_{\kappa, \Delta}(x)&= \frac{C_{\kappa,\Delta}\sqrt{\kappa}}{\sqrt{\pi \Delta}} \sum_{z\in\mathbb{Z}}
 e^{-\kappa^2z^2/2}e^{-(x-z)^2/(2\Delta^2)}
\qquad\textrm{ for }\qquad x\in\mathbb{R}\ .\label{eq:gkpkappdeltageneraldefinition}
\end{align}
Here $C_{\kappa,\Delta}>0$ is a constant such that the vector~$\ket{\gkp_{\kappa,\Delta}}$ is normalized. 
For each pair~$(\kappa,\Delta)\in [0,\infty)\times [0,\infty)$, Eq.~\eqref{eq:gkpkappdeltageneraldefinition} defines 
an approximate (finitely squeezed) GKP state~$\ket{\gkp_{\kappa, \Delta}}$
with peaks of width~$\Delta$ localized around each integer (constituting a ``grid'' following common terminology in the GKP literature), and a Gaussian envelope of variance~$1/\kappa^2$, see Fig.~\ref{fig: GKP main}.
We note that our  conventions ensures an integer spacing of the peaks. This convention is slightly different from the one typically used in the literature. More precisely, this difference is accounted for by an application of a constant-strength squeezing operator, which amounts to rescaling the parameters~$(\kappa,\Delta)$  by (irrelevant) constant factors.
\begin{figure}
    \centering
    \includegraphics[width=0.95\linewidth]{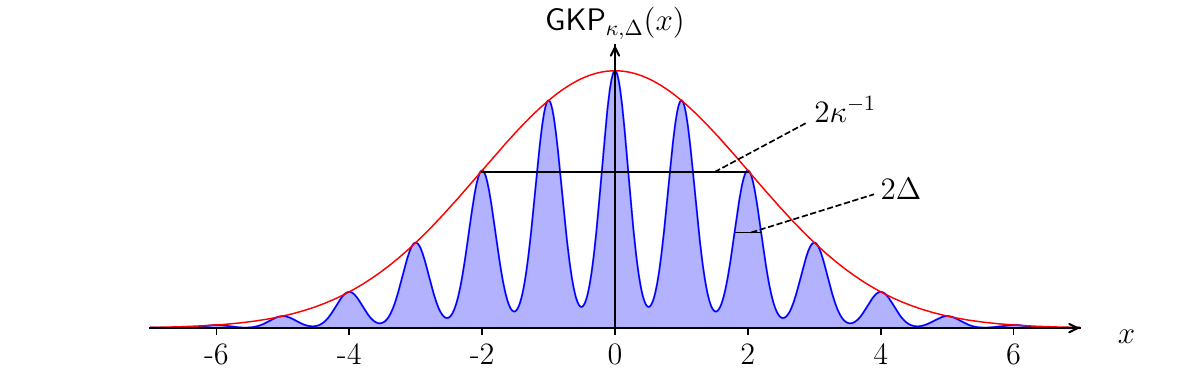}
    \caption{The approximate GKP state $\ket{\gkp_{\kappa, \Delta}}$ in position space. The red line represents the envelope~$\eta_\kappa(x)\propto e^{-\kappa^2x^2/2}$ of the state, a Gaussian  with variance~$\kappa^{-2}$.    
    The GKP wavefunction is illustrated in blue (the shading is for visual emphasis only). According to our convention, this function has Gaussian peaks of variance~$\Delta^2$ at all integers.   
    \label{fig: GKP main}}
\end{figure}

\subsubsection{New  upper and lower bounds on GKP-state preparation}
Our main result is an upper bound on the heralded complexity of approximate GKP  states. Succinctly, it can be stated
as follows: There is a polynomial~$q(\kappa,\Delta)=\poly(\kappa,\Delta)$ 
with no constant terms (i.e., $q(0,0)=0$) such that for all functions~$\varepsilon(\kappa,\Delta)$ and $p(\kappa,\Delta)$ satisfying
    $p(\kappa,\Delta) \in [0,1/10]$ and $\varepsilon(\kappa,\Delta) \ge q(\kappa,\Delta)$ for all sufficiently small $(\kappa,\Delta)$ (i.e., below some fixed constants), we have (see Corollary~\ref{cor: upper bound her state complexity}) that
    \begin{align}
        C_{p(\kappa,\Delta), \varepsilon(\kappa,\Delta)}^{*,\mathsf{her}}(\ket{\gkp_{\kappa,\Delta}}) \le O(\log1/\kappa + \log1/\Delta) \qquad\textrm{for}\qquad(\kappa,\Delta)\rightarrow(0,0)\, .\label{eq:mainresultsummarized}
    \end{align}
 Eq.~\eqref{eq:mainresultsummarized} implies that the state~$\ket{\gkp_{\kappa,\Delta}}$  with parameters~$(\kappa,\Delta)$ can be prepared with a constant success probability and an error vanishing polynomially in~$\Delta$ using resources that scale only linearly in~$\log 1/\kappa$ and $\log 1/\Delta$. We note that we are typically interested in sequence~$\{(\kappa_n,\Delta_n)\}_{n\in\mathbb{N}}$ such that $(\kappa_n,\Delta_n)\rightarrow (0,0)$ for $n\rightarrow\infty$. 
The corresponding sequence~$\{\ket{\gkp_{\kappa_n,\Delta_n}}\}_{n\in\mathbb{N}}$
of states then is an approximating sequence to the ``ideal'' GKP state, a formal linear combination of Dirac-delta distributions localized at integers.

The proof of Eq.~\eqref{eq:mainresultsummarized} actually shows a stronger statement: An approximation to the state~$\ket{\gkp_{\kappa,\Delta}}$ can be obtained probabilistically using a system consisting of only two bosonic modes and a single qubit.
We refer to Section~\ref{sec: gkp state prep} for the description and analysis of the corresponding protocol.

Complementary to this upper bound, we establish the following complexity lower bounds by using the fact that the state~$\ket{\gkp_{\kappa,\Delta}}$ has a mean photon number scaling as a function of~$(\kappa,\Delta)$, and that the operations \eqref{it:constantstrengthgaussianunitaries}--\eqref{it:constriesantqubitcontrolled} are moment-limited. We first have (see Corollary~\ref{cor: unitary state compl gkp limit}) that
\begin{align}
\cC^{*}_{1}(\ket{\gkp_{\kappa,\Delta}})&\geq \Omega\left(\log1/\kappa + \log 1/\Delta\right) \qquad\textrm{for}\qquad (\kappa,\Delta) \rightarrow (0,0)\ .
\end{align}
Furthermore, we show there exists a polynomial $s(\kappa,\Delta)$ with $s(0,0)=0$ such that for all functions $p(\kappa,\Delta)$ and $\varepsilon(\kappa,\Delta)$  satisfying $0\le \varepsilon(\kappa,\Delta) \le p(\kappa,\Delta)$ and $s(\kappa,\Delta) \le p(\kappa,\Delta)\le 1$ for sufficiently small $(\kappa,\Delta)$, we also have (see Corollary~\ref{cor: her state complexity}) that
\begin{align}
\cC_{p(\kappa,\Delta),\varepsilon(\kappa,\Delta)}^{*,\mathsf{her}}(\ket{\gkp_{\kappa,\Delta}}) \ge \Omega\left(\log1/\kappa + \log 1/\Delta \right) \qquad\textrm{for}\qquad (\kappa,\Delta) \rightarrow (0,0)\ .
    \end{align} 
This result implies, in particular, that the approximation error is lower bounded by a constant 
 (independent of $(\kappa,\Delta)$) for any protocol that has constant acceptance probability~$p_\acc$, and uses a number of elementary operations that is sublinear in~$(\log 1/\kappa,\log 1/\Delta)$. 
 
 Combining our protocol with this lower bound establishes the complexity of approximate GKP states:
\begin{corollary}
    There is a polynomial $r(\kappa,\Delta)$ with $r(0,0)=0$ such that for all functions $p(\kappa,\Delta)$ and $\varepsilon(\kappa,\Delta)$ satisfying $r(\kappa,\Delta) \le \varepsilon(\kappa,\Delta) \le p(\kappa,\Delta) \le 1/10$ for sufficiently small $(\kappa,\Delta)$, the heralded state complexity is
    \begin{align}
\cC_{p(\kappa,\Delta),\varepsilon(\kappa,\Delta)}^{*,\mathsf{her}}(\ket{\gkp_{\kappa,\Delta}}) = \Theta(\log1/\kappa + \log1/\Delta) \qquad\textrm{for}\qquad(\kappa,\Delta) \rightarrow (0,0)\, .
    \end{align}
\end{corollary}

In particular, for any two constants~$(p,\varepsilon)$ such that $\varepsilon \le p \le 1/10$, 
 the heralded state complexity~$\cC^{*,\mathsf{her}}_{(p,\varepsilon)}(\ket{\gkp_{\kappa,\Delta}})$ 
 of the approximate GKP state satisfies   \begin{align}
\cC^{*,\mathsf{her}}_{(p,\varepsilon)}(\ket{\gkp_{\kappa,\Delta}})
 &=\Theta\left(\log 1/\kappa+\log 1/\Delta\right)\qquad\textrm{for}\qquad (\kappa,\Delta) \rightarrow (0,0)\ . \end{align}

\subsubsection{Alternative figures of merit for GKP-state preparation}
Our notion of state complexity  quantifies the accuracy of the prepared states using the $L^1$-distance. To our knowledge, the consideration of this stringent quality measure for (approximate)  GKP states is new. We believe it is  especially important for their algorithmic use. 

In the quantum fault-tolerance literature, the figure of merit that is typically considered is motivated by the fact that an ideal GKP state~$\ket{\gkp} \propto \sum_{z\in \mathbb{Z}} \ket{z}$ is 
the simultaneous~$+1$-eigenstates of two commuting phase-space displacement operators~$S_P,S_Q$ (with respect to our conventions, these are $S_P=e^{-iP}$ and $S_Q=e^{2\pi iQ}$). Correspondingly, effective squeezing parameters $\Delta_P(\rho)$ and $\Delta_Q(\rho)$ of a density operator~$\rho\in\cB(L^2(\mathbb{R}))$ are introduced in~\cite{DuivenvoordenSingleMode17} as     \begin{align}
        \Delta_P(\rho) := \sqrt{\log\textfrac{1}{|\tr (S_P\rho)|^2}}\qquad\textrm{ and }\qquad 
        \Delta_Q(\rho):= \sqrt{\log\textfrac{1}{|\tr (S_Q\rho)|^2}}\ .\label{fig:effectivesqueezingparameters}
    \end{align}
    These (formally) vanish for the ideal GKP state~$\ket{\gkp}$, and, correspondingly, upper bounds on these quantities are used as a quality measure for the prepared state, see e.g., Refs.~\cite{Weigand_2018,Hastrup_2021}. In Section~\ref{sec: effective squeezing}, we show upper bounds on these quantities for states~$\rho$ produced by our protocol.
    
    We emphasize, however, that upper bounds on $\Delta_P(\rho)$ and $\Delta_Q(\rho)$ only are not sufficient to conclude that~$\rho$ is close to a state of the form~$\ket{\gkp_{\kappa,\Delta}}$. For example, although such upper bounds have been established for the output state~$\rho$ of a certain protocol (similar to ours) in~\cite{Hastrup_2021}, the state~$\rho$ obtained in that reference is  far (in $L^1$-distance) from an approximate GKP state. Nevertheless, 
    the use of the effective squeezing parameters given in~\eqref{fig:effectivesqueezingparameters} is suitable and typically sufficient for fault-tolerance applications, where the key property of the prepared states is an  approximate phase-space translation-invariance with respect to a lattice of phase-space translation vectors. Upper bounds on the quantities~\eqref{fig:effectivesqueezingparameters} provide a quantitative expression of this translation-invariance.
    
\subsubsection{Different approaches to preparing GKP states}
It is worth mentioning that a number of existing protocols for approximate GKP-state preparation are based on the defining property of ideal GKP states. For example, the  protocols proposed in~\cite{Terhal_2016,Shi_2019,Weigand_2018} rely on the idea of gradually projecting into an eigenspace of the operator $S_P$ with eigenvalue $e^{i\theta}$ for some $\theta \in [0, 2\pi)$, while extracting information about $\theta$. This is achieved by the standard phase-estimation procedures for (controlled) unitaries. Finally, the outcome state is corrected by a shift depending on an estimate  $\hat{\theta}$ of $\theta$, approximately creating a $+1$~eigenstate of $S_P$.

More precisely, the protocol first introduced in~\cite{Terhal_2016} and further modified in~\cite{Shi_2019} uses $n$ qubit ancillas to perform $n$ rounds of phase estimation by repetition. Starting from a squeezed vacuum state,  two new equidistant peaks are added in each round, with an amplitude following a binomial distribution post-selected on the outcome of certain qubit measurements being all zero. Asymptotically, the resulting
wave function can be approximated by a Gaussian distribution with envelope parameter scaling as $\kappa = O(1/n)$ by the central limit theorem.
 This  reflects the approximate projection onto the $+1$~eigenspace of $S_P$. For different (non-zero) qubit measurement outcomes, a corresponding eigenvalue can be estimated, and a suitable correction is applied.

The protocol in~\cite{Weigand_2018} has a different form, but can also be understood as applying a projection onto the~$+1$-eigenspace of~$S_P$, albeit in a gradual fashion:  It uses $n$ rounds and $2^n$ squeezed cat states as input to create an approximate GKP state with $2^n$ peaks.
In each round, two states with $2^k$ peaks are mapped to a state with $2^{k+1}$ peaks. Depending on $P$-quadrature measurement outcomes, the eigenvalue of $S_P$ can be estimated, and a correction is applied subsequently.

Two preparation protocols tailored to the platform of neutral atoms are presented in~\cite{bohnmann2024bosonicquantumerrorcorrection}. The first protocol uses post-selection on mid-circuit measurements of the auxiliary qubit, the second requires mid-circuit reset of the auxiliary  qubit.

We follow a somewhat different approach, closer to the unitary circuit proposed in~\cite{Hastrup_2021}. Our protocol proceeds in two stages: first, an approximate ``comb'' state with $2^n$~peaks (and rectangular envelope) is prepared by a process involving $n$~rounds. Subsequently, we use a protocol coupling two oscillators and a homodyne measurement to create a Gaussian envelope. We note that for some applications, the comb states may be of independent interest.

\subsubsection*{Outline}
Our  paper is structured as follows.
In Section~\ref{sec:coherentstatecomplexity},
we study the complexity of coherent states. 
In Section~\ref{sec: comb state prep}, we introduce comb states as a special kind of grid-states featuring a rectangular envelope. We provide a unitary quantum circuit that prepares these states efficiently in the number of desired peaks and with high fidelity. In Section~\ref{sec: Gaussian envelope shaping}, we present a heralding protocol that imprints a Gaussian envelope on a comb state. Consequently, in Section~\ref{sec: gkp state prep}, we combine the results from Section~\ref{sec: comb state prep} and Section~\ref{sec: Gaussian envelope shaping}, yielding a protocol that prepares approximate GKP states. In Section~\ref{sec: converse bound}, we prove a converse bound for the heralded state complexity of GKP states.

\section{The (zero-error) unitary complexity of coherent states\label{sec:coherentstatecomplexity}}
In this section, we consider the (unitary) complexity of coherent states. These states can be prepared exactly using the set of operations we consider (i.e., preparation of~$\ket{\vac}$ on bosonic modes,~$\ket{0}$ on qubits, and unitaries  belonging to the set~$\cG$, see Section~\ref{sec: allowed operations}). Motivated by this, we consider their  zero-error complexity only. This simple problem serves as a warm-up and illustrates some key concepts. 

To fix notation, let $D(d)= e^{i(d_QQ -d_PP)}$ be the (Weyl) phase-space displacement operator associated with $d= (d_Q,d_P)\in \mathbb{R}^2$. Throughout this section, we consider the coherent state
\begin{align}
\ket{d}=D(d)\ket{\vac}\qquad\textrm{ with }\qquad d=(d_Q,d_P)\in\mathbb{R}^2\ .\label{eq:coherentstatedefinition}
\end{align}

\subsection{Protocols for preparing coherent states}
Let us first consider protocols for generating the state~\eqref{eq:coherentstatedefinition}. 

Clearly, one way of generating the coherent state~$\ket{d}$ using only a single bosonic mode, no auxiliary modes (i.e., $m=0$) and no qubits~(i.e., $m'=0$) is to simply realize the displacement~$D(d)$ by a sequence of displacements: Because
$D(d_1)D(d_2)\propto D(d_1+d_2)$ by the Weyl relations, the state~$\ket{d}$ is proportional to
\begin{align}
\ket{d}\propto D(d/T)^T\ket{\vac}\qquad\textrm{ for any }\qquad T\in \mathbb{N}\ .\label{eq:factorizationdisplacmeentd}
\end{align}
Choosing $T=\lceil \|d\| \rceil$, where $\|d\|=\sqrt{d_Q^2+d_P^2}$ denotes the Euclidean norm of~$d=(d_Q,d_P)\in\mathbb{R}$ ensures that each unitary~$D(d/T)$ is a displacement of constant strength, i.e., belongs to the gate set~$\cG$.  Since~\eqref{eq:factorizationdisplacmeentd} shows that the 
state~$\ket{d}$ can be created exactly (i.e., with error~$\varepsilon=0$) with~$T=\lceil\|d\| \rceil$ gates from~$\cG$ and $1$~bosonic mode, the (zero-error) unitary state complexity of~$\ket{d}$ is upper bounded by~$\cC_0^*(\ket{d})\leq \lceil \|d\|\rceil+1$. We note that the factorization
\begin{align}
D(d)\propto D(d/\lceil \|d\|\rceil)^{\lceil \|d\|\rceil}
\label{eq:factorizationhm}
\end{align}
used here also implies that the circuit complexity of the displacement operator~$D(d)$ is bounded by~$\cC_\cG(D(d))\leq \lceil \|d\|\rceil$.

This na\"ive approach to preparing the coherent state~$\ket{d}$ expressed by~\eqref{eq:factorizationdisplacmeentd} (respectively~\eqref{eq:factorizationhm}) is, however, far from optimal: in fact, a number of gates from~$\cG$ that is only logarithmic in~$\|d\|$ suffices.
To see this, recall that the (passive) unitary~$U(\theta)=e^{-i\theta(Q^2+P^2)/2}$ realizes a rotation in phase space, i.e.,
\begin{align}
\renewcommand*{\arraystretch}{1.225}
\begin{matrix}
U(\theta) Q U(\theta)^\dagger &=&\hphantom{-}(\cos\theta)Q+(\sin\theta) P\\
U(\theta) P U(\theta)^\dagger &=&-(\sin\theta) Q+(\cos\theta) P\ 
\end{matrix}\qquad\textrm{ for }\qquad\theta\in [0,2\pi)\ ,
\end{align}
which implies that 
\begin{align}\renewcommand*{\arraystretch}{1.225}
D(d)&\propto U(\theta) D((0,c)) U(-\theta)\qquad\textrm{ with }\qquad \begin{matrix*}[l]
c&=&\|d\|\\
 \theta&=&\pi-\arg d\mod (2\pi)\in [0,2\pi)\ .
 \end{matrix*} \label{eq:dddefinitionx}
\end{align}
 It is easy to verify that 
\begin{align}
   D((0,c))=e^{-icP} = S(-\log c) e^{-iP} S(\log c) \qquad \textrm{ for all $c > 0$}\ ,\label{eq:dzerocdef}
\end{align}
where $S(z)$ for $z\in \mathbb{R}$ denotes single-mode squeezing, see Eq.~\eqref{eq:squeezerbasic}. 
Since $S(z_1)S(z_2)=S(z_1+z_2)$ for all $z_1,z_2\in\mathbb{R}$, we have 
\begin{align}
S(z)=S(z/\lceil |z|\rceil)^{\lceil |z|\rceil}\qquad\textrm{ for any }\qquad z\in\mathbb{R}\ .\label{eq:szmv}
\end{align}
Combining~\eqref{eq:dddefinitionx},~\eqref{eq:dzerocdef} and~\eqref{eq:szmv}, we obtain
the factorization
\begin{align}\renewcommand*{\arraystretch}{1.225}
D(d)\propto U(\theta_d)S(-z_d)^{N_d}e^{-iP}S(z_d)^{N_d}U(-\theta_d)\qquad\textrm{ with }\qquad 
\begin{matrix*}[l]
z_d&=&\frac{\log \|d\|}{\lceil |\log \|d\||\rceil}\\
N_d&=&\lceil |\log \|d\||\rceil\\
\theta_d&=&\pi-\arg d\mod (2\pi)\ .
\end{matrix*}\label{eq:factorizationddgreat}
\end{align}
We have $z_d\in [-1,1]$, hence~$S(z_d), S(-z_d)\in\cG$ are constant-strength Gaussian unitaries. Since we also have~$U(\theta)\in\cG$ for any $\theta\in [-2\pi,2\pi]$ and $e^{-iP}\in\cG$, it follows from~\eqref{eq:factorizationddgreat} that~$D(d)$ can be realized by a sequence of $T=2N_d+2$ gates from~$\cG$. (The first phase-space rotation~$U(-\theta_d)$ does not need to be applied because the state~$\ket{\vac}$ is invariant under such rotations.) 
In particular, using only a single mode and~$T$ gates from~$\cG$, the coherent state~$\ket{d}$ can be prepared from the vacuum state~$\ket{\vac}$. We have thus shown the following:
\begin{lemma}\label{lem:upperboundcoherentstatecomplexity}
Let $d\in \mathbb{R}^2$ be arbitrary. The circuit complexity of the displacement operator~$D(d)$ is bounded by
\begin{align}
\cC_\cG(D(d))&\leq 2\lceil |\log \|d\||\rceil+3\ .
\end{align}
In particular, the (zero-error) complexity of the coherent state~$\ket{d}$ is bounded by
\begin{align}
\cC_0^*(\ket{d})&\leq 2\lceil |\log \|d\||\rceil+3\ .
\end{align}
\end{lemma}

\subsection{Lower bounds on the complexity of coherent states}
Let us now turn to lower bounds on the complexity of a coherent state~$\ket{d}$, $d\in\mathbb{R}^2$.  Such lower bounds can be obtained by considering the (harmonic oscillator) Hamiltonian
\begin{align}
H&=Q^2+P^2\ .
\end{align}
Recalling that $\bra{\vac}H\ket{\vac}=1$ and  $\bra{\vac}Q\ket{\vac} = \bra{\vac}P\ket{\vac}=0$, it is easy to check that the energy of the coherent state~$\ket{d}$ is equal to
\begin{align}
\bra{d}H\ket{d}=\|d\|^2+1\ . \label{eq:coherentstateenergydefinition}
\end{align}
Now consider any protocol that starts from the state~$\ket{\Psi}=\ket{\vac}\otimes\ket{\vac}^{\otimes m}\otimes \ket{0}^{\otimes m'}$ on $m+1$ bosonic modes and $m'$~auxiliary qubits, applies a unitary~$U$ consisting of~$T$ gates from the set~$\cG$, and subsequently outputs the state~$\rho=\tr_{m,m'} U\proj{\Psi}U^\dagger$ on the first mode. By using moment limits, we  show below (see Lemma~\ref{lem: moment limits unitary state prep}) that the energy of the output state~$\rho$ is bounded: We have 
\begin{align}
\tr(H\rho)&\leq 
e^{8\pi T}(m+2)\qquad\textrm{ for any $T\in\mathbb{N}$, $m,m'\in\mathbb{N}\cup \{0\}$\ .}\label{eq:Hrhoupperboundm}
\end{align}
Now suppose that the considered protocol prepares the coherent state~$\ket{d}$ exactly, i.e., $\rho=\proj{d}$. 
Comparing~\eqref{eq:coherentstateenergydefinition} and~\eqref{eq:Hrhoupperboundm}, we then conclude that we must have $\|d\|^2+1\leq  e^{8\pi T}(m+2)$, i.e.,
\begin{align}
T\geq \frac{1}{8\pi}
\log\left(\frac{\|d\|^2+1}{(m+2)}\right)
&= \frac{1}{8\pi} 
\left(2\log \|d\|+\log(1+1/\|d\|^2)-\log(m+2)\right)\ .\label{eq:lowerboundfdm}
\end{align}
Using that
\begin{align}
m+1-\frac{1}{8\pi}\log(m+2)\geq 0\qquad\textrm{ for all }m\in\mathbb{N}\cup \{0\}\ ,    
\end{align}
we conclude that Eq.~\eqref{eq:lowerboundfdm} implies 
\begin{align}
T+(m+1)+m'\geq 
\frac{1}{8\pi} 
\left(2\log \|d\|+\log(1+1/\|d\|^2)\right)=:f(\|d\|)\ . \label{eq: lower bound complexiry displacement}
\end{align}
Since this inequality is satisfied for any protocol preparing the state~$\ket{d}$,  we conclude that
the (zero-error) unitary state complexity of the coherent state~$\ket{d}$ is at least $\cC_0^*(\ket{d})\geq f(\|d\|)$.

Analogous reasoning applies to the circuit complexity of the displacement operator~$D(d)$: If $U=D(d)$ for a circuit~$U$ consisting of~$T$ gates, then $T$ must satisfy~\eqref{eq:lowerboundfdm} (with $m=m'=0$) because $U$ gives rise to a protocol preparing~$\ket{d}=U\ket{\vac}$ from a single vacuum state. In particular, we have
\begin{align}
    \cC_{\cG}(D(d))&\geq f(\|d\|)-1\geq \frac{1}{4\pi}\log \|d\|\ -1\ .\label{eq: lower bound complexiry displacement complexityd}
\end{align}
In summary, we have thus shown the following:
\begin{lemma}\label{lem:upperboundcomplexitycoherent}
Let $d\in\mathbb{R}^2$ be arbitrary.
Then
\begin{align}
\cC_0^*(\ket{d})=\Omega (\log \|d\|)\qquad\textrm{ for }\qquad \|d\|\rightarrow\infty\ 
\end{align}
and
\begin{align}
\cC_\cG(D(d))=\Omega (\log \|d\|)\qquad\textrm{ for }\qquad \|d\|\rightarrow\infty\ . 
\end{align}
\end{lemma}

\subsection{The zero-error unitary state complexity of a coherent state}
Combining the preparation procedure (characterized by Lemma~\ref{lem:upperboundcoherentstatecomplexity}) with
Lemma~\ref{lem:upperboundcomplexitycoherent}, we obtain the following scaling of the zero-error complexity of a coherent state:
\begin{corollary}[Zero-error complexity of coherent states]\label{cor:coherentstatecomplexity}
For $d\in \mathbb{R}^2$, let $\ket{d}\in L^2(\mathbb{R})$ be the coherent state defined by~\eqref{eq:coherentstatedefinition}. Then, we have 
\begin{align}
\cC_0^*(\ket{d})=\Theta(\log \|d\|)\qquad\textrm{ for }\qquad \|d\|\rightarrow\infty\ .
\end{align}
\end{corollary}
The proof of Corollary~\ref{cor:coherentstatecomplexity} illustrates a few of the building blocks for our main result. For example, in the GKP-state-preparation protocol we propose, we also use the factorization~\eqref{eq:factorizationddgreat} of a phase-space displacement operator~$D(d)$. However, the protocol and its analysis are somewhat more involved, and only prepare an approximation to the state~$\ket{\gkp_{\kappa,\Delta}}$ (i.e., we consider the unitary and heralded complexities~$\cC_\varepsilon^*(\ket{\gkp_{\kappa,\Delta}})$ and ~$\cC_{p,\varepsilon}^{*,\mathsf{her}}(\ket{\gkp_{\kappa,\Delta}})$ with non-zero error~$\varepsilon>0$).

We also derive corresponding lower bounds, again using moment (energy) limits. We note that  in principle, identical arguments can be used to derive lower bounds on $\cC_\varepsilon^*(\ket{d})$ and the heralded complexity~$\cC_{p,\varepsilon}^{*,\mathsf{her}}(\ket{d})$ of a coherent state~$\ket{d}$. Since our focus is on GKP states, we omit the details here.

\section{Comb-state preparation\label{sec: comb state prep}}
In this section, we consider the problem of preparing comb states. We give the formal definition of these states in Section~\ref{sec:combstatedefinition}. 
In Section~\ref{sec:combstatepreparationprotocol}, we give a protocol for their preparation. 
In Section~\ref{sec:combstatepreparationexplanation}, we explain the underlying ideas of the protocol. Finally, in Section~\ref{sec:comb state prep error}, we prove that our protocol indeed prepares (approximate) comb states.

As we argue in Section~\ref{sec: Gaussian envelope shaping}, comb states can be converted to GKP states rather easily. In other words, the comb-state-preparation protocol presented here is a key building block for our GKP-state-preparation protocol.

\subsection{Definition of comb states\label{sec:combstatedefinition}}
Comb states are GKP-like wavefunctions that have  support on an interval of length $L>0$ and a rectangular envelope
given by 
\begin{align}
    \Box_L(z)&=\frac{1}{\sqrt{L}} \delta_{z\in [-L/2,L/2)}\ .
\end{align}
That is, for a squeezing parameter $\Delta>0$ (associated with the width of each peak), we define the states
\begin{align}
\ket{\Sha_{L,\Delta}}&=D_{L,\Delta}\sum_{z\in\mathbb{Z}} \Box_L(z) \ket{\chi_\Delta(z)}\label{eq:sha}\\
&=\frac{D_{L,\Delta}}{\sqrt{L}}\sum_{z=-L/2}^{L/2-1} \ket{\chi_\Delta(z)}\ .\label{eq:sha L}
\end{align}
Here $D_{L,\Delta}$ is normalization factor and $\chi_{\Delta}$ are translated Gaussians
\begin{align}
(\chi_\Delta(z))(x)&=\Psi_\Delta(x-z)
\qquad \textrm{ where }\qquad \Psi_\Delta(x)=\frac{1}{(\pi \Delta^2)^{1/4}}e^{-x^2/(2\Delta^2)}\,
\ .\label{eq:chiDeltadefinition}
\end{align}
Note that $\ket{\Psi_1}=\ket{\vac}$ is the vacuum state, i.e., the ground state of the harmonic oscillator Hamiltonian.
 
We call a state $\ket{\Sha_{L,\Delta}}$ a  comb state. It is illustrated in Fig.~\ref{fig:combstate}.
\begin{figure}[!h]
\begin{center}
    \includegraphics[width=0.45\textwidth]{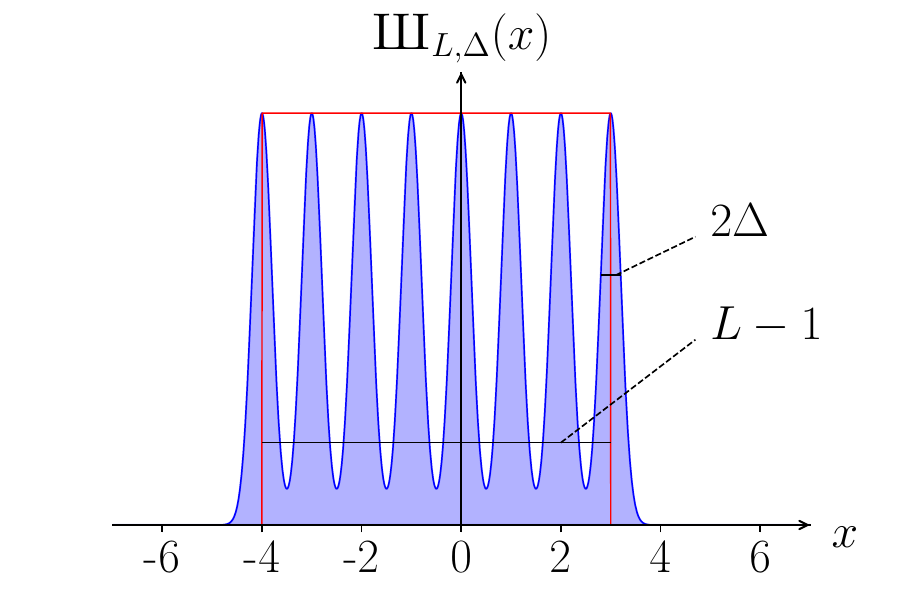}
    ~~~~~~
    \end{center}
    \caption{An illustration of the  comb state $\ket{\Sha_{L,\Delta}}$ with $L=8$ (i.e., with $8$ local maxima). For any large even integer~$L\in2\mathbb{N}$, the state~$\ket{\Sha_{L,\Delta}}$ is ``almost centered'': its peaks lie at positions $\cL_{L}:=\{-L/2,\ldots,-1,0,\ldots,L/2-1\}$.\label{fig:combstate}}
    \end{figure}

\subsection{Unitary comb-state preparation} \label{sec:combstatepreparationprotocol}

We give a unitary circuit for comb-state preparation 
 in Protocol~\ref{prot: comb state prep}. It is illustrated in Fig.~\ref{fig: circuit__protocol1}. It uses one auxiliary qubit (i.e., $m=0$ and $m'=1$).

\begin{algorithm}[H] 
	\caption{Comb-state preparation}
     \label{prot: comb state prep}
	\begin{flushleft}
		\textbf{Input:} 
		A squeezing parameter $\Delta\in(0,1/4)$ and a number of rounds $n\ge1$.\\
		\textbf{Output:} A state close to the state $\ket{\Sha_{2^n, \Delta}}$, (cf.\ Theorem~\ref{thm:comb state prep}).
		\begin{algorithmic}[1]
            \State Prepare the squeezed vacuum state $\ket{\Psi_{2^{-n}\Delta}} \gets  S(n\log2 + \log1/\Delta) \ket{\vac}$. \label{it:prot:comb:squeezed state}
            \State Prepare the qubit state $\ket{+}$. Denote the resulting state by~$\ket{\Phi^{(0)}}\gets \ket{\Psi_{2^{-n}\Delta}}\otimes \ket{+}$.
            \State Apply $(e^{iP}\otimes I)V$ to the state $\ket{\Phi^{(0)}}$ yielding $\ket{\Phi^{(1)}}\gets (e^{iP}\otimes I)V\ket{\Phi^{(0)}}$.\label{it:prot comb breeding start}
            
        \For{$k\in\{2,\ldots, n\}$}
        \State Apply $V$ to the state $\ket{\Phi^{(k-1)}}$, yielding $\ket{\Phi^{(k)}}\gets V\ket{\Phi^{(k-1)}}$.
        \label{it:prot comb breeding finish}
        \EndFor
        \State \Return the first register of the state $\ket{\Phi^{(n)}}$, i.e., the state $\tr_{\mathsf{qubit}} \proj{\Phi^{(n)}}$.
		\end{algorithmic}
	\end{flushleft}
\end{algorithm}

\begin{figure}[!ht]
	\begin{center}
 \includegraphics{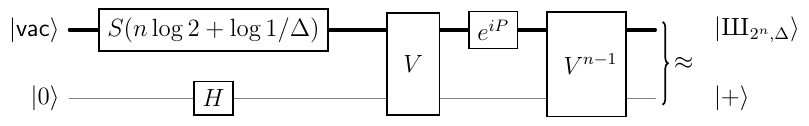}
	\end{center}
	\caption{Circuit diagram of Protocol~\ref{prot: comb state prep}. It uses the $V$ gate described in Fig.~\ref{fig: circuit V}. The exponent of $V$ indicates the number of applications. The squeezing unitary is realized by composing single-mode squeezing operations from the set~$\cG$ (see Section~\ref{sec: unitary state complexity}), see the factorization given in Eq.~\eqref{eq:squeezingdecomposition}.}
	\label{fig: circuit__protocol1}
\end{figure}

Its core is the repeated use of the unitary~$V$ defined as
\begin{align}
V := V^{(4)} V^{(3)} V^{(2)} V^{(1)} \qquad\textrm{ where }\qquad
\begin{matrix*}[l]
    V^{(1)} &=& S(-\log2) \otimes I\\
    V^{(2)} &=& \mathsf{ctrl}e^{-iP}\\
    V^{(3)} &=& I\otimes H\\
    V^{(4)} &=& \mathsf{ctrl}e^{i\pi Q} \ .
\end{matrix*}
\label{eq: V comb prep one round}
\end{align}
\begin{figure}[!ht]
	\begin{center}
\includegraphics{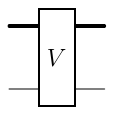}~\raisebox{0.9cm}{$:=$}~\includegraphics{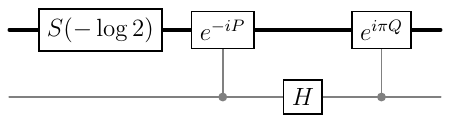}
  \end{center}
	\caption{Circuit implementing the unitary $V$ used in the comb-state-preparation protocol. It uses two qubit-controlled displacements: by $1$ in the $Q$-direction, by $\pi$-in the $P$-direction. \label{fig: circuit V}}
\end{figure}
The circuit diagram for this unitary is given in  Fig.~\ref{fig: circuit V}. The main result of this section is the following.
\begin{restatable}{theorem}{thmcombprep}\label{thm:comb state prep} Given a squeezing parameter $\Delta\in(0,1/4)$ and a number of rounds $n>0$, Protocol~\ref{prot: comb state prep} returns a quantum state $\rho\in\cB(L^2(\mathbb{R}))$ close in $L^1$-distance to the comb state $\ket{\Sha_{2^n,\Delta}}$, i.e., the output state~$\rho$ satisfies
\begin{align}
    \left\|\rho-\proj{\Sha_{2^n,\Delta}}\right\|_1\le 17\sqrt{\Delta} \ . \label{eq: thm comb no eps}
\end{align}
The protocol can be realized by a circuit using $5n+\ceil{\log 1/\Delta} + 4$ elementary operations.
\end{restatable}
Since Protocol~\ref{prot: comb state prep} only uses one auxiliary qubit in addition to the bosonic system mode, Theorem~\ref{thm:comb state prep} implies that the unitary state complexity of the comb state with $L=2^n$ peaks is upper bounded by
\begin{align}
    \cC^*_{17\sqrt{\Delta}}(\ket{\Sha_{2^n,\Delta}})\leq 5n+  \ceil{\log 1/\Delta}+4\ .
\end{align}

Before proving~\eqref{eq: thm comb no eps}, let us verify that the stated number of elementary operations (cf. Section~\ref{sec: allowed operations}) is correct.
The circuit in Fig.~\ref{fig: circuit__protocol1}
uses gates from the set~$\cG$ (used in the definition of~$V$) only, except for the single-mode squeezing operator~$S(n\log 2+\log 1/\Delta)$
which is used to prepared the squeezed vacuum state~$\ket{\Psi_{2^{-n}\Delta}}$ from the vacuum state~$\ket{\vac}$ (see Step~\eqref{it:prot:comb:squeezed state} in Protocol~\ref{prot: comb state prep}.
We can decompose this unitary as 
\begin{align}
S(n\log2 + \log1/\Delta)= S(\log(2))^n     S(z_\Delta)^{\ceil{\log1/\Delta}}\textrm{ where } z_\Delta = \frac{\log1/\Delta}{\ceil{|\log1/\Delta}|}\in (0, 1]\
         .\label{eq:squeezingdecomposition}
         \end{align}
         Observe that the \rhs only involves single-mode squeezing operators with squeezing parameters $z\in [-2\pi,2\pi]$ in accordance with the definition of the set~$\cG$, i.e., it is a sequence  of $n+ \ceil{\log 1/\Delta}$
         gates belonging to the gate set~$\cG$. 
         
Since the circuit involves~$n$ applications of~$V$ (each composed of $4$~gates from~$\cG$), one application of~$e^{iP}$, application of the Hadamard gate $H$ as well as the squeezing unitary~$S(n\log2 + \log1/\Delta)$, the total number of gates is $5n+\ceil{\log1/\Delta}+2$. Adding the initialization of $\ket{\vac}$  and $\ket{0}$ implies the claim.

\subsection{Underlying ideas for the comb-state-preparation circuit\label{sec:combstatepreparationexplanation}}
Given a copy of the truncated comb state $\ket{\Sha_{L,\Delta}}$ and a qubit in the state $\ket{+}$, the unitary~$V$ generates an approximate instance of the comb state $\ket{\Sha_{2L,2\Delta}}$, see Fig.~\ref{fig:peakdoublingV} for an illustration. We give a quantitative statement in Lemma~\ref{lem: doubling sha state} below. 
In other words, the unitary~$V$ effectively doubles the number~$L$ of peaks, while also doubling the squeezing parameter $\Delta$.

\begin{figure}
\centering
\includegraphics{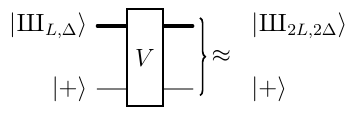}
\caption{
Key to the protocol construction is that the unitary~$V$ essentially doubles the number of peaks when the qubit is in the state~$\ket{+}$. The output represented on the right-hand side is approximate (see Lemma~\ref{lem: doubling sha state} for details). Note that the qubit approximately acts as a catalyst.\label{fig:peakdoublingV}}
\end{figure}

To give some intuition on the repeated action of the unitary~$V$, let us consider its action on a single-peak state of the form~$\ket{x}\otimes\ket{+}$, where~$\ket{x}$ is the (unnormalized) position-eigenstate associated with an integer eigenvalue~$x\in\mathbb{Z}$ of the position operator~$Q$. Applying the unitary $V$ splits the peak in two:
\begin{align}
\begin{array}{rcl}
\ket{x}\otimes \ket{+}&\xrightarrow{\mathmakebox[4.5em]{S(-\log2)\otimes I}}& 
\ket{2x}\otimes \ket{+}\\
&\xrightarrow{\mathmakebox[4.5em]{\mathsf{ctrl}e^{-iP}}}&
\ket{2x}\otimes \ket{0}+\ket{2x+1}\otimes\ket{1}\\
&\xrightarrow{\mathmakebox[4.5em]{I\otimes H}}& \ket{2x}\otimes\ket{+}+\ket{2x+1} \otimes \ket{-}\\
&&=\big(\ket{2x}+\ket{2x+1}\big)\otimes\ket{0}+\big(\ket{2x}-\ket{2x+1}\big)\otimes\ket{1}\\
&\xrightarrow{\mathmakebox[4.5em]{\mathsf{ctrl}e^{i\pi Q}}}&
\big(\ket{2x}+\ket{2x+1}\big)\otimes\ket{+}\ ,
\end{array}
\end{align}
that is,
\begin{align}
V\left(\ket{x}\otimes\ket{+}\right)\propto \big(\ket{2x}+\ket{2x+1}\big)\otimes\ket{+}\qquad\textrm{ for }\qquad x\in\mathbb{Z}\ .
\end{align}
After $n$ iterations of $V$, we thus obtain $2^n$ peaks from the initial single peak, i.e., we have 
\begin{align}
V^n\left(\ket{x}\otimes\ket{+}\right)\propto \big(\ket{2^nx}+\ket{2^nx+1}+\dots+\ket{{2^n}x+2^n-1}\big)\otimes\ket{+}\quad\textrm{ for }\quad x\in\mathbb{Z}\ .\label{eq: integer action V}
\end{align}
This brief calculation illustrates the effect of repeatedly applying~$V$. 

We note that --- according to Eq.~\eqref{eq: integer action V} --- repeated application of~$V$ to $\ket{0}\otimes\ket{+}$, where~$\ket{0}$ is the position-eigenstate with eigenvalue~$x=0$, results in a comb-like state shifted to the right. An approximately centered state can be obtained by  application of an (extensive) phase-shift unitary of the form~$e^{i2^{n-1}P}$. We use an alternative approach,  applying a phase shift~$e^{iP}$ after the first application of~$V$, see Step~\eqref{it:prot comb breeding start} of the protocol.

The detailed analysis of Protocol~\ref{prot: comb state prep} follows similar reasoning and is presented in the following section. The protocol uses a squeezed vacuum state~$\ket{\Psi_{2^{-n}\Delta}}$
in place of the position-eigenstate~$\ket{0}$. 
 Here a comb state results from the  repeated application of the unitary~$V$. The corresponding process is illustrated in Fig.~\ref{fig: comb-state preparation progression}. 
\begin{figure}[!hbtp]
\enlargethispage{\baselineskip}
    \centering
     \begin{subfigure}[b]{0.47\textwidth}
        \centering
        \includegraphics[width=\textwidth]{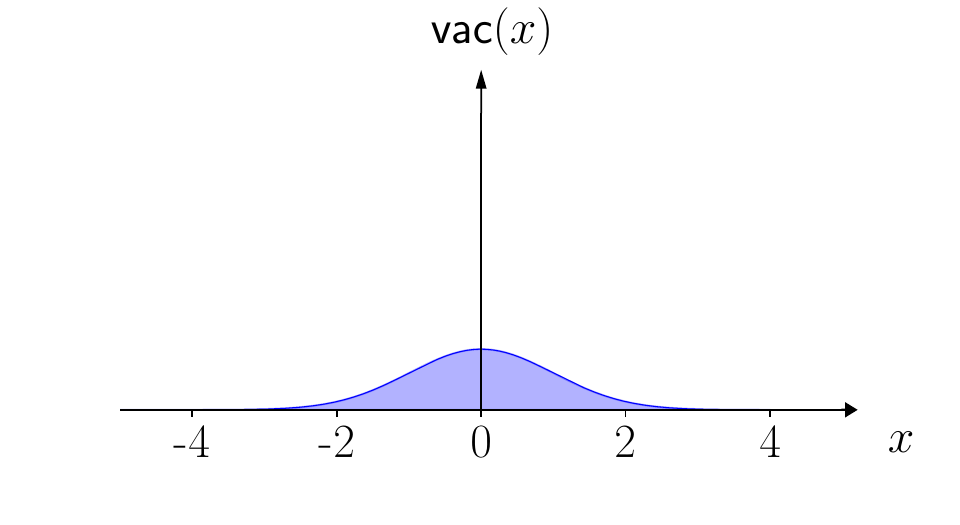}
    \caption{Vacuum state $\ket{\vac}$}
    \end{subfigure}
    \centering
    \begin{subfigure}[b]{0.47\textwidth}
        \centering
        \includegraphics[width= \textwidth]{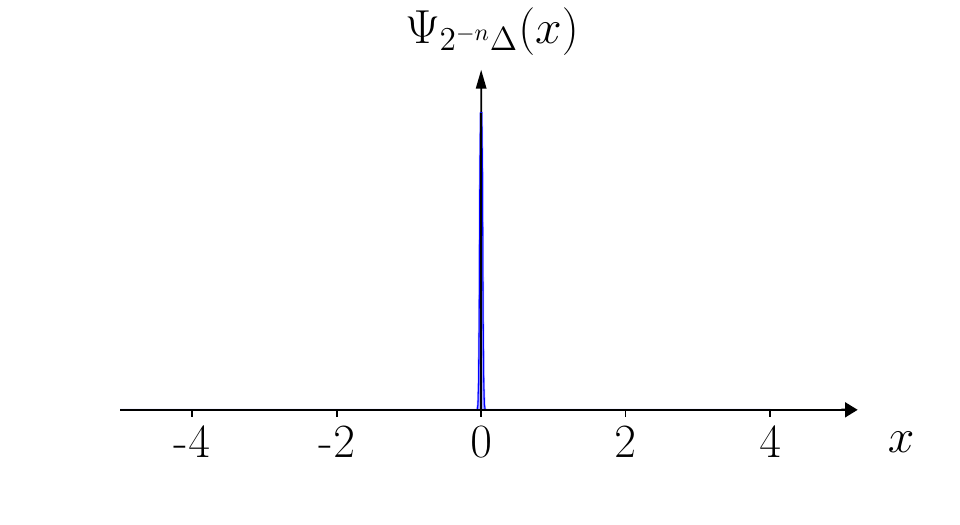}
    \caption{Squeezed state $\ket{\Psi_{2^{-n}\Delta}}$}
    \end{subfigure}\\
     \vspace*{0.5cm}
    \centering
    \begin{subfigure}[b]{0.47\textwidth}
        \centering
        \includegraphics[width=\textwidth]{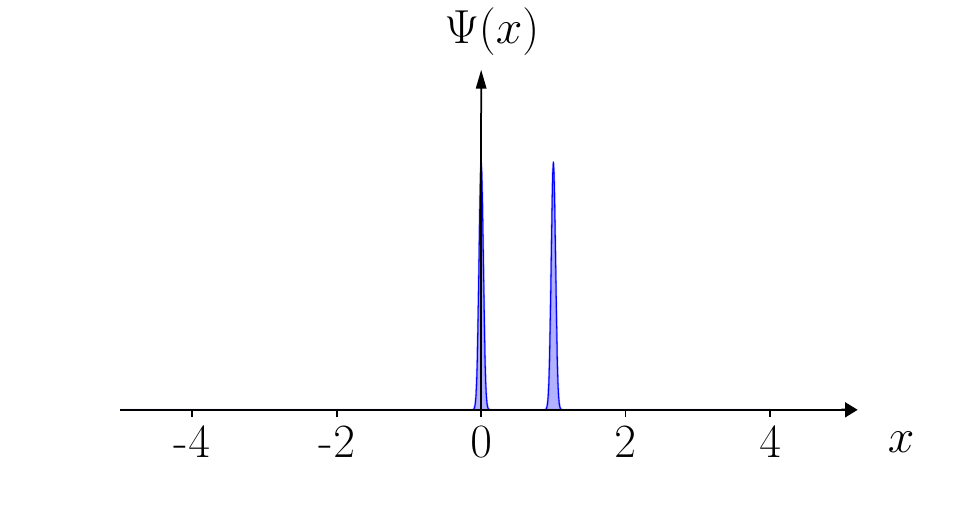}
    \caption{State 
    $\ket{\Psi}\propto\!\ket{\chi_{2^{-n+1}\Delta}(0)}+\ket{\chi_{2^{-n+1}\Delta}(1)}$}
    \end{subfigure}
    \centering
    \begin{subfigure}[b]{0.47\textwidth}
        \centering
        \includegraphics[width=\textwidth]{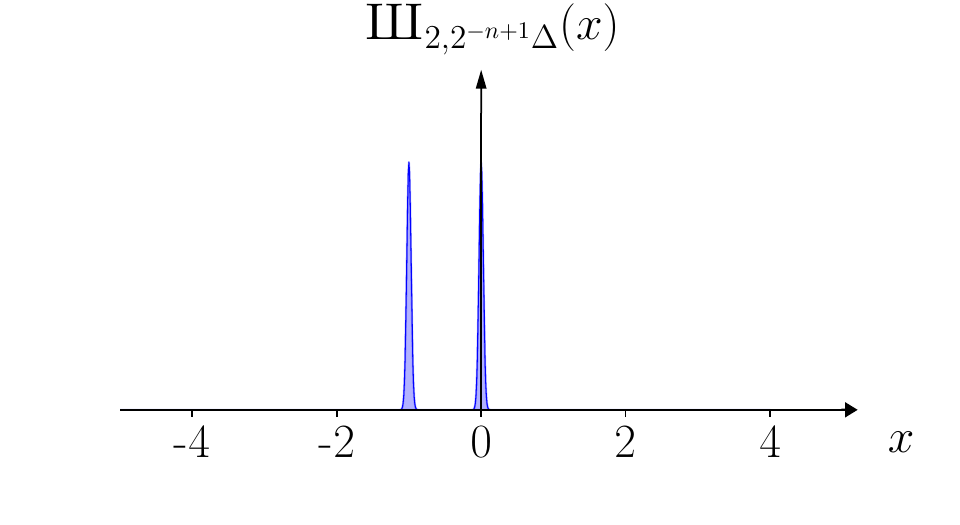}
    \caption{State $\ket{\Sha_{2,2^{-n+1}\Delta}}$\\}
    \end{subfigure}\\
    
     \vspace*{0.3cm}
    \centering
    \begin{subfigure}[b]{0.47\textwidth}
        \centering
        \includegraphics[width= \textwidth]{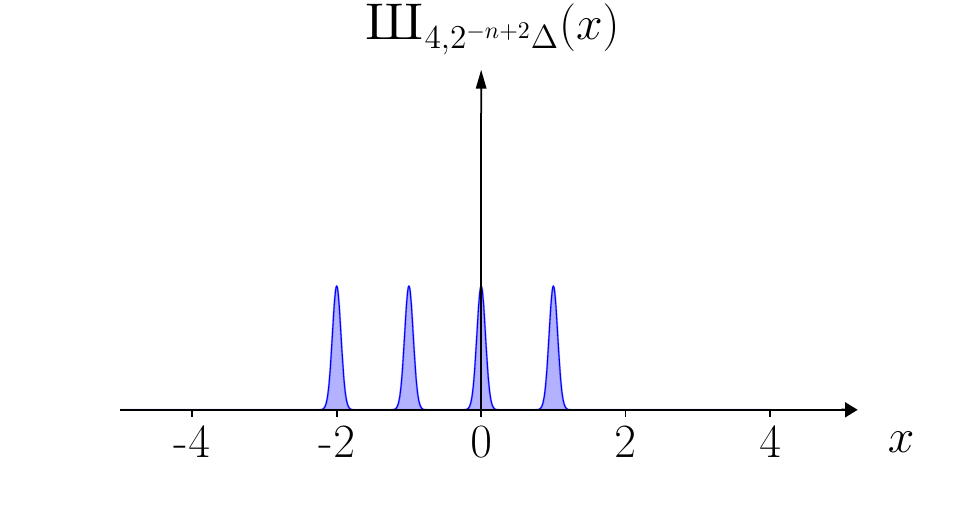}
    \caption{State $\ket{\Sha_{4,2^{-n+2}\Delta}}$}
    
    \end{subfigure}
    \centering
    \begin{subfigure}[b]{0.47\textwidth}
        \centering
        \includegraphics[width=\textwidth]{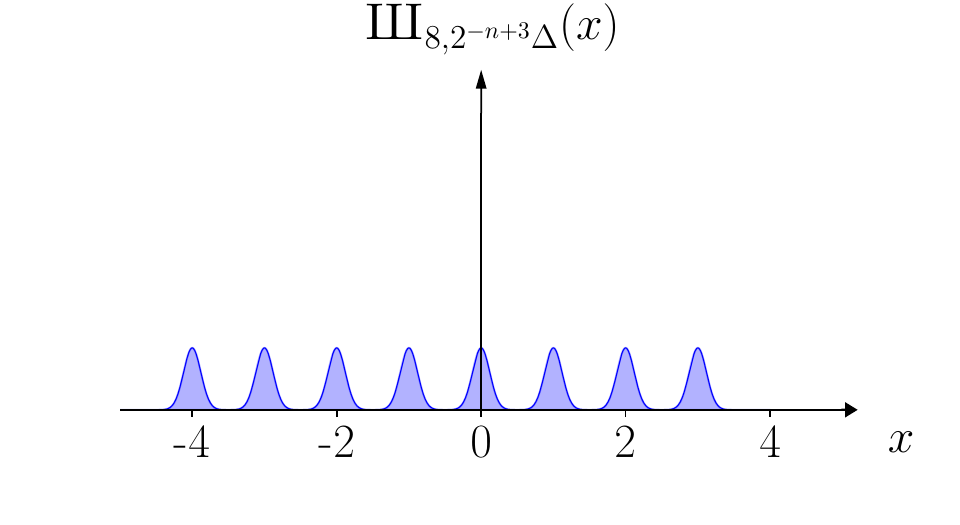}
    \caption{State $\ket{\Sha_{8,2^{-n+3}\Delta}}$}
    \end{subfigure}\\
    \vspace*{0.3cm}
    \centering
    \begin{subfigure}[b]{0.9\textwidth}
        \centering
        \includegraphics[width= \textwidth]{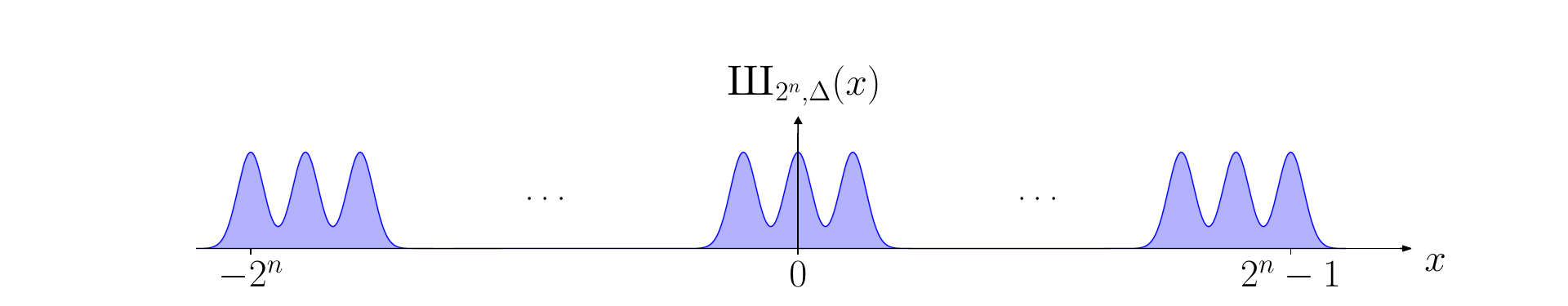}
    \caption{State $\ket{\Sha_{2^n,\Delta}}$ (peaks not in scale)}
    \end{subfigure}
    \caption{Representation of the sequence of states generated in the comb-state-preparation protocol (Protocol~\ref{prot: comb state prep}). \label{fig: comb-state preparation progression}}
\end{figure}

\subsection{Proof of Theorem~\ref{thm:comb state prep}}\label{sec:comb state prep error}

The proof of Theorem~\ref{thm:comb state prep} uses approximate comb states which we introduce in Section~\ref{sec:approximatecombstates}.  In Section~\ref{sec:peakdoublingV}, we analyze a single application of the unitary~$V$. In Section~\ref{sec:proofcompletioncombstate}, we complete the proof of Theorem~\ref{thm:comb state prep}.

\subsubsection{Definition of approximate comb states\label{sec:approximatecombstates}}
Our analysis involves truncated comb states.
A truncated comb state is obtained by taking a comb state $\ket{\Sha_{L,\Delta}}$ and truncating each Gaussian peak to have support only on an interval of width~$2\varepsilon$. Concretely, for $\varepsilon\in (0,1/2)$, we define
\begin{align}
\Sha^{\varepsilon}_{L,\Delta}=\frac{\Pi_{\mathbb{Z}(\varepsilon)}\Sha_{L,\Delta}}{\left\|\Pi_{\mathbb{Z}(\varepsilon)}\Sha_{L,\Delta}\right\|}\ ,\label{eq:Sha L ep}
\end{align}
where $\mathbb{Z}(\varepsilon)=\mathbb{Z}+[-\varepsilon,\varepsilon]$ (where addition is understood as Minkowski sum), and where for a subset~$S\subseteq\mathbb{R}$, we denote by $\Pi_S$  the orthogonal projection  onto functions having support contained in~$S$.

\begin{figure}[!ht]
\centering
       \includegraphics[width=0.45\textwidth]{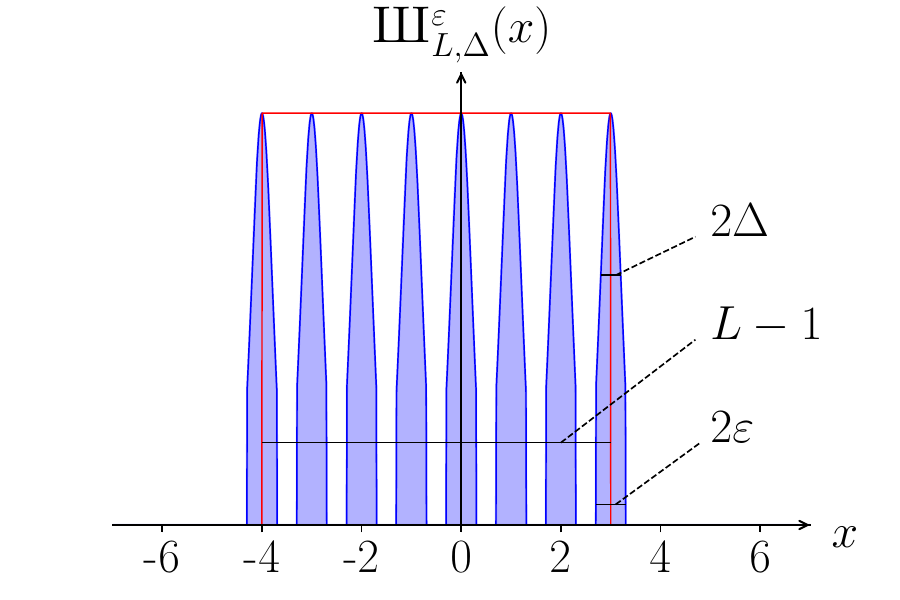}
    \caption{An illustration  of the truncated comb state $|\Sha_{L,\Delta}^\varepsilon\rangle$ with $L=8$ and~$\varepsilon=0.3$. \label{fig:approximatecombstate} }
    \label{fig:comb state ep}
\end{figure}

The state
$\ket{\Sha^{\varepsilon}_{L,\Delta}}$ illustrated in Fig.~\ref{fig:comb state ep}. 
Observe that for an even integer $L\in 2\mathbb{N}$, this state has the form
\begin{align}
    \ket{\Sha^{\varepsilon}_{L,\Delta}}&=\frac{1}{\sqrt{L}}\sum_{z=-L/2}^{L/2-1}\ket{\chi^\varepsilon_\Delta(z)}\ ,\label{eq:Sha L ep sum}
\end{align}
where we used
\begin{align}
(\chi^\varepsilon_\Delta(z))(x)&=\Psi^\varepsilon_\Delta(x-z)
\qquad \textrm{ and }\qquad \Psi^\varepsilon_\Delta=
\frac{\Pi_{[-\varepsilon,\varepsilon]}\Psi_{\Delta}}{\|\Pi_{[-\varepsilon,\varepsilon]}\Psi_{\Delta}\|}\ .
\label{eq:chiDelta ep definition}
\end{align}
Here~$\Psi_\Delta\in L^2(\mathbb{R})$ is the squeezed vacuum state defined in Eq.~\eqref{eq:chiDeltadefinition}.
Clearly, for suitably chosen~$(\varepsilon,\Delta)$, the truncated state~$\ket{\Psi^\varepsilon_\Delta}$ is close to~$\ket{\Psi_\Delta}$, see Lemma~\ref{lem: gaussian gaussian epsilon} in the appendix for a quantitative statement. 

Similarly, the state~$\ket{\Sha_{L,\Delta}^\varepsilon}$ 
is close to the state~$\ket{\Sha_{L,\Delta}}$
for suitably chosen parameters~$(L,\Delta,\varepsilon)$. We refer to Lemma~\ref{lem: Sha - Sha epsilon overlap}  for a quantitative statement.

\subsubsection{Analysis of the peak-doubling unitary~$V$\label{sec:peakdoublingV}}
Our analysis of Protocol~\ref{prot: comb state prep} 
starts with a quantitative description of the effect of a single application of the unitary~$V$ introduced in Eq.~\eqref{eq: V comb prep one round} (and illustrated in Fig.~\ref{fig: circuit V}).

We first consider the application of $V$ to a single (truncated) squeezed vacuum state~$\ket{\Psi^\varepsilon_\Delta}$. 
\begin{lemma}\label{lem: doubling sha state base case}
    Let $\varepsilon\in (0,1/2)$ and $\Delta>0$. Then
    \begin{align}
        \left\| e^{-iP}\proj{\Sha^{2\varepsilon}_{2,2\Delta} }e^{iP} \otimes \proj{+} - V\left(\proj{\Psi^\varepsilon_\Delta}\otimes \proj{+}\right) V^\dagger \right\|_1 \le 9\varepsilon\ .
    \end{align}
\end{lemma}
\begin{proof}
We use the factorization of $V= V^{(4)} V^{(3)} V^{(2)} V^{(1)}$ of the unitary~$V$ from~\eqref{eq: V comb prep one round}. 
Since the squeezing operator~$S(z)$ acts on an element~$\Psi\in L^2(\mathbb{R})$ as $(S(z)\Psi)(x)=e^{z/2}\Psi(e^{z} x)$,
we have 
$S(-\log 2)\ket{\Psi^\varepsilon_\Delta}=\ket{\Psi^{2\varepsilon}_{2\Delta}}$. It follows that
\begin{align}
V^{(1)}(\ket{\Psi^\varepsilon_\Delta}\otimes\ket{+})&=\ket{\Psi^{2\varepsilon}_{2\Delta}}\otimes\ket{+}\ .
\end{align}
Writing~$\ket{\Psi^{2\varepsilon}_{2\Delta}}=\ket{\chi^{2\varepsilon}_{2\Delta}(0)}$ and using that $ e^{-iP}\ket{\Psi^{2\varepsilon}_{2\Delta}}=\ket{\chi^{2\varepsilon}_{2\Delta}(1)}$, we obtain
\begin{align}
V^{(2)}V^{(1)}(\ket{\Psi^\varepsilon_\Delta}\otimes\ket{+})&=
\frac{1}{\sqrt{2}}(\ket{\chi^{2\varepsilon}_{2\Delta}(0)}\otimes\ket{0}+\ket{\chi^{2\varepsilon}_{2\Delta}(1)}\otimes\ket{1})\ .
\end{align}
In particular, applying a Hadamard gate to the qubit results in the state
\begin{align}
V^{(3)}V^{(2)}V^{(1)}(\ket{\Psi^\varepsilon_\Delta}\otimes\ket{+})&=
\frac{1}{\sqrt{2}}(\ket{\chi^{2\varepsilon}_{2\Delta}(0)}\otimes\ket{+}+\ket{\chi^{2\varepsilon}_{2\Delta}(1)}\otimes\ket{-})\\
&=\frac{1}{\sqrt{2}}(\frac{1}{\sqrt{2}}(\ket{\chi^{2\varepsilon}_{2\Delta}(0)}+\ket{\chi^{2\varepsilon}_{2\Delta}(1)})\otimes\ket{0}+\frac{1}{\sqrt{2}}(\ket{\chi^{2\varepsilon}_{2\Delta}(0)}-\ket{\chi^{2\varepsilon}_{2\Delta}(1)})\otimes\ket{1})\ .
\end{align}
The final unitary~$V^{(4)}=\mathsf{ctrl}e^{i\pi Q}$ has the effect of approximately eliminating the phase~$(-1)$ in the second  term. In more detail, the final state is 
\begin{align}
V(\ket{\Psi^\varepsilon_\Delta}\otimes\ket{+})
&=\frac{1}{\sqrt{2}}(\frac{1}{\sqrt{2}}(\ket{\chi^{2\varepsilon}_{2\Delta}(0)}+\ket{\chi^{2\varepsilon}_{2\Delta}(1)})\otimes\ket{0}+\frac{1}{\sqrt{2}}(e^{i\pi Q}\ket{\chi^{2\varepsilon}_{2\Delta}(0)}-e^{i\pi Q}\ket{\chi^{2\varepsilon}_{2\Delta}(1)})\otimes\ket{1})\\
&=
\frac{1}{\sqrt{2}}
\left(\frac{1}{2}(I+e^{i\pi Q})\ket{\chi^{2\varepsilon}_{2\Delta}(0)}+\frac{1}{2}(I-e^{i\pi Q})\ket{\chi^{2\varepsilon}_{2\Delta}(1)}\right)\otimes\ket{+}\\
&\quad +
\frac{1}{\sqrt{2}}
\left(\frac{1}{2}(I-e^{i\pi Q})\ket{\chi^{2\varepsilon}_{2\Delta}(0)}+\frac{1}{2}(I+e^{i\pi Q})\ket{\chi^{2\varepsilon}_{2\Delta}(1)}\right)\otimes\ket{-}\ .\label{eq:targetstateapproxproduced}
\end{align}
Now consider the state
\begin{align}
e^{-iP}\ket{\Sha^{2\varepsilon}_{2,2\Delta}}&=\frac{1}{\sqrt{2}}e^{-iP}\left(\ket{\chi^{2\varepsilon}_{2\Delta}(-1)}+\ket{\chi^{2\varepsilon}_{2\Delta}(0)}\right)\\
&=\frac{1}{\sqrt{2}}\left(\ket{\chi^{2\varepsilon}_{2\Delta}(0)}+\ket{\chi^{2\varepsilon}_{2\Delta}(1)}\right)\ .\label{eq:targetstatenewm}
\end{align}
The overlap of~\eqref{eq:targetstateapproxproduced} and~\eqref{eq:targetstatenewm}
is 
\begin{align}
\Big\langle e^{-iP}\ket{\Sha^{2\varepsilon}_{2,2\Delta}}\otimes \ket{+},V(\ket{\Psi^\varepsilon_\Delta}\otimes\ket{+})\Big\rangle&=\quad\frac{1}{2}
\cdot
\Big\langle \chi^{2\varepsilon}_{2\Delta}(0),\frac{1}{2}(I+e^{i\pi Q})\chi^{2\varepsilon}_{2\Delta}(0)\Big\rangle\\
&\quad+ \frac{1}{2}\cdot
\Big\langle \chi^{2\varepsilon}_{2\Delta}(1),\frac{1}{2}(I-e^{i\pi Q})\chi^{2\varepsilon}_{2\Delta}(1)\Big\rangle
\end{align}
where 
we used that $\ket{\chi^{2\varepsilon}_{2\Delta}(0)}$ and $\ket{\chi^{2\varepsilon}_{2\Delta}(1)}$
(and thus also $e^{i\pi Q}\ket{\chi^{2\varepsilon}_{2\Delta}(0)}$ and $\ket{\chi^{2\varepsilon}_{2\Delta}(1)}$ etc.) have non-overlapping support (because~$\varepsilon<1/2$)  and are thus orthogonal, i.e.,
\begin{align}
\langle\chi^{2\varepsilon}_{2\Delta}(0),\chi^{2\varepsilon}_{2\Delta}(1)\rangle=\langle e^{i\pi Q}\chi^{2\varepsilon}_{2\Delta}(0),\chi^{2\varepsilon}_{2\Delta}(1)\rangle=\langle\chi^{2\varepsilon}_{2\Delta}(0),e^{i\pi Q}\chi^{2\varepsilon}_{2\Delta}(1)\rangle=0\ .
\end{align}
Using that 
\begin{align}
\left\langle \chi_\Delta^\varepsilon(z), e^{-i\pi z} e^{i\pi Q} \chi_\Delta^\varepsilon(z)\right\rangle\ge 1 - 5\varepsilon^2\qquad\textrm{ for every }z\in\mathbb{R}\ ,
\end{align}
see Lemma~\ref{lem: e^iQ vs e^iz} in the appendix, we obtain (for $z\in \{0,1\}$, and for $2\varepsilon$ instead of~$\varepsilon$)
\begin{align}
\langle \chi^{2\varepsilon}_{2\Delta}(z),\frac{1}{2}(I+(-1)^ze^{i\pi Q})\chi^{2\varepsilon}_{2\Delta}(z)\rangle
&\geq 1-10\varepsilon^2\qquad\textrm{ for }\qquad z\in \{0,1\}
\end{align}
and thus
\begin{align}
\Big\langle e^{-iP}\ket{\Sha^{2\varepsilon}_{2,2\Delta}}\otimes \ket{+},V(\ket{\Psi^\varepsilon_\Delta}\otimes\ket{+})\Big\rangle\geq 1-10\varepsilon^2\ .
\end{align}
In particular, 
\begin{align}
\left|\Big\langle e^{-iP}\ket{\Sha^{2\varepsilon}_{2,2\Delta}}\otimes \ket{+},V(\ket{\Psi^\varepsilon_\Delta}\otimes\ket{+})\Big\rangle\right|^2\geq 1-20\varepsilon^2\ 
\end{align}
and the claim follows from the identity
\begin{align}
\left\|\proj{\Psi}-\proj{\Phi}\right\|_1=2\sqrt{1-|\langle \Psi,\Phi\rangle|^2}\ .  \label{eq: trace distance overlap}
\end{align}
relating the trace distance and the overlap for two pure states~$\ket{\Phi},\ket{\Psi}$ and the inequality~$2\sqrt{20}\leq 9$. 

\end{proof}

Given the state $\ket{\Sha^{\varepsilon}_{L,\Delta}}$, the unitary $V$ generates an approximation of the state $\ket{\Sha^{2\varepsilon}_{2L,2\Delta}}$ when applied to a product state with the qubit in the state~$\ket{+}$.
(The qubit approximately acts as a catalyst, i.e., 
the state of the qubit after application of~$V$ is approximately equal to~$\ket{+}$.)  A detailed description of this ``peak-doubling'' effect is the following:

\begin{lemma}\label{lem: doubling sha state}
    Let $\varepsilon\in (0,1/2)$, $\Delta>0$ and $L\in 2\mathbb{N}$.  
    Then
    \begin{align}
        \left\| \proj{\Sha^{2\varepsilon}_{2L,2\Delta} } \otimes \proj{+} - V\left(\proj{\Sha^{\varepsilon}_{L,\Delta}}\otimes \proj{+}\right) V^\dagger \right\|_1 \le 9\varepsilon\ .
    \end{align}
\end{lemma}
\begin{proof}
Let us analyze the action of $V$ on the state $\ket{\Sha^{\varepsilon}_{L,\Delta}} \otimes\ket{+}$.
To lighten the notation, we first define the states
\newcommand{\psieven}{\ket{\Psi^{\rm even}_{z}}}
\newcommand{\psiodd}{\ket{\Psi^{\rm odd}_{z}}}
\newcommand{\psievenz}{\ket{\Psi^{\rm even}_{z'}}}
\newcommand{\psioddz}{\ket{\Psi^{\rm odd}_{z'}}}
\begin{align}
\psieven=\ket{\chi^{2\varepsilon}_{2\Delta}(2z)}\qquad\textrm{ and }\qquad 
\psiodd=\ket{\chi^{2\varepsilon}_{2\Delta}(2z+1)}\qquad\textrm{ for }\qquad z\in\mathbb{Z}\ .
\end{align}
Since~$\varepsilon<1/2$, these states are pairwise orthogonal, i.e., we have
\begin{align}
\renewcommand*{\arraystretch}{1.225}
\begin{matrix}
\langle \Psi^{\rm odd}_{z},\Psi^{\rm even}_{z'}\rangle &=&0\\
\langle \Psi^{\rm even}_{z},\Psi^{\rm even}_{z'}\rangle &=&\langle \Psi^{\rm odd}_{z},\Psi^{\rm odd}_{z'}\rangle =\delta_{z,z'}
\end{matrix}\qquad\textrm{ for }\qquad z,z'\in\mathbb{Z}\ .\label{eq:orthogonalityfirstpsievenodd}
\end{align}
For later reference, we note that --- since these wavefunctions have pairwise orthogonal support (in the position-basis), these orthogonality relations also hold when an additional phase (in the position-basis) is introduced. In particular, we have 
\begin{align}
\renewcommand*{\arraystretch}{1.225}
\begin{matrix}
\langle \Psi^{\rm odd}_{z},e^{i\pi Q}\Psi^{\rm even}_{z'}\rangle &=&0\\
\langle \Psi^{\rm even}_{z},e^{i\pi Q}\Psi^{\rm even}_{z'}\rangle &=&\langle \Psi^{\rm odd}_{z},e^{i\pi Q}\Psi^{\rm odd}_{z'}\rangle=0
\end{matrix}\qquad\textrm{ for }\qquad z\neq z'\in\mathbb{Z}\ .\label{eq:orthogonalitysecondpsievenodd}
\end{align}
With the factorization $V= V^{(4)} V^{(3)} V^{(2)} V^{(1)}$ of the unitary~$V$ from~\eqref{eq: V comb prep one round}, we can analyze the action of~$V$ as follows.  It is easy to check that 
\begin{align}
     V^{(1)} \left(\ket{\Sha^{\varepsilon}_{L,\Delta}} \otimes \ket{+}\right) = \frac{1}{\sqrt{L}}\sum_{z=-L/2}^{L/2-1}\psieven  \ket{+}\ .
\end{align}
Hence, we obtain   
\begin{align}
    V^{(2)} V^{(1)}\left(\ket{\Sha^{\varepsilon}_{L,\Delta}} \otimes\ket{+} \right) = \frac{1}{\sqrt{L}}\sum_{z=-L/2}^{L/2-1}\frac{1}{\sqrt{2}}\left(\psieven  \otimes\ket{0} + \psiodd \otimes \ket{1}\right)\, 
\end{align}
using  $\mathsf{ctrl}e^{-iP}\psievenz=\psioddz$.
 Applying the single-qubit Hadamard gate~$V^{(3)}=I\otimes H$ to this state yields
\begin{align}
     V^{(3)} V^{(2)} V^{(1)}\left( \ket{\Sha^{\varepsilon}_{L,\Delta}} \otimes\ket{+}\right)  &= \frac{1}{\sqrt{L}}\sum_{z=-L/2}^{L/2-1}\frac{1}{\sqrt{2}}\left(\psieven  \otimes\ket{+} + \psiodd \otimes \ket{-}\right)\\
     &= \frac{1}{2\sqrt{L}}\sum_{z=-L/2}^{L/2-1}\Big(\  \left(  
    \psieven + \psiodd
     \right)\otimes \ket{0}\\
     &\qquad\qquad\quad\qquad +\left(  
    \psieven - \psiodd
     \right)\otimes \ket{1} \Big)\ .
\end{align}
The final state after application of~$V^{(4)}=\mathsf{ctrl}e^{i\pi Q}$ is thus
\begin{align}
    V \left(\ket{\Sha^{\varepsilon}_{L,\Delta}}\otimes \ket{+}\right) &=        \frac{1}{2\sqrt{L}}\sum_{z=-L/2}^{L/2-1}\Big(\left(  
    \psieven + \psiodd
     \right)\otimes \ket{0} +e^{i\pi Q}\left(  
    \psieven - \psiodd
     \right)\otimes \ket{1}\Big)\\[-0.5cm]
     &=  \frac{1}{2\sqrt{2L}}\sum_{z=-L/2}^{L/2-1}\Big(\left(  
     (I+e^{i\pi Q})\psieven + (I-e^{i\pi Q})\psiodd  
     \right)\otimes \ket{+} \\
     &\hspace{6.65em}+\left(  
     (I-e^{i\pi Q})\psieven + (I+e^{i\pi Q})\psiodd  
     \right)\otimes \ket{-} \Big)\ .
\end{align}
\renewcommand{\psieven}{{\Psi^{\rm even}_{z}}}
\renewcommand{\psiodd}{{\Psi^{\rm odd}_{z}}}
\renewcommand{\psievenz}{{\Psi^{\rm even}_{z'}}}
\renewcommand{\psioddz}{{\Psi^{\rm odd}_{z'}}}

We compute the overlap between the target state $\ket{\Sha^{2\varepsilon}_{2L,2\Delta}} \otimes \ket{+}$ and the state $V\left(\ket{\Sha^{\varepsilon}_{L,\Delta}}\otimes \ket{+}\right)$ prepared by the protocol.
For convenience, we rewrite the target state in a form that resembles the form of the latter. We have
\begin{align}
    &\ket{\Sha^{2\varepsilon}_{2L,2\Delta}} = \frac{1}{\sqrt{2L}}\sum_{z = -L}^{L-1} \ket{\chi^{2\varepsilon}_{2\Delta}(z)} = \frac{1}{\sqrt{2L}} \sum_{z = -L/2}^{L/2-1} \left(\ket{\psieven} + \ket{\psiodd} \right)\ . 
\end{align}
Therefore, the overlap between $\ket{\Sha^{2\varepsilon}_{2L,2\Delta}} \otimes\ket{+}$ and $V\left(\ket{\Sha^{\varepsilon}_{L,\Delta}} \otimes\ket{+}\right)$ is 
\begin{align}
    &\big(\langle \Sha^{2\varepsilon}_{2L,2\Delta}|\otimes \langle+|\big) V\big( |\Sha^{\varepsilon}_{L,\Delta}\rangle \otimes \ket{+}\big)\\
    &=\frac{1}{4L} \sum_{z = -L/2}^{L/2-1} \sum_{z' = -L/2}^{L/2-1} \left\langle \psievenz + \psioddz\ , (I+e^{i\pi Q})\psieven + (1-e^{I\pi Q})\psiodd \right\rangle \label{eq:comb double sum}\\
    &=\frac{1}{4L} \sum_{z = -L/2}^{L/2-1} \Big(2 + \left\langle \psieven, e^{i\pi Q} \psieven    \right\rangle  -   \left\langle \psiodd\ , e^{i\pi Q} \psiodd    \right\rangle   \Big)\ 
\end{align}
where we used the orthogonality relations~\eqref{eq:orthogonalityfirstpsievenodd} and~\eqref{eq:orthogonalitysecondpsievenodd}. 
Using that $e^{-i\pi (2z)}=1$ and  $e^{-i\pi (2z+1)}=-1$ for every integer~$z\in\mathbb{Z}$, this can be rewritten as 
\begin{align}
    &\big(\langle \Sha^{2\varepsilon}_{2L,2\Delta}|\otimes \langle+|\big) V\big( |\Sha^{\varepsilon}_{L,\Delta}\rangle \otimes \ket{+}\big)\\
    &=\frac{1}{4L} \sum_{z = -L/2}^{L/2-1} \Big(2 + \left\langle \psieven\ , e^{-i\pi(2z)} e^{i\pi Q} \psieven    \right\rangle  +   \left\langle \psiodd , e^{-i\pi (2z+1)} e^{i\pi Q} \psiodd    \right\rangle   \Big) \\
    & \ge\frac{1}{4L} \sum_{z= -L/2}^{L/2-1} \Big(2 + 2(1 - 5(2\varepsilon)^2) \Big)\\
    &= 1 - 10\varepsilon^2\ . \label{eq: V overlap to target}
\end{align}
Here, we used the fact (see Lemma~\ref{lem: e^iQ vs e^iz}) that $\left\langle \chi_\Delta^\varepsilon(z), e^{-i\pi z} e^{i\pi Q} \chi_\Delta^\varepsilon(z)\right\rangle\ge 1 - 5\varepsilon^2$  to obtain the last inequality. We prove this fact in appendix. Since Eq.~\eqref{eq: V overlap to target} implies
\begin{align}
    \abs{\left(\langle \Sha^{2\varepsilon}_{2L,2\Delta}|\otimes\langle+|\right)V(|\Sha^{\varepsilon}_{L,\Delta}\rangle\otimes\ket{+})}^2\ge1-20\varepsilon^2\ , 
\end{align}
By using the relation between the overlap of two states and their trace distance (cf.\ Eq.~\eqref{eq: trace distance overlap}), we conclude that
\begin{align}
    \left\| \proj{\Sha^{2\varepsilon}_{2L,2\Delta} } \otimes \proj{+} - V\left(\proj{\Sha^{\varepsilon}_{L,\Delta}}\otimes \proj{+}\right) V^\dagger \right\|_1 \le 2\sqrt{20}\varepsilon\ .
\end{align}
This implies the claim since~$2\sqrt{20}\leq 9$.
\end{proof}

\subsubsection{Completing the proof of Theorem~\ref{thm:comb state prep}\label{sec:proofcompletioncombstate}}
In this section, we complete the proof of Theorem~\ref{thm:comb state prep}.  We have already argued that 
Protocol~\ref{prot: comb state prep} can be realized using at most~$5n+\ceil{\log 1/\Delta} + 4$ allowed elementary operations, see the discussion following the statement of the Theorem. It thus remains to show Eq.~\eqref{eq: thm comb no eps}, i.e., that the output state~$\rho$ of the protocol is close to the state~$\ket{\Sha_{2^n,\Delta}}$.

To do so, let us consider the repeated action of~$V$.
The following combines  Lemma~\ref{lem: doubling sha state base case} and Lemma~\ref{lem: doubling sha state}.
\begin{lemma}\label{lem:iteratedlemmaV}
For $n\in\mathbb{N}$, let us define the unitary
\begin{align}
U_n:&=V^{n-1}(e^{iP}\otimes I)V\ .
\end{align}
Suppose $\varepsilon\in (0,2^{-(n+1)})$ and $\Delta>0$. Then, 
\begin{align}
    \left\| 
\proj{\Sha^{2^n\varepsilon}_{2^n,2^n\Delta}}\otimes\proj{+}-    U_n\left(\proj{\Psi^\varepsilon_\Delta}\otimes\proj{+}\right)U_n^\dagger\right\|_1 &\leq 9\varepsilon\cdot (2^{n}-1)\ .
\end{align}
\end{lemma}
\begin{proof}
Define
\begin{align}    \ket{\Phi^{(0)}}&=\ket{\Psi^\varepsilon_\Delta}\otimes\ket{+}\\
    \ket{\Phi^{(1)}}&=(e^{iP}\otimes I) V\ket{\Phi^{0}}\\
    \ket{\Phi^{(k)}}&=V\ket{\Phi^{(k-1)}}\qquad\textrm{ for  }k\in \{2,\ldots,n\}\ .    
\end{align}
We show inductively that 
\begin{align}
\left\|
\proj{\Phi^{(k)}}-\proj{\Sha^{2^k\varepsilon}_{2^k,2^k\Delta}}\otimes\proj{+}
\right\|_1&\leq 9\varepsilon\cdot (2^{k}-1)\ \textrm{ for every } k\in\{1,\ldots,n\}\ .\label{eq:inducationclaimone}
\end{align}
By Lemma~\ref{lem: doubling sha state base case} and the invariance of the norm~$\|\cdot\|_1$ under unitaries we have
\begin{align}
    \left\|\proj{\Phi^{(1)}}-\proj{\Sha_{2,2\Delta}^{2\varepsilon}}\otimes\proj{+}
    \right\|_1\leq 9\varepsilon\ .
\end{align}
This establishes~\eqref{eq:inducationclaimone} for $k=1$.

Suppose that we have shown the claim~\eqref{eq:inducationclaimone} for $k-1$.  Then
 we have by definition and by the triangle inequality that
\begin{align}
&&\left\|
\proj{\Phi^{(k)}}-\proj{\Sha^{2^k\varepsilon}_{2^k,2^k\Delta}}\otimes\proj{+}
\right\|_1 \\
 &&  =\left\|
V\proj{\Phi^{(k-1)}}V^\dagger-\proj{\Sha^{2^k\varepsilon}_{2^k,2^k\Delta}}\otimes\proj{+}
\right\|_1\\
&&\leq \left\| V\proj{\Phi^{(k-1)}}V^\dagger-V\left(\proj{\Sha^{2^{k-1}\varepsilon}_{2^{k-1},2^{k-1}\Delta}}\otimes\proj{+}\right)V^\dagger\right\|_1\\
&&+\left\|V\left(\proj{\Sha^{2^{k-1}\varepsilon}_{2^{k-1},2^{k-1}\Delta}}\otimes\proj{+}\right)V^\dagger-\proj{\Sha^{2^k\varepsilon}_{2^k,2^k\Delta}}\otimes\proj{+}
\right\|_1\ . \hspace{-0.65em}
\end{align}
By the invariance of the norm under unitaries and by the induction hypothesis, we have
\begin{align}
    \left\| V\proj{\Phi^{(k-1)}}V^\dagger-V\left(\proj{\Sha^{2^{k-1}\varepsilon}_{2^{k-1},2^{k-1}\Delta}}\right)V^\dagger\right\|_1&\leq
    9\varepsilon\cdot  (2^{k-1}-1)\ .
\end{align}
Furthermore, we have 
\begin{align}
\left\|V\left(\proj{\Sha^{2^{k-1}\varepsilon}_{2^{k-1},2^{k-1}\Delta}}\otimes\proj{+}\right)V^\dagger-\proj{\Sha^{2^k\varepsilon}_{2^k,2^k\Delta}}\otimes\proj{+}
\right\|_1\leq 9\cdot 2^{k-1}\varepsilon
\end{align}
by Lemma~\ref{lem: doubling sha state}.
The latter can be applied since $2^k\varepsilon\in (0,1/2)$ by the assumption~$\varepsilon\in (0,2^{-(n+1)})$. Since $2^{k-1}-1+2^{k-1}=2^k-1$, this implies Eq.~\eqref{eq:inducationclaimone} for $k$.

Because  $\ket{\Phi^{(n)}}=U_n\left(\ket{\Psi^\varepsilon_\Delta}\otimes\ket{+}\right)$ by definition, Eq.~\eqref{eq:inducationclaimone} with $k=n$ implies the claim.
\end{proof}

With Lemma~\ref{lem:iteratedlemmaV}, we can complete the proof of Theorem~\ref{thm:comb state prep} as follows. Let $\Delta\in (0,1/4)$ and $\varepsilon\in (0,1/2)$. Then, Lemma~\ref{lem:iteratedlemmaV}
(with $(2^{-n}\varepsilon,2^{-n}\Delta)$ in place of~$(\varepsilon,\Delta)$) implies that
\begin{align}
\left\|
\proj{\Sha^\varepsilon_{2^n,\Delta}}\otimes\proj{+}-
U_n\left(\proj{\Psi^{2^{-n}\varepsilon}_{2^{-n}\Delta}}\otimes\proj{+}\right)U_n^\dagger
\right\|_1&\leq 9(2^{-n}\varepsilon)\cdot (2^n-1)\\
&\leq 9\varepsilon\  .    \label{eq:upperboundunpsimv}
\end{align}
By Corollary~\ref{cor: gaussian - gaussian epsilon trace distance}, we have that for any $\varepsilon\in(\sqrt{\Delta},1/2)$ the truncated squeezed vacuum state~$\ket{\Psi^{2^{-n}\varepsilon}_{2^{-n}\Delta}}$ is close to the squeezed vacuum state $\ket{\Psi_{2^{-n}\Delta}}$, i.e.,
\begin{align}
    \left\| \proj{\Psi_{2^{-n}\Delta}} - \proj{\Psi_{2^{-n}\Delta}^{2^{-n}\varepsilon}}\right\|_1
    \le 
    3 \sqrt{\Delta} 
    \ . \label{eq:comb state thm gaussian ep} 
\end{align}
Combining~\eqref{eq:upperboundunpsimv} and~\eqref{eq:comb state thm gaussian ep} with the triangle inequality and using the invariance of the norm under unitaries, we conclude that 
\begin{align}
\left\|
\proj{\Sha^\varepsilon_{2^n,\Delta}}\otimes\proj{+}-
U_n(\proj{\Psi_{2^{-n}\Delta}}\otimes\proj{+})U_n^\dagger
\right\|_1&\le 3 \sqrt{\Delta}+9\varepsilon\ .
    \end{align}
Corollary~\ref{cor: Sha - Sha epsilon trace distance} 
states that for $\Delta\in (0,1/4)$ and any $\varepsilon\in[\sqrt{\Delta},1/2)$, we have
\begin{align}
\left\|\proj{\Sha_{2^n,\Delta}}-\proj{\Sha^\varepsilon_{2^n,\Delta}}\right|_1&\leq 5\sqrt{\Delta}\ ,
    \end{align}
    hence we obtain with the choice~$\varepsilon=\sqrt{\Delta}$
    \begin{align}
\left\|
\proj{\Sha_{2^n,\Delta}}\otimes\proj{+}-
U_n(\proj{\Psi_{2^{-n}\Delta}}\otimes\proj{+})U_n^\dagger
\right\|_1&\le 
3\sqrt{\Delta} + 9\varepsilon+5\sqrt{\Delta}\\
&=17\sqrt{\Delta}
\end{align}
by the triangle inequality. Using the fact that  $L^1$-norm is contractive under CPTP maps (in particular, under tracing out the qubit system),
the claim follows, since the output state of the protocol is
    \begin{align}        \rho=\tr_{\mathsf{qubit}}U_n(\proj{\Psi_{2^{-n}\Delta}}\otimes\proj{+})U_n^\dagger
    \end{align}
    by definition. 

\section{The envelope-Gaussification protocol} \label{sec: Gaussian envelope shaping}
In the following we
explain how to turn a  comb state (a state with a rectangular envelope) into a state with a Gaussian envelope. In Section~\ref{sec: env models}, we introduce an alternative notion of 
approximate GKP states (where the envelope is defined differently). This proves helpful for our analysis. In Section~\ref{sec: adaptive Gaussian envelope shaping}, we introduce our heralded envelope-Gaussification protocol, and establish its main properties. Given an input comb state~$\ket{\Sha_{L,\Delta}}$ and a parameter $\kappa$ (specifying the width of the desired Gaussian envelope),
the protocol 
either rejects or accepts. Conditioned on acceptane, the output state is a quantum state close to the approximate GKP state $\ket{\gkp_{\kappa,\Delta}}$. Finally, in Section~\ref{sec:envelope Gaussification error}, we present the proof to the main result of this section: We show that applying the Gaussification protocol to a comb state produces a state with a Gaussian envelope.

\subsection{An alternative type of approximate GKP state}\label{sec: env models}
The protocol considered in this section (Protocol~\ref{prot:envelopeshaping}) takes a comb state and applies a Gaussian envelope to it. 
The protocol does not produce approximate GKP state with ``peak-wise'' Gaussian envelope $\ket{\gkp_{\kappa,\Delta}}$ but (states close to) approximate GKP states with ``point-wise'' Gaussian envelope $\ket{\tGKP_{\kappa,\Delta}}$ that we define here.  

To define the state  $\ket{\tGKP_{\kappa,\Delta}}$
and to highlight the difference to the 
the ``peak-wise'' GKP state~$\ket{\gkp_{\kappa,\Delta}}$ (cf.\ Eq.~\eqref{eq:gkpkappdeltageneraldefinition}), let us rewrite the latter  wavefunction as
\begin{align}
    \gkp_{\kappa, \Delta}(x)&:= C_{\kappa,\Delta} \sum_{z\in\mathbb{Z}}\eta_\kappa(z)\chi_\Delta(z)(x) \ , \label{eq:approximategkpstate eta chi}
\end{align}
where $\eta_\kappa\in L^2(\mathbb{R})$ is the Gaussian envelope  
\begin{align}
    \eta_\kappa(z)&=\frac{\sqrt{\kappa}}{\pi^{1/4}} e^{-\kappa^2z^2/2}\,\label{eq:etakappadefinition}
\end{align}
with parameter $\kappa>0$, 
and $\chi_\Delta(z)\in L^2(\mathbb{R})$ is 
a  Gaussian with variance~$\Delta^2$ centered at~$z\in\mathbb{R}$, see Eq.~\eqref{eq:chiDeltadefinition}. 

We define a ``point-wise'' GKP state $\ket{\tGKP_{\kappa,\Delta}}\in L^2(\mathbb{R})$ by
\begin{align}
	\tGKP_{\kappa,\Delta}(x):=D_{\kappa,\Delta}
	\sum_{z\in\mathbb{Z}} \eta_\kappa(x)\chi_\Delta(z)(x)\ ,\label{eq:pointwisegkp}
\end{align}
where $D_{\kappa,\Delta}$ is normalization factor. We illustrate the difference between the
states~$\ket{\tGKP_{\kappa,\Delta}}$ and $\ket{\gkp_{\kappa,\Delta}}$ in Fig.~\ref{fig: pointwise peakwise approx GKP}.

\begin{figure}[!htb]
	\centering
		\includegraphics[width= 0.75\textwidth]{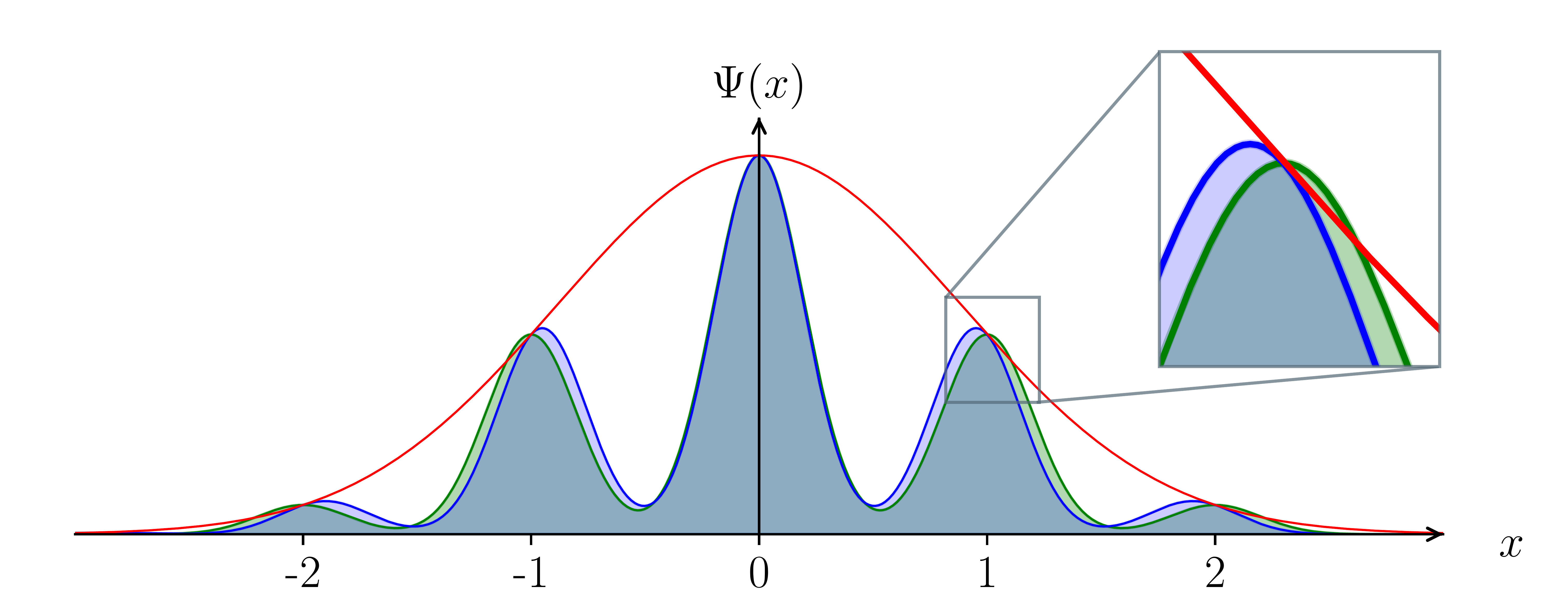}
	\caption{\label{fig: pointwise peakwise approx GKP} Comparison of the ``point-wise'' and ``peak-wise'' envelope models. The ``point-wise'' GKP state $\ket{\tGKP_{\kappa,\Delta}}$ is depicted in blue and the ``peak-wise'' GKP state $\ket{\gkp_{\kappa,\Delta}}$ is depicted in green.}
\end{figure}
We note that ``peak-wise'' and ``point-wise'' approximate GKP states (with an appropriate choice of parameters~$(\kappa,\Delta)$) are close to each other in $L^1$-distance (see~Corollary~\ref{cor: approximate tgkpL GKP trace dist} in the appendix for further details).

\subsection{Envelope shaping
by an adaptive ``measure-then-correct'' protocol
}\label{sec: adaptive Gaussian envelope shaping}
Here, we describe our envelope-Gaussification protocol (cf.\ Protocol~\ref{prot:envelopeshaping}). The protocol takes as input (a state close to) a comb state~$\ket{\Sha_{L,\Delta}}$, as well as a parameter $\kappa>0$  specifying  the targeted Gaussian envelope~$\eta_\kappa$.
 It either accepts or rejects, and outputs a one-mode state when it accepts. We will show that the acceptance probability is lower bounded by a constant. Furthermore, we will prove that the (average) output state conditioned on acceptance is close to the state~$\ket{\gkp_{\kappa,\Delta}}$.

Protocol~\ref{prot:envelopeshaping} is implemented by the adaptive circuit in Fig.~\ref{fig:envelope-protocol}. This circuit is adaptive in the sense that it involves a unitary (displacement) that is classically controlled by (a function of) the measurement result. 

\begin{algorithm}[H]
	\caption{Envelope-Gaussification protocol}
	\label{prot:envelopeshaping} 
	\begin{flushleft}
		\textbf{Input:} 
		A state $\rho \in \cB(L^2(\mathbb{R}))$, a parameter $\kappa\in (0,1/4)$ and a parameter $L\in 8\mathbb{N}$.\\%
		\textbf{Output:} Either \textsf{accept} or \textsf{reject}, and in the case of acceptance a state of a single mode. Conditioned on acceptance, this output state is close to $\ket{\gkp_{\kappa,\Delta}}$, see Theorem~\ref{thm:envelope gaussification}.   
		\begin{algorithmic}[1]
            \State Prepare the squeezed vacuum state $\ket{\eta_\kappa}= S(\log \kappa)\ket{\vac}$ in the first register.
            \State  Apply the unitary $e^{-iP_1Q_2}$.
            \State Perform a homodyne position measurement on the first mode, resulting in an outcome~$x\in\mathbb{R}$
            and a post-measurement state 
            of the second mode.
            \If{$x\in \Omega_L= [-L/8-1/2, L/8+1/2]$}\label{prot: gaussification acceptance}
            \State Round the result $x$ to the nearest integer, yielding $\round{x}\in\mathbb{Z}$.
            \State Apply the classically controlled correction unitary $e^{i\round{x}P}$ on the second mode.

                \State\Return $\mathsf{accept}$ and the state of the second mode.
                \Else
                  \State \Return $\mathsf{reject}$.
                \EndIf
		\end{algorithmic}
	\end{flushleft}
\end{algorithm}

\begin{figure}[!ht]
	\begin{center}
 \includegraphics{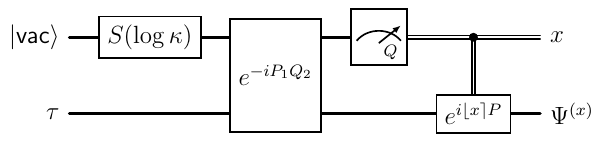}
	\end{center}
 	\caption{Circuit used in the envelope-Gaussification protocol (Protocol~\ref{prot:envelopeshaping}). For measurement outcomes~$x\in \mathbb{R}$ with $x\not\in [-L/8-1/2,L/8+1/2]$, the protocol returns $\mathsf{reject}$, and there is no output state in this case. For values $x\in [-L/8-1/2,L/8+1/2]$ (a case illustrated in the figure), 
   the unitary $\e^{i\round{x}P}$ is applied to the second mode.  This is a classically controlled displacement gate, i.e., it involves the parameter $\round{x}\in\mathbb{Z}$ which is computed classically. Note that both the squeezing gate $S(\log\kappa)$ and $\e^{i\round{x}P}$ need to be decomposed in terms of constant-strength squeezing and displacements gates respectively to obtain operations from the set~$\cG$ (see Section~\ref{sec: allowed operations}). 
	}\label{fig:envelope-protocol}
\end{figure}
The main result of this section is the following.

\begin{restatable}{theorem}{thmenvelopeshaping}
\label{thm:envelope gaussification}
There are constants $b_1, b_2>0$ such that the following holds.
Assume $\xi >0$, $\kappa\in(0,1/4)$, $\Delta\in(0,1/4)$ and $L \in 8\mathbb{N}$.
Let $\tau\in\cB(L^2(\mathbb{R}))$ be a state close to $\ket{\Sha_{L,\Delta}}$, i.e.,
\begin{align}
	\left\|\tau - \proj{\Sha_{L,\Delta}}\right\|_1\le \xi\ .
\end{align}
Given the comb state parameter $L$,  the squeezing parameter $\kappa$ (specifying a Gaussian envelope)
and the input state~$\tau$,  Protocol~\ref{prot:envelopeshaping}
accepts with probability at least
\begin{align}
    \Pr\Big[\textnormal{Protocol~\ref{prot:envelopeshaping} } \accepts \mid \tau
    \Big]\ge \frac18\left(1 - 2e^{-\kappa^2L^2/256}\right)-\frac52\sqrt{\Delta}-\frac\xi 2\ .\label{eq:acceptanceprobabilityclaims}
\end{align}
Conditioned on acceptance, the output state $\tau_\acc\in\cB(L^2(\mathbb{R}))$ on the second mode is close to the state $\ket{\gkp_{\kappa,\Delta}}$, i.e., 
\begin{align}
    \|\tau_\acc-\proj{\gkp_{\kappa,\Delta}}\|_1\le \frac{5\sqrt{\Delta}+\xi}{\frac14(1 - 2e^{-\kappa^2L^2/16})} +6\sqrt{\Delta} +6\kappa\sqrt{L} + 7e^{-\kappa^2 L^2/128} \ .\label{eq:closenessclaimm}
\end{align}
The protocol can be realized using at most $b_1 \log L + b_2 \log1/\kappa$ elementary operations.
\end{restatable}
Before proving Eq.~\eqref{eq:acceptanceprobabilityclaims} and~\eqref{eq:closenessclaimm},
let us compute the number of elementary operations that are needed to implement Protocol~\ref{prot:envelopeshaping}. 

The protocol first prepares the Gaussian state $\ket{\eta_\kappa}=S(\log \kappa)\ket{\vac}$.  As we only allow for bounded strength operations, we will decompose the unitary~$S(\log\kappa)$ into consecutive bounded strength squeezing operators $S(z)$ with $z \in [-2\pi, 2\pi]$. An analysis similar to~\eqref{eq:squeezingdecomposition} shows that the unitary~$S(\log \kappa)$ can be realized by $\ceil{\log1/\kappa}$ gates of this form.

Subsequently, the unitary $e^{-iP_1Q_2}$ is applied. 
This is a Gaussian unitary of constant strength, hence it can be written as a product of a constant number of Gaussian unitaries from the set~$\cG$, (see the discussion in Section~\ref{sec: allowed operations}).

Then, the protocol performs a homodyne measurement of the first register that results in an outcome $x\in\mathbb{R}$. The classical outcome~$x$ is rounded to the next integer~$\round{x}\in\mathbb{Z}$. Then, the protocol applies a classically controlled shift (displacement) unitary $e^{i\round{x}P_2}$ depending on $\round{x}$. This operation again does not have bounded strength in general (we have $|x|\leq L/8+1/2$, meaning that $\round{x}$ can scale with~$L$), and needs to be decomposed into gates from $\cG$. From Lemma~\eqref{lem:upperboundcoherentstatecomplexity}, we know that we can decompose the unitary $e^{i\round{x}P_2}$ using at most $2\ceil{|\log\round{x}|} + 3$ unitaries from $\cG$. As 
the acceptance region is $\Omega_L = [-L/8-1/2, L/8+1/2]$, we conclude that the shift correction after acceptance needs at most $ 2\ceil{\log (L/8 +1/2)} + 3$ gates from $\cG$ to be implemented. As the state initialization of the vacuum $\ket{\vac}$ and the homodyne measurement requires a constant amount (two) elementary operations, choosing $b_1, b_2>0$ large enough shows the claim.

\subsection{Analysis of envelope-Gaussification with~$\ket{\Sha^\varepsilon_L}$ as input\label{sec:shastateinputanalysis}}
In this section, we analyze Protocol~\ref{prot:envelopeshaping}
in the case where the input state is~$\ket{\Sha^\varepsilon_L}$, i.e., a
 truncated comb state.

Specifically, we proceed as follows: In Section~\ref{sec: non adaptive prot}, we translate the circuit defined by Protocol~\ref{prot:envelopeshaping} to an equivalent circuit that 
is non-adaptive (to simplify the analysis). 
In Section~\ref{sec: implication trunc comb states}, we then show that upon acceptance, the output state is indeed a state close the desired approximate GKP state. In Section~\ref{sec: bound acc prob prot 2}, we show that the acceptance probability is lower bounded by a constant independent of the envelope parameters.

These results are subsequently used in Section~\ref{sec:envelope Gaussification error} to extend the analysis to input states that are close to a comb state.

\subsubsection{A non-adaptive  description of Protocol~\ref{prot:envelopeshaping} \label{sec: non adaptive prot}}
We note that Protocol~\ref{prot:envelopeshaping} (cf. Fig.~\ref{fig:envelope-protocol}) is adaptive, i.e., it involves a unitary  (the unitary $e^{i\round{x}R}$) whose parameter is classically controlled and determined by the measurement result~$x\in\mathbb{R}$. To analyse 
Protocol~\ref{prot:envelopeshaping}, it will be convenient to consider a non-adaptive version of the circuit depicted in Fig.~\ref{fig:envelope-protocol}. The non-adaptive circuit contains a unitary that is outside our allowed gate set~$\cG$: This is the (non-Gaussian) unitary~$e^{i\round{Q_1}P_2}$. 
Here, the  operator $\round{Q}$
acts in position-space on functions $\Psi$ as a multiplication operator, i.e.,
\begin{align}
 (\round{Q}\Psi)(x)&=\round{x}\Psi(x)\qquad\textrm{ for }\qquad x\in\mathbb{R}\ .
\end{align}
That is, we consider the non-adaptive circuit in Fig.~\ref{fig:envelope-shaping nonadaptive} which is equivalent to the circuit in Fig.~\ref{fig:envelope-protocol adaptive}.
\begin{figure}[!ht]
\centering
 \begin{subfigure}{0.9\textwidth}
 \begin{center}
 \includegraphics{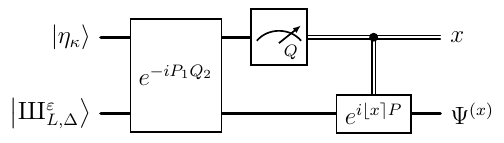}
	 	\caption{Adaptive circuit implementing Protocol~\ref{prot:envelopeshaping}. The last gate in the circuit is classically controlled by the parameter $\lfloor x\rceil\in\mathbb{Z}$, a function of the measurement result~$x\in\mathbb{R}$.
	}\label{fig:envelope-protocol adaptive}
 \end{center}
\end{subfigure}

\hfill

\begin{subfigure}{0.9\textwidth}
\begin{center}
\includegraphics{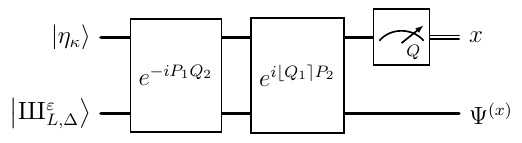}
        \caption{Non-adaptive circuit implementing Protocol~~\ref{prot:envelopeshaping}\label{fig:envelope-shaping nonadaptive}.
        Here the unitary before the measurement does not belong to the set of operations~$\cG$ as it is non-Gaussian. All unitaries are non-adaptive (i.e., they are not controlled by measurement results). 
        }
 \end{center}
 \
 \end{subfigure}
\caption{\label{}Two equivalent circuits realizing the envelope-Gaussification protocol (Protocol~\ref{prot:envelopeshaping}) on input~$\ket{\Sha^\varepsilon_{L,\Delta}}$. 
The first one is identical to that shown in Fig.~\ref{fig:envelope-protocol}, up to the fact that the squeezed vacuum state~$\ket{\eta_\kappa}=S(\log\kappa)\ket{\vac}$ is drawn as an input. We use the second circuit to analyze the behavior of the first.}
\end{figure}

We will use the notion of a quantum instrument to describe the homodyne position-measurement and the associated post-measurement state of the circuit in Fig.~\ref{fig:envelope-shaping nonadaptive}. Quantum instruments are based on completely positive trace-non-increasing maps (CPTNIM)~$\cK:\cB(\cH)\rightarrow\cB(\cH')$ defined on two Hilbert spaces~$\cH,\cH'$, see e.g.,~\cite{Holevo+2013} for further details. An instrument is a CPTNIM-valued measure on a suitable measure space. In the case of homodyne position-measurement, the measure space is given by the Borel-$\sigma$-algebra of~$\mathbb{R}$. If the position of the first mode of a bipartite system of two oscillators is measured, then $\cH\cong L^2(\mathbb{R})^{\otimes 2}=:\cH_1\otimes\cH_2$ and $\cH'\cong L^2(\mathbb{R})$, and the instrument~$\cK$ associated with homodyne position-measurement  is defined by
\begin{align}
\begin{matrix*}[c]
\cK[A]:&\cB(L^2(\mathbb{R})^{\otimes 2}) & \rightarrow & \cB(L^2(\mathbb{R}))\\
&\rho & \mapsto &\cK[A](\rho)=\tr_{\cH_1}\left((\Pi_A\otimes \mathsf{id})\rho\right)\ ,
\end{matrix*} \label{eq: def cK map}
\end{align}
for any Borel-set $A\subseteq\mathbb{R}$. In this expression,  $\Pi_A$  denotes the orthogonal projection onto the subspace~$L^2(\mathbb{R})$ of functions~$\Psi$ having support contained in~$A$. The interpretation is as follows: For a bipartite state~$\rho\in \cB(L^2(\mathbb{R})^{\otimes 2})$, the measurement outcome~$X$ is a random variable satisfying
\begin{align}
\Pr\left[X\in A\right]&=\tr\cK[A](\rho)\ .
\end{align}
Furthermore, if this expression is non-zero, then the conditional post-measurement state~$\rho_{|A}$ conditioned on the event~$X\in A$ is given by the expression
\begin{align}
\rho_{|A}&=\frac{\cK[A](\rho)}{\tr\cK[A](\rho)}\ . \label{eq:rhoA}
\end{align}

The following lemma gives an expression for the overlap of this conditional state when measuring the first mode of 
a state of the form~$e^{i\round{Q_1}P_2}e^{-iP_1Q_2}(\Psi_1\otimes \Psi_2)$. This will be useful to analyse the circuit in Fig.~\ref{fig:envelope-shaping nonadaptive}.

\begin{lemma}\label{lem:convolutioneffect}
Let $\Psi_1,\Psi_2\in L^2(\mathbb{R})$ be two states. Assume that~$\Psi_1$ is even, i.e., 
\begin{align}
\Psi_1(x)&=\Psi_1(-x)\qquad\textrm{ for all }x\in\mathbb{R}\ .\label{eq:evenconditionpsi2}
\end{align}
Define 
\begin{align}
\Phi(y)= \frac{1}{\sqrt{p(0)}}\Psi_1(y)\Psi_2(y)\qquad\textrm{ for all }y\in\mathbb{R}\  \textrm{ where }\qquad p(0)=\int_\mathbb{R} |\Psi_1(z)\Psi_2(z)|^2dz\ .
\end{align}
(That is, $\Phi=M_{\Psi_1} \Psi_2/\|M_{\Psi_1} \Psi_2\|$ where $M_{\Psi_1}:L^2(\mathbb{R})\rightarrow L^2(\mathbb{R})$ is the multiplication operator
which outputs the pointwise product of~$\Psi_1$ and the function it is applied to.) 
Suppose we apply a position-measurement to the first mode of the state
\begin{align}
    \ket{\Psi}=e^{i\round{Q_1}P_2}e^{-iP_1Q_2}(\ket{\Psi_1}\otimes \ket{\Psi_2})\ .
\end{align}
For $x\in\mathbb{R}$ and its associated non-integer part $\delta(x)=x-\round{x}$, we define  
\begin{align}
    p(x)&=\int_{\mathbb{R}} |\Psi_1(x-y)|^2|\Psi_2(y)|^2 dy \label{eq:tgkp p}\\
    m(x)&=\int_{\mathbb{R}}  \overline{\Psi_1(y - \delta(x))} \Psi_1(y)  \overline{\Psi_2(\round{x} + y)} \Psi_2(y)dy\ .
\end{align}
Let $A\subseteq \mathbb{R}$ be a Borel set. Assume that 
$p(A)=\Pr\left[X\in A\right]>0$ for the measurement outcome~$X$.  Let $\rho|_{A}$ denote the corresponding post-measurement state on the second mode conditioned on the event~$X\in A$ that the measurement outcome belongs to~$A$.
 
 Then, the probability of the event~$X\in A$ is equal to
\begin{align}
    p(A)&=\int_A p(x) dx
    \end{align}
    and the conditional state~$\rho_{|A}$
conditioned on the event~$X\in A$ satisfies
\begin{align}
    \bra{\Phi}\rho_{|A}\ket{\Phi}
    &=\frac{1}{p(0) p(A)}\cdot \int_A |m(x)|^2 dx\label{eq:PhirhocondaPhi}\ .
\end{align}
\end{lemma}
\begin{proof}
Let $\ket{\Psi'}=e^{-iP_1Q_2}(\ket{\Psi_1}\otimes\ket{\Psi_2})$. Then, we have 
\begin{align}
\Psi'(x,y)&=(e^{-iyP_1}\Psi_1)(x) \Psi_2(y)=\Psi_1(x-y)\Psi_2(y)\qquad\textrm{ for all }\qquad (x,y)\in\mathbb{R}^2\ .
\end{align}
Since $\ket{\Psi}=e^{i\round{Q_1}P_2}\ket{\Psi'}$, it follows that 
\begin{align}
\Psi(x,y)=\Psi'(x,y+\round{x})=\Psi_1(x-(y+\round{x}))\Psi_2(x+y)
&=\Psi_1(y-\delta(x))\Psi_2(\round{x}+y)\ ,
\end{align}
for all $(x,y)\in\mathbb{R}^2$, where we used that~$\Psi_1$ is even in the last step.
By definition of the conditional post-measurement state~$\rho_{|A}=p(A)^{-1}\cK[A](\rho)$, we have 
\begin{align}
\bra{\Phi}\rho_{|A}\ket{\Phi}&=\frac{1}{p(A)}\bra{\Psi}(\Pi_A\otimes \proj{\Phi})\ket{\Psi}\\
&=\frac{1}{ p(A)} 
\int_{A}dx 
\int_{\mathbb{R}^2}dydy' 
\overline{\Psi(x,y)} \Phi(y)
\overline{\Phi(y')} \Psi(x,y')\\
&=\frac{1}{ p(A)} 
\int_{A}dx 
\int_{\mathbb{R}^2}dydy' 
\overline{\Psi_1(y-\delta(x))\Psi_2(\round{x}+y)} \Phi(y)\\ &\qquad\qquad\qquad\qquad\qquad  \cdot\overline{\Phi(y')} \Psi_1(y'- \delta)\Psi_2(\round{x} + y')\\
&=\frac{1}{p(0) p(A)} 
\int_{A}dx 
\int_{\mathbb{R}}dy
\overline{\Psi_1(y - \delta(x))}  \Psi_1(y) \overline{\Psi_2(\round{x} + y)} \Psi_2(y) \\
&\qquad\qquad\qquad\qquad \cdot \int_{\mathbb{R}}dy' \overline{\Psi_1(y')}  \Psi_1(y' - \delta(x)) \overline{\Psi_2(y')} \Psi_2(\round{x}+y')\\
&=\frac{1}{p(0) p(A)} 
\int_{A}dx 
\left|\int_{\mathbb{R}} dy \overline{\Psi_1(y -\delta(x)})  \Psi_1(y) \overline{\Psi_2(\round{x} + y)} \Psi_2(y) \right|^2\ .
\end{align}
The claim follows from this. 
\end{proof}

\subsubsection{Implications for approximate comb states \label{sec: implication trunc comb states}}
We use Lemma~\ref{lem:convolutioneffect} to analyse the circuit in Fig.~\ref{fig:envelope-shaping nonadaptive}. To do so we apply 
it to the state~$\ket{\Psi_1}=\ket{\eta_\kappa}$ and the state~$\ket{\Psi_2}=\ket{\Sha^{\varepsilon}_{L,\Delta}}$.
 In this case, the state~$\ket{\Phi}$ of Lemma~\ref{lem:convolutioneffect}  is equal to
\begin{align}
    \Phi(y)
    &=\frac{1}{\sqrt{p(0)}} \eta_\kappa(y)\Sha^{\varepsilon}_{L,\Delta}(y)\ .
\end{align}
That is, the state $\ket{\Phi}=\ket{\tGKP^\varepsilon_{L,\kappa,\Delta}}$, where 
\begin{align}
    \tGKP^\varepsilon_{L,\kappa,\Delta}(x)
    &=\frac{1}{\sqrt{p(0)}}
    \frac{1}{\sqrt{L}}\sum_{z=-L/2}^{L/2-1} 
    \eta_\kappa(x) \chi^\varepsilon_\Delta(z)(x)\, ,\label{eq:postmeasurementstateyoutcomedef}
\end{align}
is a truncated version of the state~$\ket{\tGKP_{L,\kappa,\Delta}}$ defined by~\eqref{eq:pointwisegkp}. 
For later reference, we note that the norm
of this state is equal to
\begin{align}
1&=\left\|\tGKP^\varepsilon_{L,\kappa,\Delta} \right\|^2=\frac{1}{L p(0)}\sum_{k=-L/2}^{L/2-1}I_k(0)\ ,\label{eq:postmeasurementgkpkappamv}
\end{align}
where we use the expression
\begin{align}
I_k(\delta)&=\int\chi^\varepsilon_\Delta(k)(y)^2\eta_\kappa(y-\delta)^2 dy\ .\label{eq:ikdeltadefinitionuse}
\end{align}
(Eq.~\eqref{eq:postmeasurementgkpkappamv}
follows because the functions~$\{\chi_\Delta^\varepsilon(z)\}_{z\in \{-L/2,\ldots,L/2-1\}}$ have pairwise-disjoint supports for~$\varepsilon<1/2$.)

The corresponding probability density functions of outcomes is equal to
\begin{align}
    p(x)&=
\int \eta_\kappa(x-y)^2 |\Sha^{\varepsilon}_{L,\Delta}(y)|^2 dy\\
   &= \int \eta_\kappa(y-\delta(x))^2 |\Sha^{\varepsilon}_{L,\Delta}(y+\round{x})|^2dy\ ,
    \end{align}
where we used that $x=\round{x}+\delta(x)$, that $\eta_\kappa(z)=\eta_\kappa(-z)$ for $z\in\mathbb{R}$, and where we substituted~$y-\round{x}$ for $y$. Inserting the definition of~$\Sha^\varepsilon_{L,\Delta}$ into the above, we obtain
    \begin{align}
  p(x)  &=\frac1L \int \eta_\kappa(y-\delta(x))^2 \left(\sum_{z=-L/2}^{L/2-1} \chi^\varepsilon_\Delta(z)(y+\round{x})\right)^2 dy\\
    &=\frac1L \int \eta_\kappa(y-\delta(x))^2 \left(\sum_{z=-L/2}^{L/2-1} \chi^\varepsilon_\Delta(z-\round{x})(y)\right)^2 dy\\
    &=\frac1L \int\eta_\kappa(y-\delta(x))^2\left(\sum_{k=-L/2-
    \round{x}}^{L/2-\round{x}-1} \chi^\varepsilon_\Delta(k)(y)\right)^2 dy\ ,
\end{align}
where we first used that
\begin{align}
    \chi_\Delta^\varepsilon(z)(y+\round{x})&=\Psi^\varepsilon_\Delta(y+\round{x}-z)=
    \Psi^\varepsilon_\Delta(y-(z-\round{x}))=\chi^\varepsilon_\Delta(z-\round{x})(y)\label{eq:integralsquaredetaydetla}
\end{align}
by the symmetry of the truncated centered Gaussian~$\Psi^\varepsilon_\Delta$, and then changed the index of summation.
We can write~\eqref{eq:integralsquaredetaydetla}  as
\begin{align}
p(x)   &=\frac1L \sum_{k_1=-L/2-\round{x}}^{L/2-\round{x}-1} \sum_{k_2=-L/2-\round{x}}^{L/2-\round{x}-1} M_{k_1,k_2}(x)\ ,\label{eq:pdfoutcomex}
\end{align}
where 
\begin{align}
    M_{k_1,k_2}(x)&:=\int \eta_\kappa(y-\delta(x))^2 \chi^\varepsilon_\Delta(k_1)(y)\chi^\varepsilon_\Delta(k_2)(y) dy\ . \label{eq: Mkk}
\end{align}
Because $\varepsilon<1/2$,  the expression~$\chi^\varepsilon_\Delta(k_1)(y)\chi^\varepsilon_\Delta(k_2)(y)$ can be non-zero only if $k_1=k_2$. Thus, the integral $M_{k_1,k_2}(x)$ vanishes unless $k_1=k_2$, and we have $M_{k,k}(x)=I_k(\delta(x))$
where $I_k(\delta)$ is defined by~\eqref{eq:ikdeltadefinitionuse}.  
We conclude that 
\begin{align}
    p(x)&=\frac{1}{L}\sum_{k=-L/2-\round{x}}^{L/2-1-\round{x}}I_k(\delta(x)) \\
    &= \frac{1}{L}\sum_{k=-L/2}^{L/2-1}I_{k-\round{x}}(\delta(x))\\
    &=\frac{1}{L} \sum_{k=-L/2}^{L/2-1} I_{k} (\round{x}+\delta(x))\\
    &=\frac{1}{L}\sum_{k=-L/2}^{L/2-1} I_{k}(x) \qquad\textrm{ since $x = \round{x} + \delta(x)$}\label{eq:pdf outcomes x Ik}\ ,
\end{align}
where we relabelled the index of the sum in the second step and used that $I_{k-\round{x}}(\delta(x))=I_k(\round{x}+\delta(x))$ in the third line (cf.~\eqref{eq:ikmdeltarelation} in the Appendix for a proof of this relation) .
For later reference, we observe that we can rephrase~\eqref{eq:pdf outcomes x Ik} as follows. For an integer~$m\in\mathbb{Z}$ and $\delta\in[-1/2,1/2)$, the probability density function $p(m+\delta)$ is equal to the following two expressions:
\begin{align}
p(m+\delta) &= 
\frac{1}{L}\sum_{k=-L/2-m}^{L/2-m-1}I_{k}(\delta)=
\frac{1}{L}\sum_{k=-L/2}^{L/2-1}I_k(m+\delta)\ 
\label{eq:pdeltabound}\ .
\end{align}

We will show that the conditional post-measurement state~$\rho_{|\Omega_L}$ in the second register, which is conditioned on the measurement result~$X$ in the first register being  in the acceptance region
\begin{align}
    \Omega_L =[-L/8-1/2,L/8+1/2]\label{eq:acceptancesetomegaL}
\end{align}
of Protocol~\ref{prot:envelopeshaping} (line~\ref{prot: gaussification acceptance}), i.e., conditioned on the event~$X\in \Omega_L$, 
is close to the state $\ket{\tGKP^\varepsilon_{L,\kappa,\Delta}}$. 
\begin{lemma}\label{lem:translationcorrectedpost}
Let $L\in 8\mathbb{N}$ be an integer multiple of~$8$, let $\Delta>0$, let $\kappa >0$, let $\varepsilon \in (0,1/2)$, and let $\Omega_L\subset\mathbb{R}$ be the set~\eqref{eq:acceptancesetomegaL} of measurement results~$x\in\mathbb{R}$ for which Protocol~\ref{prot:envelopeshaping}  accepts. 
Let  us denote by
$\rho_{|\Omega_L}(\varepsilon,\kappa,\Delta,L)\in \cB(L^2(\mathbb{R}))$ the output state conditioned on the event that the protocol accepts 
 on input~$(\ket{\Sha^\varepsilon_{L,\Delta}},\kappa,L)$. Then, 
\begin{align}
    \bra{\tGKP^\varepsilon_{L,\kappa,\Delta}}\rho_{|\Omega_L}(\varepsilon,\kappa,\Delta,L)\ket{\tGKP^\varepsilon_{L,\kappa,\Delta}}&\ge 1-3\kappa^2 L/2 -4e^{-\kappa^2 L^2/32}\ .
\end{align}
\end{lemma}
\begin{proof}
Because the state~$\ket{\Phi}$ defined in Lemma~\ref{lem:convolutioneffect} is equal to~$\ket{\Phi}=\ket{\tGKP^\varepsilon_{L,\kappa,\Delta}}$, the claim follows from Eq.~\eqref{eq:PhirhocondaPhi},  which can be rewritten as
\begin{align}
\bra{\tGKP^\varepsilon_{L,\kappa,\Delta}}\rho_{|\Omega_L}(\varepsilon,\kappa,\Delta,L)\ket{\tGKP^\varepsilon_{L,\kappa,\Delta}}&=
\int_{\Omega_L} \frac{p(x)}{p(\Omega_L)} \cdot \frac{|m(x)|^2}{p(0)p(x)} dx\ 
\end{align}
because $p(x)>0$ for all $x\in \Omega_L$ (in fact,  we even have $p(x)>0$ for all $x\in\mathbb{R}$).
Observing that $p(x)/p(\Omega_L)=p(x|\Omega_L)$ 
 is the conditional distribution given that the measurement result satisfies~$X\in \Omega_L$,  it follows that
\begin{align}
\bra{\tGKP^\varepsilon_{L,\kappa,\Delta}}\rho_{|\Omega_L}(\varepsilon,\kappa,\Delta, L)\ket{\tGKP^\varepsilon_{L,\kappa,\Delta}}&= 
\int_{\Omega_L}p(x|\Omega_L)
\frac{|m(x)|^2}{p(0)p(x)} dx\geq \inf_{x\in \Omega_L} \frac{m(x)^2}{p(0)p(x)} \ .
\end{align}
 Here, we used that  $m(x)\geq 0$ for all $x\in\mathbb{R}$. 
  The claim is a consequence of this inequality, the definition~\eqref{eq:acceptancesetomegaL} of~$\Omega_L$ and Lemma~\ref{lem:lowerboundmxpx} below.  
\end{proof}

\begin{lemma}\label{lem:lowerboundmxpx}
Suppose that $\kappa>0$, $\Delta>0$, $\varepsilon\in(0,1/2)$ and $L\in 8\mathbb{N}$. 
Let $x=m+\delta$ with $m=\round{x}$, $|m|\leq L/8$ and $|\delta|\leq 1/2$.
Consider the quantities $p(x)\in\mathbb{R}$, $m(x)\in\mathbb{C}$
defined by Lemma~\ref{lem:convolutioneffect} applied to~$\Psi_1=\eta_\kappa$ and~$\Psi_2=\Sha^\varepsilon_{L,\Delta}$.
 Then
\begin{align}
\frac{m(x)^2}{p(0)p(x)}\ge   
1-3\kappa^2 L/2-4e^{-\kappa^2 L^2/32}\ .
\end{align}
\end{lemma}

\begin{proof}
By definition of $x=m+\delta$, we have $\round{x} = m$ and $\delta=x-\round{x}$ and thus
\begin{align}
m(x)&=\int_{\mathbb{R}} \eta_\kappa(y - \delta) \eta_\kappa(y)
\Sha^{\varepsilon}_{L,\Delta}(y)
\Sha^{\varepsilon}_{L,\Delta}(m+y) dy\\
&=\frac{1}{L}\sum_{k_1=-L/2}^{L/2-1}\sum_{k_2=-L/2}^{L/2-1} \int \eta_\kappa(y-\delta)\eta_\kappa(y) \chi^\varepsilon_\Delta(k_1)(y) \chi^\varepsilon_\Delta(k_2)(m + y)  dy
\end{align}
Similarly as before (see Eq.~\eqref{eq:integralsquaredetaydetla}), we can use that $\chi^\varepsilon_\Delta(k_2)(m+y)=\chi^\varepsilon_\Delta(k_2-m)(y)$ to obtain
\begin{align}
m(x)&= \frac{1}{L}\sum_{k_1=-L/2}^{L/2-1}\sum_{k_2=-L/2}^{L/2-1} \int \eta_\kappa(y-\delta)\eta_\kappa(y) \chi^\varepsilon_\Delta(k_1)(y) \chi^\varepsilon_\Delta(k_2- m)( y)  dy\\
&=\frac{1}{L}\sum_{k_1=-L/2}^{L/2-1}\sum_{k_2=-L/2-m}^{L/2-m-1} \int \eta_\kappa(y-\delta)\eta_\kappa(y) \chi^\varepsilon_\Delta(k_1)(y) \chi^\varepsilon_\Delta(k_2)( y)  dy\\
&=\frac{1}{L}\sum_{k_1=-L/2}^{L/2-1}\sum_{k_2=-L/2-m}^{L/2-m-1} M_{k_1,k_2}'(\delta)\label{eq:mk1k2sumk1}
\end{align}
where we shifted the summation index~$k_2$ and introduced the scalars
\begin{align}
 M'_{k_1,k_2}(\delta)=\int \eta_\kappa(y-\delta)\eta_\kappa(y) \chi^\varepsilon_\Delta(k_1)(y) \chi^\varepsilon_\Delta(k_2)(y) dy \ . \label{eq: Mkk prime}
\end{align}
Because  the expression~$\chi^\varepsilon_\Delta(k_1)(y)\chi^\varepsilon_\Delta(k_2)(y)$ can only be non-zero when $k_1=k_2$ (since $\varepsilon<1/2$), the integral $M'_{k_1,k_2}(\delta)$ vanishes unless $k_1=k_2$. 
We conclude that
\begin{align}
    m(x)
    &=\begin{cases}
        \frac{1}{L}
    \sum_{k=-L/2}^{L/2-m-1}
    M'_{k}(\delta)\qquad &\textrm{ if } m\geq 0\\
    \frac{1}{L}\sum_{k=-L/2+|m|}^{L/2-1}M_k'(\delta)\qquad &\textrm{ if }m<0\ ,
    \end{cases}
\end{align}
where we set $M'_{k}=M'_{k,k}$. Recalling that we are considering $m\in\mathbb{Z}$ with $|m|\leq L/8$ and using that each term \begin{align}
  M'_k(\delta)=\int \eta_\kappa(y-\delta)\eta_\kappa(y) \chi^\varepsilon_\Delta(k)(y)^2dy  \label{eq:mkprimedeltav}
\end{align} is non-negative, we obtain the lower bound
\begin{align}
    m(x)&\geq \frac{1}{L}\sum_{k=-L/2+L/8}^{L/2-L/8-1}M'_k(\delta)\\
    &\geq\frac{1}{L}\sum_{k=-L/4}^{L/4-1}M'_k(\delta)\label{eq:mxlkdeltalow}
\end{align}
where we used that $L/2-L/8-1\geq L/4-1$ and $-L/2+L/8<-L/4$ since $L\geq 8$ by the assumption that $L$ is an integer multiple of~$8$.

In the expression~\eqref{eq:mkprimedeltav} defining~$M_k'(\delta)$, we can restrict the domain of integration to~$[k-\varepsilon,k+\varepsilon]$
 as the expression~$\chi^\varepsilon_\Delta(k)(y)$ vanishes for $y\in\mathbb{R}$ outside this interval, i.e., we have 
\begin{align}
    M'_{k}(\delta)=\int_{k-\varepsilon}^{k+\varepsilon} \eta_\kappa(y-\delta)\eta_\kappa(y) \chi^\varepsilon_\Delta(k)(y)^2  dy \ .\label{eq:mprimekdeltaex}
\end{align}
We will show that
\begin{align}
    M'_{k}(\delta)\ge e^{-\kappa^2 L/4} \int_{k-\varepsilon}^{k+\varepsilon} \eta_\kappa(y)^2 \chi^\varepsilon_\Delta(k)(y)^2  dy\quad\textrm{ for all }k\in \{-L/4,\ldots,L/4-1\}\ . \label{eq: M prime k lower bound}
\end{align}
We note that
using the definition~\eqref{eq:ikdeltadefinitionuse}
of the integral~$I_k(\delta)$, Eq.~\eqref{eq: M prime k lower bound} can be expressed as 
\begin{align}
     M'_{k}(\delta)\ge e^{-\kappa^2 L/4} I_k(0)\quad\textrm{ for all }k\in \{-L/4,\ldots,L/4-1\}\ . \label{eq: M prime k lower boundcompact} 
\end{align}
We  use the identity
\begin{align}
    \eta_\kappa(y-\delta)
    &=\frac{\sqrt{\kappa}}{\pi^{1/4}} e^{-\kappa^2 (y-\delta)^2/2}\\
    &=\frac{\sqrt{\kappa}}{\pi^{1/4}} e^{-\kappa^2 y^2/2}\cdot e^{\kappa^2 y\delta}\cdot e^{-\kappa^2 \delta^2/2}\\
    &=\eta_\kappa(y)\cdot e^{\kappa^2 y\delta }\cdot e^{-\kappa^2 \delta^2/2}
\end{align}
which implies that
\begin{align}
    \eta_\kappa(y-\delta)&\geq \eta_\kappa(y)\cdot e^{-\kappa^2 (|y|\cdot |\delta|+\delta^2/2)}\ .
\end{align}
For $y\in [k-\varepsilon,k+\varepsilon]$
we have~$|y|\leq |k|+\varepsilon\leq L/4+\varepsilon$ by the assumption $k\in \{-L/4,\ldots,L/4-1\}$, hence we have for $\delta \in [-1/2,1/2]$
\begin{align}
\eta_\kappa(y-\delta)&\geq \eta_\kappa(y)\cdot 
e^{-\kappa^2 \left((L/4+\varepsilon)/2+1/8\right)}\\
&\geq 
\eta_\kappa(y)\cdot e^{-\kappa^2 L/4}
\quad\textrm{ for }y\in [k-\varepsilon,k+\varepsilon]\ .\label{eq:basiinequaltm}
\end{align}
Here, we used that $(L/4+\varepsilon)/2+1/8\leq L/4$ for $L\in 8\mathbb{N}$ and $\varepsilon<1/2$.
Bounding the term~$\eta_\kappa(y-\delta)$ in Eq.~\eqref{eq:mprimekdeltaex} using~\eqref{eq:basiinequaltm} and using the monotonicity of the integral implies the claim in Eq.~\eqref{eq: M prime k lower bound}

Applying 
Eq.~\eqref{eq: M prime k lower boundcompact}
to each term in~\eqref{eq:mxlkdeltalow}, we obtain the bound
\begin{align}
   \frac{m(x)}{\sqrt{p(x)p(0)}}
    &\ge\frac{1}{L\sqrt{p(x)p(0)}}
    \sum_{k=-L/4}^{L/4-1}
    M'_{k}(\delta)\\
    &\ge e^{-\kappa^2L/4} \frac{1}{L\sqrt{p(x)p(0)}}
    \sum_{k=-L/4}^{L/4-1}
    I_{k}(0)\label{eq:lowerbndqinterest}
\end{align}
on the quantity of interest. 
Dividing~\eqref{eq:lowerbndqinterest} by the norm~\eqref{eq:postmeasurementgkpkappamv} 
gives the  expression
\begin{align}
\frac{m(x)}{\sqrt{p(x)p(0)}}
&\ge e^{-\kappa^2L/4}\left(\frac{p(0)}{p(x)}\right)^{1/2}\cdot\,
\frac{\sum_{k=-L/4}^{L/4-1}I_k(0)}{\sum_{k=-L/2}^{L/2-1}I_k(0)}\\
&=e^{-\kappa^2L/4}\left(\frac{\sum_{k=-L/2}^{L/2-1}I_{k}(0)}{\sum_{k=-L/2-m}^{L/2-1-m}I_k(\delta)}\right)^{1/2}\cdot\,
\frac{\sum_{k=-L/4}^{L/4-1}I_k(0)}{\sum_{k=-L/2}^{L/2-1}I_k(0)}\ ,\label{eq:mxpxp0 lowerbound}
\end{align}
where in the last step, we used Eq.~\eqref{eq:pdeltabound} and the assumption~$x=m+\delta$. 

We will lower bound each factor in~\eqref{eq:mxpxp0 lowerbound} separately.
By Lemma~\ref{lem:deltazeroonehalfintegermtwo} and the assumption $\abs{m}\le L/8$, we have 
\begin{align}
    \sum_{k=-L/2-m}^{L/2-1-m}I_k(\delta)  &\leq     e^{\kappa^2 L}    \sum_{k=-L/2-m}^{L/2-1-m}I_k(0)\\
    &\leq    e^{\kappa^2 L}    \sum_{k=-L}^{L-1}I_k(0)\ ,\label{eq:upperboundikdeltam}
\end{align}
where we used that each term~$I_k(0)$ is non-negative in the second inequality.
Eq.~\eqref{eq:upperboundikdeltam}
implies that
\begin{align}
\frac{\sum_{k=-L/2}^{L/2-1}I_{k}(0)}{\sum_{k=-L/2-m}^{L/2-1-m}I_k(\delta)}
     &\ge e^{-\kappa^2 L}
     \frac{\sum_{k=-L/2}^{L/2-1}I_k(0)}{\sum_{k=-L}^{L-1}I_k(0)}\\
     &\geq e^{-\kappa^2 L}\cdot (1-e^{-\kappa^2 L^2/8})\ \label{eq:firstfactoreqprodm}
     \end{align}
by Lemma~\ref{lem:fractionalboundmnv}. 
Similarly, Lemma~\ref{lem:fractionalboundmnv}
applied with $L/2$ instead of~$L$ gives the lower bound
\begin{align}
    \frac{\sum_{k=-L/4}^{L/4-1}I_{k}(0)}{\sum_{k=-L/2}^{L/2-1}I_k(0)}
    &\geq 1-e^{-\kappa^2 L^2/32}\ .\label{eq:combinedeqam}
     \end{align}
Combining~\eqref{eq:firstfactoreqprodm} and~\eqref{eq:combinedeqam} with \eqref{eq:mxpxp0 lowerbound}
gives the lower bound
\begin{align}
\frac{m(x)}{\sqrt{p(x)p(0)}}
&\ge e^{-\kappa^2 L/4} e^{-\kappa^2 L/2}
(1-e^{-\kappa^2 L^2/8})^{1/2} (1-e^{-\kappa^2 L^2/32})\\
&\ge e^{-\kappa^2 L/4} e^{-\kappa^2 L/2}
(1-e^{-\kappa^2 L^2/8}) (1-e^{-\kappa^2 L^2/32})\\
&\ge e^{-3\kappa^2 L/4} (1-e^{-\kappa^2 L^2/8}-e^{-\kappa^2 L^2/32})\\
&\geq (1-3\kappa^2 L/4)(1-2e^{-\kappa^2 L^2/32})\\
&\geq 1-3\kappa^2 L/4-2e^{-\kappa^2 L^2/32}
\end{align}
where we used inequality $\sqrt{1-x}\ge 1-x$ for $x\in(0,1)$  in the second step,  the inequality $(1-x)(1-y) \ge 1-x-y$ for $x,y\geq 0$ in the third and last step, and the inequality~$e^{-x}\geq 1-x$ in the fourth step. The claim follows by inequality $(1-x)^2\ge 1-2x$ for $x\in\mathbb{R}$.
\end{proof}

We will show that the conditional output state~$\rho_{|\Omega_L}(\varepsilon, \kappa,\Delta,L) \in \cB(L^2(\mathbb{R}))$ of the second mode --- conditioned on 
the measurement result~$x$ (when applying a position-measurement to the first mode)
belonging to the 
``acceptance region'' $\Omega_L$ (cf.\ Protocol~\ref{prot:envelopeshaping}, line~\ref{prot: gaussification acceptance}) --- is close to an approximate GKP state (for suitably chosen parameters~$(\kappa,\Delta)$ and~$L$). 

\begin{lemma} \label{lem: env shaping rhoacc gkp trace dist}
    Assume $\kappa\in(0,1/4)$, $\Delta\in(0, 1/4)$, $\varepsilon\in [\sqrt{\Delta},1/2)$, and $L\in\mathbb{N}$.
     Let $\rho_{|\Omega_L}(\varepsilon,\kappa,\Delta,L)\in L^2(\mathbb{R})$ 
 be the output state 
 of Protocol~\ref{prot:envelopeshaping} conditioned on acceptance 
 (see Lemma~\ref{lem:convolutioneffect}), on input $(\ket{\Sha^\varepsilon_{L,\Delta}},\kappa,L)$.
 Then, \begin{align}
        \left\|\rho_{|\Omega_L}(\varepsilon,\kappa,\Delta,L)-\proj{\gkp_{\kappa,\Delta}}\right\|_1 \le 6\kappa\sqrt{L} +6\sqrt{\Delta}+ 7e^{-\kappa^2 L^2/64}\ .
    \end{align}
\end{lemma}
\noindent We note that  Lemma~\ref{lem: env shaping rhoacc gkp trace dist} implies that in the limit~$(\kappa,\Delta)\rightarrow(0,0)$
a choice of $L$~scaling, e.g., as $L=\Theta((1/\kappa)^{4/3})$ ensures that the approximation error scales polynomially in~$(\kappa, \Delta)$.

\begin{proof}
Lemma~\ref{lem:translationcorrectedpost} together with relation 
 $\|\rho-\proj{\Psi}\|_1= 2\sqrt{1-\langle \Psi,\rho\Psi\rangle}$
between the overlap and the $L^1$-distance
for a state~$\rho\in\cB(L^2(\mathbb{R}))$ and a pure state~$\ket{\Psi}\in L^2(\mathbb{R})$ implies 
\begin{align}     
    \left\|\proj{\tGKP^\varepsilon_{L,\kappa,\Delta}}-\rho_{|\Omega_L}(\varepsilon,\kappa,\Delta, L)\right\|_1 
    &\le 2\sqrt{3\kappa^2 L/2+4e^{-\kappa^2 L^2/32}}\\
    &\le 3 \kappa\sqrt{L} + 4 e^{-\kappa^2 L^2/64}\label{eq:gkprhoclone}
\end{align}
where we used the inequality $\sqrt{x+y} \le \sqrt{x} + \sqrt{y}$ for all $x,y \ge 0$ on the last line.
By assumption, we have $\varepsilon\in[\sqrt{\Delta}, 1/2)$, thus Corollary~\ref{cor: approximate tgkpL GKP trace dist} yields
\begin{align}
    \left\| \proj{\tGKP^\varepsilon_{L,\kappa,\Delta}}-\proj{\gkp_{\kappa,\Delta}}\right\|_1
    &\le 3\kappa\sqrt{L} +6\sqrt{\Delta}+ 3e^{-\kappa^2L^2/8}\ . \label{eq:hmonex}
\end{align}
Eqs.~\eqref{eq:gkprhoclone},~\eqref{eq:hmonex} combined with the triangle inequality  imply the claim.
\end{proof}

\subsubsection{Bounding the acceptance probability of Protocol~\ref{prot:envelopeshaping} \label{sec: bound acc prob prot 2}}
We prove a lower bound on the acceptance probability of Protocol~\ref{prot:envelopeshaping} given the input state $\ket{\Sha_{L,\Delta}^\varepsilon}$ and the input parameters~$(\kappa,L)$.
\begin{lemma}\label{lem: prot 2 accept prob}
Assume $\kappa>0$, $\Delta >0$, $\varepsilon \in (0,1/2)$ and $L\in 16\mathbb{N}$. Given as input the state $\ket{\Sha_{L,\Delta}^\varepsilon}$ and the parameters~$\kappa$ and $L$, 
Protocol~\ref{prot:envelopeshaping} accepts with probability at least
\begin{align}
    \Pr\Big[\textnormal{Protocol~\ref{prot:envelopeshaping} } \accepts \mid \ket{\Sha_{L,\Delta}^\varepsilon}
    \Big]&\ge \frac{1}{8}\left(1 - 2e^{-\kappa^2L^2/256}\right) \ .
\end{align}
\end{lemma}
\begin{proof}
The probability that Protocol~\ref{prot:envelopeshaping} accepts is the probability of obtaining a measurement outcome $x$ belonging to the acceptance region (cf.\ line~\ref{prot: gaussification acceptance} of Protocol~\ref{prot:envelopeshaping})
\begin{align}
    \Omega_{L}=[-L/8 - 1/2, L/8+ 1/2]\ .
\end{align}
That is,
\begin{align}
    \Pr\left[\textnormal{Protocol~\ref{prot:envelopeshaping} } \accepts \mid \ket{\Sha_{L,\Delta}^\varepsilon} \right] &=p(\Omega_{L})=\int_{\Omega_{L}}p(x) dx\ ,
\end{align}
where $p(x)$ is the probability density of obtaining outcome $x$ (defined in Eq.~\eqref{eq:pdfoutcomex}). By Eq.~\eqref{eq:pdf outcomes x Ik} it satisfies 
\begin{align}
    p(\Omega_{L})&=\int_{-L/8-1/2}^{L/8 +1/2}
    p(x) dx\\
    &=  
    \frac{1}{L} 
    \sum_{k=-L/2}^{L/2-1}
    \int_{-L/8-1/2}^{L/8+1/2}
    I_k(x)  dx \ ,
\end{align}
where we recall (cf.~\eqref{eq:ikdeltadefinitionuse}) that
\begin{align}
    I_k(\delta)
    &:=\int\eta_\kappa(u-\delta)^2\chi^\varepsilon_\Delta(k)(u)^2 du\ .
\end{align}
By using Fubini's Theorem, we obtain
\begin{align}
    p(\Omega_{L})
    &= \frac{1}{L} \sum_{k=-L/2}^{L/2-1} \int \left(\int_{-L/8-1/2}^{L/8 +1/2} \eta_\kappa(u-x)^2 dx\right) \chi^\varepsilon_\Delta(k)(u)^2 du \\
    &= \frac{1}{L}\sum_{k=-L/2}^{L/2-1} \int \left(\int_{-L/8-1/2-u}^{L/8+1/2-u} \eta_\kappa(z)^2 dz\right) \chi^\varepsilon_\Delta(k)(u)^2 du \ ,
\end{align}
where we substituted $z = u-x$. Inserting the definition of $\chi_\Delta^\varepsilon(k)$ and defining the inner integral as $\Theta(u)$, i.e.,
\begin{align}
    \Theta(u) = \int_{-L/8-1/2-u}^{L/8+1/2-u} \eta_\kappa(z)^2 dz
\end{align}
we can write
\begin{align}
     p(\Omega_{L}) &= \frac{1}{L}\sum_{k=-L/2}^{L/2-1} \int \Theta(u) \Psi_\Delta^\varepsilon(u-k)^2 du \\
     &=  \frac{1}{L} \sum_{k=-L/2}^{L/2-1} \int  \Theta(v+k) \Psi_\Delta^\varepsilon(v)^2 dv\ ,
\end{align} where we substituted $v = u-k$. Moreover, as $\textrm{supp}(\Psi_\Delta^\varepsilon) \subseteq [-\varepsilon, \varepsilon]$, we can restrict us to $|v|\le \varepsilon < 1/2$. Due to the positivity of the integrand, we have
\begin{align}
    \Theta(v+k)= \int_{-L/8-1/2-v -k }^{L/8 +1/2 - v-  k}  \eta_\kappa(z)^2 dz \ge \int_{-L/8 - k}^{L/8 - k} \eta_\kappa(z)^2 dz \qquad \textrm{ for $-L/2 \le k \le L/2-1$}\ . 
\end{align}
    Using that $\|\Psi_\Delta^\varepsilon\|=1$, we can bound
    \begin{align}
         p(\Omega_{L}) &\ge \frac{1}{L} \sum_{k=-L/2}^{L/2-1}\int \left( \int_{-L/8 - k}^{L/8 - k} \eta_\kappa(z)^2 dz\right) \Psi_\Delta^\varepsilon(v)^2 dv \\
         &= \frac{1}{L} \sum_{k=-L/2}^{L/2-1} \int_{-L/8 - k}^{L/8 - k} \eta_\kappa(z)^2 dz \int \Psi_\Delta^\varepsilon(v)^2 dv\\
         & =  \frac{1}{L} \sum_{k=-L/2}^{L/2-1} \int_{-L/8 - k}^{L/8 - k} \eta_\kappa(z)^2 dz\, . \label{eq: double sum eta}
    \end{align}
Using that the integrand is non-negative, and that the interval $[-L/8-k,L/8-k]$ contains the interval $[-L/16,L/16]$ for all  $k\in \{-L/16, \dots L/16\}$, and that there are at least $L/8$ such values of~$k\in\{-L/2,\ldots,L/2-1\}$ (by the assumption $L\in 8\mathbb{N}$), we obtain the lower bound
\begin{align}
    p(\Omega_{L}) &\ge \frac{1}{8}\int_{-L/16}^{L/16} \eta_\kappa(z)^2 dz\ .
\end{align}
Inserting this into~\eqref{eq: double sum eta} yields
    \begin{align}
        p(\Omega_L) \ge \frac{1}{8} \int_{-L/16}^{L/16} \eta_\kappa(z)^2 dz  \label{eq: prob bound gaussian}\ .
    \end{align}
    Since $\eta_\kappa(\cdot)^2$ is the probability density function of a centered normal random variable $X \sim \cN(0, 1/(2\kappa^2))$
    we obtain by the Chernoff bound (see Ref.~\cite{Vershynin_2018}):
    \begin{align}
        \Pr[\abs{X}\ge L/16]&\le 2e^{-(L/16)^2\kappa^2}\\
        &= 2e^{-L^2\kappa^2/256}\, .
    \end{align}
  Inserting this tail bound into Eq.~\eqref{eq: prob bound gaussian} yields the claim.  
\end{proof}

\subsection{Proof of Theorem~\ref{thm:envelope gaussification}} \label{sec:envelope Gaussification error}
In this section, we prove Theorem~\ref{thm:envelope gaussification}. We rely on
the analysis of Protocol~\ref{prot:envelopeshaping} 
in the case where the input is a truncated comb state~$\ket{\Sha^\varepsilon_{L,\Delta}}$, see Section~\ref{sec:shastateinputanalysis}.  To extend to arbitrary input states that are close to~$\ket{\Sha_{L,\Delta}}$, we first analyze the effect of heralding, see Section~\ref{sec: heralded channels}.
We show that  --- assuming a constant lower bound on the acceptance probability --- a heralding channel
is stable under 
deviations in the input state (quantified in terms of the $L^1$-distance). Building on this stability result, we then complete the proof of Theorem~\ref{thm:envelope gaussification} in Section~\ref{sec: complete proof env gaussification}

\subsubsection{Heralding channels and approximate input states}\label{sec: heralded channels}
In this section, we examine what conclusions can be drawn if we know that a heralding channel prepares, with known probability, an output that is close to a target state. More formally, we consider a quantum channel, i.e., a completely positive map (CPTPM)
\begin{align}
\begin{matrix}
\cE: &\cB(\cH) & \rightarrow & \cB(\cH') \otimes \mathbb{C}^2\\
&\rho & \mapsto &\cE_\acc (\rho) \otimes \proj{\acc} + \cE_\rej(\rho) \otimes \proj{\rej}\ ,
\end{matrix}
\end{align}
where 
$\cE_\acc$ and $\cE_\rej$ are completely positive trace-non-increasing maps (CPTNIM), $\cH, \cH'$ are Hilbert spaces, and the second register is a qubit representing classical information spanned by the two orthonormal basis states~$\{\ket{\acc},\ket{\rej}\}$.

For any input state~$\rho$, let us defined the heralded states
\begin{align}
   \rho_\acc = \frac{1}{\Pr[\acc|\rho]} \cE_\acc(\rho) \qquad\textrm{ and }\qquad \rho_\rej = \frac{1}{\Pr[\rej|\rho]} \cE_\rej(\rho)\ ,
\end{align}
where 
\begin{align}
    \Pr[\acc|\rho] = \tr \cE_\acc(\rho) \qquad\textrm{ and }\qquad \Pr[\rej|\rho] = \tr \cE_\rej(\rho)
\end{align}
are, respectively, the success and failure probabilities of the heralding. We note that the acceptance probability can be written as 
\begin{align}
    \Pr[\acc|\rho] &=\tr (\Pi\rho) \qquad \textrm{ where } \qquad \Pi=\cE_\acc^\dagger(I)\ ,\label{eq:povmpiaccrho}
\end{align}
and $\cE^\dagger_\acc$ denotes the adjoint map of $\cE_\acc$. The operator~$\Pi$ satisfies $0\leq \Pi\leq I$. We will call an operator~$\Pi$ with this property a POVM element in the following. It corresponds to the binary-outcome POVM~$\{\Pi,I-\Pi\}$.

The main result of this section is the following statement expressed by Lemma~\ref{lem: heralded channel appx states dist pacc}. Assume that for an input state~$\rho$,  application of~$\cE$ results in acceptance with probability close to~$1$, and that the heralded state~$\rho_\acc$  is close to a desired target state~$\sigma_{\mathsf{targ}}$.
Then the same is true for a different input state~$\tau$, assuming that~$\tau$ is close to~$\rho$.

\begin{lemma}\label{lem: heralded channel appx states dist pacc}
Let $\delta, \gamma>0$ and $\rho, \tau, \sigma_{\mathsf{targ}}\in\cB(\cH)$ be states. Consider the  heralding quantum channel $\cE:\cB(\cH)\to \cB(\cH')\otimes\mathbb{C}^2$ as introduced above and assume the following conditions are satisfied:
\begin{enumerate}[(i)]
    \item $ \left\| \rho - \tau \right\|_1 \le \delta$\,,\label{it:input closeness}
    \item $\|\rho_\acc - \sigma_{\mathsf{targ}}\|_1 \le \gamma $\,,\label{it:pointwise closeness}
    \item $\Pr[\acc|\rho] \neq 0$\,.
\end{enumerate}
Then the
probability of acceptance on input $\tau$ is at least
\begin{align}
    \Pr[\acc|\tau]\ge \Pr[\acc|\rho]-\delta/2 \ , \label{eq: heralded prob}
\end{align}
and upon acceptance the heralded state $\tau_\acc$ and the target state $\sigma_{\mathsf{targ}}$ satisfy
\begin{align}
    \left\| \tau_{\acc}-\sigma_{\mathsf{targ}}\right\|_1\le \frac{\delta}{\Pr[\acc|\rho]}+\gamma\ .\label{eq: min tr eq}
\end{align}
\end{lemma}

\begin{proof}
Our proof  heavily exploits that the trace distance between two states $\rho,\tau\in\cB(\cH)$ can be written in variational form as
\begin{align}
    \left\| \rho  - \tau \right\|_1= 2 \max_{0\le \Pi\le I} \tr\left(\Pi(\rho-\tau) \right)\ ,\label{eq:variational trace dist}
\end{align}
where $\Pi\in\cB(\cH)$ satisfies $0\leq \Pi\leq I$. 
The assumption \eqref{it:input closeness} and the specialization of the expression~\eqref{eq:variational trace dist} to $\Pi=\cE^\dagger(I)$ ,see Eq.~\eqref{eq:povmpiaccrho}, give us the bound
\begin{align}
    \Pr[\acc|\rho]-\Pr[\acc|\tau]\le \max_{0\le \Pi\le I} \tr(\Pi(\rho-\tau))
    =\frac12 \|\rho-\tau\|_1
\le \frac\delta 2 \ ,\label{eq:bound abs delta prob}
\end{align}
cf.~\eqref{eq:povmpiaccrho}. 
This implies the claimed lower bound~\eqref{eq: heralded prob} on $\Pr[\acc|\tau]$.

We now prove that on input $\tau$ and conditioned on acceptance, the output $\tau_\acc$ is close to the target state~$\sigma_{\mathsf{targ}}$ (see Eq.~\eqref{eq: min tr eq}).
Because of the contractivity of the trace distance under quantum channels, we have by the assumption~\eqref{it:input closeness} that
\begin{align}
\left\| \cE(\rho) - \cE(\tau) \right\|_1 \le \left\| \rho - \tau \right\|_1 \le \delta\,\label{contractivity_step}\ .
\end{align}
\noindent For simplicity, we abbreviate the difference in acceptance probability as 
\begin{align}
    \Delta_{\acc}:=\Pr[\acc|\rho]-\Pr[\acc|\tau]\,
\end{align}
and use the compact notation
\begin{align}
\alpha_{\acc}&:=\big(\Pr[\acc|\rho]
\left(\rho_{\acc} - \tau_{\rm acc}\right)+\Delta_{\acc}\tau_{\acc}
\big)\otimes\proj{\acc}\,\\
    \alpha_{\rej}&:=
    \big(\Pr[\rej|\rho]\rho_{\rej}-\Pr[\rej|\tau]\tau_{\rej}
    \big)\otimes \proj{\rej}\ ,
\end{align}
such that $\cE(\rho)- \cE(\tau) = \alpha_\acc + \alpha_\rej$.
By expressing the norm $\|\cE(\rho)-\cE(\tau)\|_1$ in Eq.~\eqref{contractivity_step} in variational form (cf.\ Eq.~\eqref{eq:variational trace dist}) and using that $\|\!-\!A\|_1=\|A\|_1$, we have
\begin{align}
    \delta&\ge \left\|(-1)^\sigma( \alpha_{\acc} +\alpha_{\rej}) \right\|_1\\&=2\max_{0\le\Pi\le I} \tr \big(\Pi (-1)^\sigma\left(\alpha_{\acc} +\alpha_{\rej}   \right) \big)\qquad\textrm{ for any $\sigma\in \{0,1\}$}\ .
    \end{align}
Let us restrict to POVM elements of the form $\Pi=\Pi'\otimes \proj{\acc}+\Pi''\otimes \proj{\rej}$ where~$\Pi', \Pi''\in \cB(\cH')$ are arbitrary POVM elements. Then,
\begin{align}
   2\max_{0\le\Pi\le I} \tr \big(\Pi(-1)^\sigma \left(\alpha_{\acc} +\alpha_{\rej}   \right) \big)&\ge 2\max_{\substack{0\le\Pi'\le I\\0\le\Pi''\le I}} \tr \Big(\big(\Pi'\otimes \proj{\acc}+\Pi''\otimes \proj{\rej}\big) (-1)^\sigma\big(\alpha_{\acc} +\alpha_{\rej}   \big) \Big)\notag\\
    &= 2\max_{\substack{0\le\Pi'\le I\\0\le\Pi''\le I} }(-1)^\sigma\tr \big((\Pi'\otimes I) \alpha_{\acc}+(\Pi''\otimes I)\alpha_{\rej} \big) \label{eq:split proj orthogonal}\\
    &\ge 2\max_{0\le\Pi'\le I} (-1)^\sigma\tr \big((\Pi'\otimes I)\alpha_{\acc}\big) \label{eq: restrict to acc}\\
    &=2\max_{0\le\Pi'\le I} (-1)^\sigma\tr \Big(\Pi'\big({\Pr[\acc|\rho]}(\rho_{\acc} - \tau_{\acc})+\Delta_{\acc}\tau_{\acc}\big)\Big)
\label{there exists one}\ ,
\end{align}
where Eq.~\eqref{eq:split proj orthogonal} is a consequence of the orthogonality of the classical states $\proj{\acc}$ and $\proj{\rej}$, Eq.~\eqref{eq: restrict to acc} can be seen setting $\Pi'' = 0$, and Eq.~\eqref{there exists one} is obtained by tracing out the classical register.

Choose $\sigma\in \{0,1\}$ such that
$ (-1)^\sigma\Delta_{\acc}\ge0$, i.e., $(-1)^\sigma$ is the sign of $\Delta_\acc$. Then positivity of $\tr(\Pi' \tau_{\acc})$ and Eq.~\eqref{there exists one} implies that
\begin{align}
\delta&\ge2\max_{0\le\Pi'\le I}(-1)^\sigma  \tr \Big({\Pr[\acc|\rho]}\cdot\Pi'(\rho_{\acc} - \tau_{\acc}) \Big) \\
&=2\Pr[\acc|\rho]\cdot  \max_{0\le\Pi'\le I} \tr \Big(\Pi'(-1)^\sigma(\rho_{\acc} - \tau_{\acc}) \Big)\\
&=\Pr[\acc|\rho]\cdot  \|(-1)^\sigma( \rho_{\acc} - \tau_{\acc}) \|_1\\
&=\Pr[\acc|\rho]\cdot  \| \rho_{\acc} - \tau_{\acc} \|_1
\ ,
\end{align}
where we used the linearity of the trace in the first identity, the variational form of the norm~$\|\cdot\|_1$ to obtain the second identity, and the fact $\|\!-\!A\|_1=\|A\|_1$ to reach the last identity.  

By the triangle inequality and by the assumption~\eqref{it:pointwise closeness}, we conclude that
\begin{align}
\left\| \tau_{\acc}-\sigma_{\mathsf{targ}}\right\|_1&\le
\left\| \tau_{\acc}-\rho_{\acc}\right\|_1+\left\| \rho_{\acc}-\sigma_{\mathsf{targ}}\right\|_1
\\
&\le\frac{\delta}{\Pr[\acc|\rho]}+\gamma\ ,\label{eq: eq: min tr eq}
\end{align}
as claimed.
\end{proof}

\subsubsection{Completing the proof of Theorem~\ref{thm:envelope gaussification} \label{sec: complete proof env gaussification}}
We now combine the
results on 
the Protocol~\ref{prot:envelopeshaping} with the truncated comb state~$\ket{\Sha^\varepsilon_L}$ as input (obtained in Section~\ref{sec:shastateinputanalysis})
with Lemma~\ref{lem: heralded channel appx states dist pacc} from Section~\ref{sec: heralded channels} about heralding channels. This gives our main result, Theorem~\ref{thm:envelope gaussification}, establishing that 
the envelope-shaping protocol also produces an approximate GKP state if the input state is only close to a comb state.

\begin{proof}[Proof of Theorem~\ref{thm:envelope gaussification}]:
We consider Protocol~\ref{prot:envelopeshaping} with input state~$\rho=\proj{\Sha^\varepsilon_{L,\Delta}}$. Let us denote the output state conditioned on acceptance by $\rho_\acc=\rho_{|\Omega_L}(\varepsilon,\kappa,\Delta,L)$.
By Lemma~\ref{lem: env shaping rhoacc gkp trace dist} applied with $\varepsilon = \sqrt{\Delta}$, this output state is close to the target state~$\sigma_{\mathsf{targ}}=\proj{\gkp_{\kappa,\Delta}}$, that is, we have 
\begin{align}
        \left\|\rho_\acc-\proj{\gkp_{\kappa,\Delta}}\right|_1 \le 6\kappa\sqrt{L} +6\sqrt{\Delta}+ 7e^{-\kappa^2 L^2/64}\ .\label{eq:dstbnd}
    \end{align}
Moreover, by Lemma~\ref{lem: prot 2 accept prob}, we have
\begin{align}
    \Pr\Big[\textnormal{Protocol~\ref{prot:envelopeshaping} } \accepts \mid \rho
    \Big]&\ge \frac18\left(1 - 2e^{-\kappa^2L^2/256}\right)\ .
\end{align}
    
By assumption, $\tau\in\cB(L^2(\mathbb{R}))$ is a state that is $\xi$-close to~$\proj{\Sha_{L, \Delta}}$. By Corollary~\ref{cor: Sha - Sha epsilon trace distance}, we have
    \begin{align}
        \left\| \proj{ \Sha_{L,\Delta}} - \proj{\Sha_{L,\Delta}^{\varepsilon}}\right\|_1
        &\le 5 \sqrt{\Delta}\ 
    \end{align}
   and hence the triangle inequality implies that
    \begin{align}
        \left\| \tau - \proj{\Sha_{L,\Delta}^{\varepsilon}}\right\|_1
        &\le 5 \sqrt{\Delta} + \xi\ .\label{eq:deltadefinitiondeltaxi}
    \end{align}

We apply Lemma~\ref{lem: heralded channel appx states dist pacc} with parameters
\begin{align}
    \delta&=5 \sqrt{\Delta}+\xi\qquad\textrm{ (cf.~\eqref{eq:deltadefinitiondeltaxi})}\\
    \gamma&=6\kappa\sqrt{L} + 6\sqrt{\Delta}+ 7e^{-\kappa^2 L^2/64}\ ,\qquad\textrm{(cf.~\eqref{eq:dstbnd})}.
\end{align}
We 
conclude that on input~$\tau$, 
Protocol~\ref{prot:envelopeshaping} accepts 
with probability at least 
\begin{align}
    \Pr \Big[\textnormal{Protocol~\ref{prot:envelopeshaping} } \accepts \mid \tau
    \Big]\ge \frac18\left(1 - 2e^{-\kappa^2L^2/256}\right)-\frac52\sqrt{\Delta}-\frac\xi 2\ .
\end{align}
Furthermore, conditioned on acceptance, the output state~$\tau_{\acc}$
is close to the target state $\proj{\Sha_{L,\Delta}}$, i.e.,
\begin{align}
    \left\|\tau_\acc-\proj{\gkp_{\kappa,\Delta}}\right\|_1
    &\le \frac{5\sqrt{\Delta}+\xi}{\frac18(1 - 2e^{-\kappa^2L^2/256})}+6\kappa\sqrt{L} + 6\sqrt{\Delta} + 7e^{-\kappa^2 L^2/64}\ .
\end{align}
This is the claim. (We have already discussed the complexity, i.e., number of operators from the set~$\cG$ of the circuit implementing this protocol.)
\end{proof}

\section{Approximate GKP-state preparation}\label{sec: gkp state prep}

In this section, we present our main protocol for preparing approximate GKP states. We show that given parameters $(\kappa, \Delta)$, this protocol accepts with probability at least
\begin{align}
   \Pr[\acc]\ge \frac{1}{10}\ ,
\end{align} 
in which case it prepares a quantum state $\tau_\acc\in \cB(L^2(\mathbb{R}))$ satisfying
\begin{align}
         \left\| \tau_\acc - \proj{\gkp_{\kappa,\Delta}} \right\|_1 \le   O(\sqrt{\Delta}) + O(\kappa^{1/3})\ ;
\end{align}
or it rejects. The protocol is efficient --- it uses only a linear number of  operations in $(\log1/\kappa,1/\Delta)$ from the set $\cG$ (cf.~\ref{subsec: complexity} \eqref{it:constantstrengthgaussianunitaries} to~\eqref{it:constriesantqubitcontrolled}).

The protocol works in two stages. First, it creates a comb state (using the comb-state-preparation protocol in Section~\ref{sec: comb state prep}). Second, it shapes the prepared comb state by a Gaussian envelope (using the envelope-Gaussification protocol in Section~\ref{sec: Gaussian envelope shaping}).

\begin{algorithm}[H] 
	\caption{Approximate GKP-state-preparation protocol}
     \label{prot: appx GKP state prep}
	\begin{flushleft}
		\textbf{Input:} Parameters $\kappa\in (0,1/4)$ and $\Delta\in (0,1/4)$.\\
		\textbf{Output:} Either \textsf{accept} or \textsf{reject}; and in the case of acceptance, a state of a single mode that is close to the state $\ket{\gkp_{\kappa,\Delta}}$ state, see Theorem~\ref{thm:main}.
		\begin{algorithmic}[1]
            \State Apply the comb-state-preparation Protocol~\ref{prot: comb state prep} with input $\Delta$
            and $n  = \lfloor \frac{4}{3} \log_2 1/\kappa\rfloor$. This results in a state close to the approximate comb state $\proj{\Sha_{L,\Delta}}$ where $L=2^n$.
            \State Apply the envelope-shaping Protocol~\ref{prot:envelopeshaping} with input state $ 
            \tau 
            $ and with parameters $\kappa$ and $L$.
            If the Protocol~\ref{prot:envelopeshaping} accepts, \Return $\mathsf{accept}$ and the single-mode state that it produced.
            \State \Return \textsf{reject} otherwise.
		\end{algorithmic}
	\end{flushleft}
\end{algorithm}

\begin{restatable}{theorem}{thmmain}\label{thm:main} 
There are constants~$c_1,c_2>0$ such that the following holds. 
Given inputs $\kappa, \Delta \in(0,10^{-6})$, the output state of Protocol~\ref{prot: appx GKP state prep}
conditioned on acceptance is close to $\ket{\gkp_{\kappa,\Delta}}$. The protocol accepts with probability at least
\begin{align}
    \Pr[\mathsf{acc}]\ge \frac{1}{10}\ .
\end{align}
The output state $\rho \in \cB(L^2(\mathbb{R}))$ conditioned on acceptance satisfies
\begin{align}
    \left\| \rho - \proj{\gkp_{\kappa,\Delta}} \right\|_1 \le  190 \sqrt{\Delta} + 24\kappa^{1/3}\ .
\end{align}
Furthermore, the protocol requires fewer than $c_1\log1/\kappa+c_2\log1/\Delta$ elementary operations (see Section~\ref{sec: allowed operations}).
\end{restatable}

\begin{figure}[!ht]
	\begin{center}
  \includegraphics{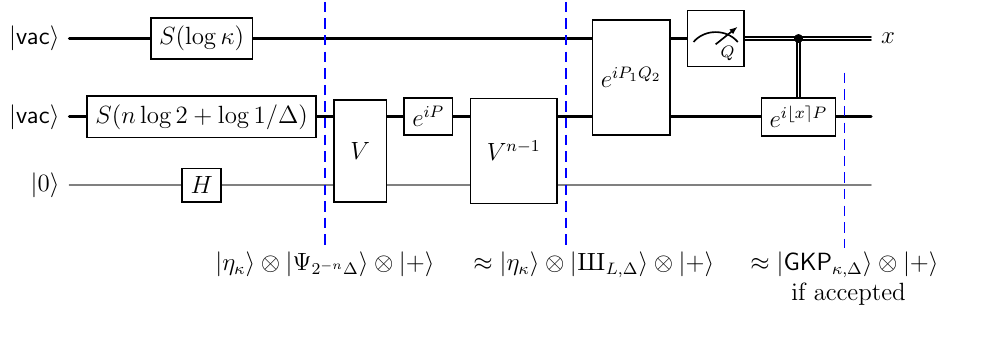}
  \vspace*{-1.2cm}
\end{center}
\caption{Circuit for heralded preparation of approximate GKP state. 
 The unitary $V$ is defined in Fig.~\ref{fig: circuit V}. The exponent above the gate
 $V$ indicates the number of repetitions.
 Again, both the squeezing operations and the classically controlled displacements are implemented as a product of their constant-strength equivalents;
see the discussion surrounding Eq.~\eqref{eq:squeezingdecomposition}.
}\label{fig:GKP-preparation-protocol}
\end{figure}

\begin{proof}
    As described in Protocol~\ref{prot: appx GKP state prep}  we choose parameters 
    \begin{align}
        n=\left\lfloor\frac43 \log_2 1/\kappa\right\rfloor \qquad \textrm{ and }\qquad
        L= 2^n\ . \label{eq: fix parameters}
    \end{align}

    First, Protocol~\ref{prot: appx GKP state prep} 
    prepares a state close to the comb state $\ket{\Sha_{L,\Delta}}$ using the comb-state-preparation protocol (Protocol~\ref{prot: comb state prep}). 
    The comb-state-preparation protocol is run with input parameters $\Delta$ and number of rounds $n=\log_2 L$.
    By Theorem~\ref{thm:comb state prep} this protocol (deterministically) prepares a quantum state $\tau \in \cB(L^2(\mathbb{R}))$ such that
    \begin{align}
         \left\| \tau - \proj{\Sha_{L,\Delta}} \right\|_1 \le 17\sqrt{\Delta}\ .
    \end{align}

Subsequently, Protocol~\ref{prot: appx GKP state prep} applies the envelope-Gaussification protocol (Protocol~\ref{prot:envelopeshaping}) with parameters $\kappa,L$ to the state $\tau$.
By Theorem~\ref{thm:envelope gaussification} (setting $\xi = 17\sqrt{\Delta})$ the corresponding output state $\tau_\acc$ conditioned on acceptance in Protocol~\ref{prot:envelopeshaping} satisfies 
\begin{align}
    \left\|\tau_\acc-\proj{\gkp_{\kappa,\Delta}}\right\|_1\le \left(\frac{22}{\frac18\left(1 - 2e^{-\kappa^2L^2/256}\right)}+6
    \right)\sqrt{\Delta}+6\kappa\sqrt{L} + 7e^{-\kappa^2 L^2/64}\ ,
    \label{eq: rho gkp closeness}
\end{align}
With the chosen parameters, we have  $\kappa L \ge \kappa \cdot 2^{\frac{4}{3} (\log_2 1/\kappa)-1}=\kappa/(2{\kappa^{\frac{4}{3}}})=\kappa^{-1/3}/2$ and by the assumption $0<\kappa< 10^{-6}$, we have
    \begin{align}
        \frac18 \left(1-2e^{-\kappa^2L^2/256}\right)\ge \frac{12}{100} \ . \label{eq: bound prob frac}
    \end{align}
Therefore, using this bound in Eq.~\eqref{eq: rho gkp closeness} gives
\begin{align}
    \left\|\tau_\acc-\proj{\gkp_{\kappa,\Delta}}\right\|_1 &\le 190\sqrt{\Delta} +6\kappa\sqrt{L} + 7e^{-\kappa^2 L^2/64}\\ 
    &\le  190 \sqrt{\Delta} + 6 \kappa^{1/3} + 7e^{-\kappa^2 L^2/64} \qquad \textrm{ by $\kappa \sqrt{L}\le \kappa^{1/3}$\,,}\\
    &\le  190 \sqrt{\Delta} + 6 \kappa^{1/3} + 7e^{-\kappa^{-2/3}/256} \qquad \textrm{ by $\kappa L \ge  \kappa^{-1/3}/2$\,,}\\
    &\le 190 \sqrt{\Delta} + 6 \kappa^{1/3}+ 1792\kappa^{2/3} \qquad \textrm{ by $e^{-x} \le x^{-1}$ for $x>0$\,,}\\
    &\le 190 \sqrt{\Delta} + 6 \kappa^{1/3} + 18 \kappa^{1/3} \quad \textrm{by the assumption $\kappa < 10^{-6}$\,,}\\
    &= 190 \sqrt{\Delta} + 24 \kappa^{1/3}\ .
\end{align}

Furthermore, the Protocol~\ref{prot:envelopeshaping} accepts with probability at least
\begin{align}
\Pr\Big[\textnormal{Protocol~\ref{prot:envelopeshaping} } \accepts \mid \tau
    \Big] \ge \frac1 8\left(1 - 2e^{-\kappa^2L^2/256}\right)- 11\sqrt{\Delta}\ . \label{eq: rho accept prob}
\end{align}
By inserting the bound~\eqref{eq: bound prob frac} into~\eqref{eq: rho accept prob} and using the assumption $0<\Delta<10^{-6}$, we obtain
\begin{align}
     \Pr\Big[\textnormal{Protocol~\ref{prot:envelopeshaping} } \accepts \mid \tau
    \Big]\ge  \frac{12}{100}-11\sqrt{\Delta} \ge \frac{1}{10}\ .
\end{align}
Finally, we analyse the circuit complexity of Protocol~\ref{prot: appx GKP state prep}. Since this protocol is simply the composition of Protocol~\ref{prot: comb state prep} and Protocol~\ref{prot:envelopeshaping}, the 
number of operations used is the sum of the corresponding numbers for these protocols.
By Theorem~\ref{thm:comb state prep}, we have that Protocol~\ref{prot: comb state prep} uses $5n + \ceil{\log 1/\Delta}+5$ elementary operations and by Theorem~\ref{thm:envelope gaussification}, we have that Protocol~\ref{prot:envelopeshaping} uses at most $b_1 \log L + b_2\log 1/\kappa$ operations for some constants $b_1, b_2>0$. Since the parameters in Protocol~\ref{prot: appx GKP state prep} are fixed as in Eq.~\eqref{eq: fix parameters}, i.e., $n = \ceil{(4\log_2 \kappa)/3}$ and $\log L = 2^n$, we can find constants $c_1$ and $c_2$ such that in Protocol~\ref{prot: appx GKP state prep} the total number of elementary operations used is upper bounded by $c_1 \log 1/\kappa + c_2 \log1/\Delta$.
\end{proof} 
Theorem~\ref{thm:main} directly implies the following asymptotic statement about the heralded complexity of approximate GKP states. 
\begin{corollary}\label{cor: upper bound her state complexity}
    There is a polynomial $q(\kappa,\Delta)$ with $q(0,0)=0$ such that for all functions~$\varepsilon(\kappa,\Delta)$ and $p(\kappa,\Delta)$ satisfying
    $p(\kappa,\Delta) \in [0,1/10]$ and $\varepsilon(\kappa,\Delta) \ge q(\kappa,\Delta)$ for sufficiently small $(\kappa,\Delta)$, we have
    \begin{align}
        C_{p(\kappa,\Delta), \varepsilon(\kappa,\Delta)}^{*,\mathsf{her}}(\ket{\gkp_{\kappa,\Delta}}) \le O(\log1/\kappa + \log1/\Delta) \qquad\textrm{for}\qquad(\kappa,\Delta)\rightarrow(0,0)\, .
    \end{align}
\end{corollary}

It is useful to phrase Theorem~\ref{thm:main} in terms of a single parameter. 
\begin{corollary}
    Let $N \in \mathbb{N}$ be sufficiently large and assume $\kappa = \mathsf{poly}(1/N)$ and $\Delta = \mathsf{poly}(1/N)$. Then Protocol~\ref{prot: appx GKP state prep} prepares  a state $\rho\in\cB(L^2(\mathbb{R}))$ such that
     \begin{align}
         \left\lVert \rho - \proj{\gkp_{\kappa,\Delta}} \right\rVert_1 \le  \mathsf{poly}(1/N)\,, 
    \end{align}
    with probability at least $1/10$.
    The protocol uses $O(\log N)$ elementary operations.
\end{corollary}

\section{Lower bounds on the complexity of approximate GKP states\label{sec: converse bound}}
In this section, we establish lower bounds on the unitary and heralded complexity of the state~$\ket{\gkp_{\kappa,\Delta}}$. We proceed as follows. In Section~\ref{sec:energylowerbound}
we establish an upper bound on
the energy of a state~$U\ket{\Psi}$ produced by 
applying a circuit~$U=U_T\cdots U_1$ 
with a limited number~$T$ of gates from the set~$\cG$. Here~$\ket{\Psi}=\ket{\vac}\otimes\ket{\vac}^{\otimes m}\otimes\ket{0}^{\otimes m'}$. Concretely, we show that 
\begin{align}
\langle U\Psi,H U\Psi\rangle\leq e^{8\pi T}(m+2)=:E_m(T)\ \label{eq:etupperboundexplicit}
\end{align}
where $H=\sum_{k=1}^{m+1} (Q_k^2+P_k^2)$, see Lemma~\ref{lem: moment limits unitary state prep}.  This is an immediate consequence of the fact that the unitary operations constituting the set~$\cG$ are moment-limited. 

Eq.~\eqref{eq:etupperboundexplicit} implies that the reduced density operator~$\rho=\tr_{m,m'} U\proj{\Psi}U^\dagger$ on the first mode has most of its support on the subspace spanned by functions with support on a bounded interval~$[-R,R]$ in position-space, i.e., we have
\begin{align}
\tr(\Pi_{[-R,R]}\rho)&\geq 1-E_m(T)/R^2\qquad\textrm{ for }\qquad R>0\ .\label{eq:lowerboundpir}
    \end{align}
Here $\Pi_{[-R,R]}$ denotes the orthogonal projection onto the subspace of~$L^2(\mathbb{R})$ of functions whose support is contained in the interval~$[-R,R]$.  Eq.~\eqref{eq:lowerboundpir} is a direct consequence of Markov's inequality, see Lemma~\ref{lem:markovinequalityprojection}. A bound analogous to~\eqref{eq:lowerboundpir} applies to the orthogonal projection~$\widehat{\Pi}_{[-R,R]}$ onto the subspace of functions whose Fourier transform has support on~$[-R,R]$.

Finally, we show that a bound
of the form~\eqref{eq:lowerboundpir}
immediately implies a lower bound on the distance of~$\rho$ to the state~$\proj{\gkp_{\kappa,\Delta}}$, 
 see Lemma~\ref{lem:distancetogkpprojection}. This is because the norm $\|\Pi_{[-R,R]} \ket{\gkp_{\kappa,\Delta}}\|$ of the projected state is bounded (for suitably chosen $R$), and the same holds for the projection~$\widehat{\Pi}_{[-R,R]}$ in momentum space. We establish corresponding tail bounds in 
 Section~\ref{sec:concentrationboundgkp}. 
 Combined, these results yield a lower bound on the unitary complexity of~$\ket{\gkp_{\kappa,\Delta}}$, see Section~\ref{sec:unitarycomplexitygkp}. We then extend these arguments to the heralded state complexity in Section~\ref{sec: lower bound heralded state complexity}.

\subsection{Moment limits on the gate set~$\cG$}
Here, we argue that each unitary~$U\in\cG$ in our gate set~$\cG$ is moment-limited, i.e., it cannot significantly increase the norm of the displacement vector (i.e., the first moments) of a state, nor the energy (a sum of second moments). Such moment bounds are well-known~\cite{winter2017energyconstraineddiamondnormapplications, 2018_Shirokov} and widely used. For example, denoting
the Hamiltonian of $n$ independent harmonic oscillators by~$H=\sum_{j=1}^{2n}R_j^2$ where $R= (Q_1, P_1,\dots, Q_{n},P_{n})$,
phase-space displacements, rotations, beamsplitters and single-mode squeezers satisfy (see e.g. \cite[p. 30]{dias2024classicalsimulationnongaussianbosonic})
\begin{align}
\begin{aligned}
\operatorname{tr}\left(H e^{-ia R_i} \rho \e^{i a R_i} \right)&\le \operatorname{tr}(H \rho)+ (2\pi)^2 +4\pi \|s(\rho)\| \label{eq: energy bound displacement}  & &  \text { for all } & & a \in [-2\pi,2\pi] \\
\operatorname{tr}\left(H P_j(\phi)^{\dagger} \rho P_j(\phi)\right) & =\operatorname{tr}(H \rho) & & \text { for all } & & \phi \in [-2\pi, 2\pi] \\
\operatorname{tr}\left(H B_{j, k}(\omega)^{\dagger} \rho B_{j, k}(\omega)\right) & =\operatorname{tr}(H \rho) & & \text { for all } & & \omega \in [-2\pi,2\pi] \\
\operatorname{tr}\left(H S_j(z)^{\dagger} \rho S_j(z)\right) & \leq e^{4\pi} \operatorname{tr}(H \rho) & & \text { for all } & & z\in[-2\pi,2\pi]\,
\end{aligned}
\end{align}
when applied to a state~$\rho\in\cB(L^2(\mathbb{R})^{\otimes n})$ with finite first and second moments. Here $s(\rho) \in \mathbb{R}^{2n}$ is defined by its entries $s(\rho)_j = \tr(R_j \rho)$.

For completeness, we establish analogous bounds for our gate set~$\cG$, whose unitaries act on a Hilbert space of the form~$L^2(\mathbb{R})^{\otimes n}\otimes (\mathbb{C}^2)^{\otimes n'}$.

\begin{lemma}[Moment-limit on phase-space displacements]\label{lem:momentlimitphasespacedisplacements}
Consider an $n$-mode bosonic system with (vector of) mode operators~$R=(Q_1,P_1,\ldots,Q_n,P_n)$. 
For $d=(d^Q,d^P)\in\mathbb{R}^n\times\mathbb{R}^n\cong \mathbb{R}^{2n}$, let 
$D(d)=e^{i\sum_{j=1}^n (d^Q_j Q_j-d^P_jP_j)}$ be the displacement operator in the direction~$d$. 
Let $\rho\in\cB(L^2(\mathbb{R}))$ be a state with finite first and second moments. Then, 
\begin{align}
\tr(H D(d)\rho D(d)^\dagger) &\leq \tr(H\rho) + 2 \|d\|\cdot \|s(\rho)\|+\|d\|^2\label{eq:secondmomentsHeffect}
\end{align}
where $\|d\|=\sqrt{\sum_{j=1}^{2n} d_j^2}$ denotes the Euclidean norm of~$d$ and where $s(\rho)\in\mathbb{R}^{2n}$ is the displacement vector of~$\rho$ defined by its entries~$s(\rho)_j=\tr(R_j\rho)$.  Furthermore,
 the Euclidean norm of the displacement of the resulting state is bounded as
 \begin{align}
     \left\|s(D(d)\rho D(d)^\dagger)\right\|&\leq 
     \|s(\rho)\|+\|d\|\ .\label{eq:upperboundnormdisplacmeentm}
 \end{align}
\end{lemma}
By definition of~$\cG$, Lemma~\ref{lem:momentlimitphasespacedisplacements} implies that for any single-mode displacement $D(d)\in \cG$, we have 
\begin{align}
\renewcommand*{\arraystretch}{1.225}
\begin{matrix*}[l]
\left\|s(D(d)\rho D(d)^\dagger)\right\|&\leq & \|s(\rho)\| + 2\pi\\  
\tr(H D(d) \rho D(d)^\dagger) &\leq& \tr(H\rho)+ 4\pi \|s(\rho)\|+(2\pi)^2\ .
\end{matrix*}\label{eq:displacementsingboundmoment}
\end{align}

\begin{proof}
Since the displacement operator~$D(d)$ acts on mode operators as
\begin{align}
\renewcommand*{\arraystretch}{1.225}
\begin{matrix*}[l]
D(d)^\dagger Q_j D(d)&=& Q_j+d^Q_j I\\
D(d)^\dagger P_j D(d)&=& P_j+d^P_j I
\end{matrix*}\qquad\textrm{ for all }\qquad j\in [n]\ ,\label{eq:displacementactionqrpj}
\end{align}
we conclude that the displacement vector of the state~$D(d)\rho D(d)^\dagger$ is equal to
\begin{align}
    s(D(d)\rho D(d)^\dagger)&=s(\rho)+d\ .
\end{align}
Eq.\eqref{eq:upperboundnormdisplacmeentm} immediately follows from this using the triangle inequality for the Euclidean norm. 

Again using~\eqref{eq:displacementactionqrpj}, we also have 
\begin{align}
D(d)^\dagger H D(d)&=\sum_{j=1}^{2n} (R_j+d_jI)^2\\
&=\sum_{j=1}^{2n} R_j^2 +2\sum_{j=1}^{2n} d_j R_j+\left(\sum_{j=1}^{2n} d_j^2\right)\cdot I\\
&=H+2\sum_{j=1}^{2n} d_j R_j+\|d\|^2\cdot I\ .
\end{align}
This implies that
\begin{align}
\tr( HD(d)\rho D(d)^\dagger)&=\tr(D(d)^\dagger H D(d)\rho)\\
&=\tr(H\rho)+2 s(\rho)^T d+\|d\|^2\\
&\leq \tr(H\rho) +2 \|s(\rho)\|\cdot \|d\|+\|d\|^2
\end{align}
by the Cauchy-Schwarz inequality. This is the claim in Eq.~\eqref{eq:secondmomentsHeffect}.
\end{proof}

\begin{lemma}[Moment-limit on quadratic Gaussian unitaries]\label{lem:momentlimitquadraticunitaries}
Consider an $n$-mode bosonic system with (vector of) mode operators~$R=(Q_1,P_1,\ldots,Q_n,P_n)$. 
Let $H=\sum_{j=1}^n (Q_j^2+P_j^2)=\sum_{k=1}^{2n} R_k^2$. 
Let $A=A^T\in\mathsf{Mat}_{2n\times 2n}(\mathbb{R})$ be a symmetric matrix, and let $U(A)=e^{iH(A)}$ be the Gaussian unitary defined in terms of the Hamiltonian~$H(A)=\frac{1}{2}R^TA R$.
Let $\rho\in\cB(L^2(\mathbb{R}))$ be a state with finite first and second moments. Then 
the norms of the displacement vectors~$s(\rho),s(U(A)\rho U(A)^\dagger)\in\mathbb{R}^{2n}$ are related by 
\begin{align}
    \left\|s(U(A)\rho U(A)^\dagger)\right\|&\leq 
e^{\|A\|}\cdot \|s(\rho)\|\ . \label{eq: limitdisplacement}
\end{align} 
Furthermore, we have 
\begin{align}
\tr(H U(A) \rho U(A)^\dagger) &\leq e^{2\|A\|}\tr(H\rho)\ .\label{eq:momentlimithua}
\end{align}
\end{lemma}
By definition of~$\cG$, Lemma~\ref{lem:momentlimitquadraticunitaries} means that for any Gaussian unitary~$U=U(A)\in\cG$, we have the inequalities
\begin{align}
\left\|s(U(A)\rho U(A)^\dagger)\right\|&\leq 
e^{2\pi}\|s(\rho)\|\\
\tr(H U(A) \rho U(A)^\dagger) &\leq e^{4\pi}\tr(H\rho)\ .
\label{eq:momentlimitgaussianquadratic}
\end{align}

\begin{proof}
Let $J\in \mathsf{Mat}_{2n\times 2n}(\mathbb{R})$ be the symplectic form associated with the $n$-mode system, and let $S(A)=e^{A J}\in \mathrm{Sp}(2 n)$
be the symplectic matrix describing the action of~$U(A)$, see Eq.~\eqref{eq:symplecticactionunitarygaussian}.
Since
\begin{align}
U(A)^\dagger R_j U(A)&=\sum_{k=1}^{2n} S(A)_{j,k}R_k\ ,
    \end{align}
 the displacement vector of the state~$U(A)\rho U(A)^\dagger$ is given by matrix-vector multiplication, i.e.,
\begin{align}
s(U(A)\rho U(A)^\dagger)&=S(A)s(\rho)\ .
    \end{align}
    It follows that 
    \begin{align}
        \|s(U(A)\rho U(A)^\dagger)\|&\leq 
        \|S(A)\|\cdot \|s(\rho)\|\ ,\label{eq:intermediatevectorinequalt}
    \end{align}
    where $\|S(A)\|$ denotes the operator norm of~$S(A)$.
By submultiplicativity of the operator norm, we have
\begin{align}
    \|S(A)\|&=\|e^{AJ}\|\leq e^{\|AJ\|}\leq e^{\|A\|}\ ,
\end{align}
where we use the fact that $\|J\|=1$. 
It follows from~\eqref{eq:intermediatevectorinequalt}
that
\begin{align}
    \|s(U(A)\rho U(A)^\dagger)\|&\leq e^{\|A\|}\|s(\rho)\|\ ,
\end{align}
as claimed in~\eqref{eq: limitdisplacement}. 
    
In the following, we show~\eqref{eq:momentlimithua}. Without loss of generality we assume that $\rho=\proj{\Psi}$ is a pure state. The general case follows by spectrally decomposing $\rho$.
For any symmetric matrix~$M=M^T\in\mathsf{Mat}_{2n\times 2n}(\mathbb{R})$, we also have 
\begin{align}
U(A)^\dagger H(M) U(A)&=H(S(A)^TMS(A))
\end{align}
by Eq.~\eqref{eq:symplecticactionunitarygaussian}. In particular, because $H=2H(I_{2n\times 2n})$, this means that
\begin{align}
U(A)^\dagger H U(A)&=2H(S(A)^TS(A))\ .
\end{align}
Observe that for any symmetric matrix~$M=M^T\in\mathsf{Mat}_{2n\times 2n}(\mathbb{R})$, we have
\begin{align}
\langle \Psi,H(M)\Psi\rangle &=\frac{1}{2}\sum_{j,k=1}^{2n} M_{j,k} G_{j,k}=\frac{1}{2}\tr(MG)\ ,
\end{align}
where  the Hermitian matrix~$G=G^\dagger\in\mathsf{Mat}_{2n\times 2n}(\mathbb{C})$ is defined by its entries
\begin{align}
G_{j,k}&= \langle R_j\Psi,R_k\Psi\rangle\qquad\textrm{ for }\quad j,k\in [2n]\ ,
\end{align}
and where we used the symmetry of~$M$ in the second identity.

Applied to the expression of interest with $H=2H(I_{2n\times 2n})$, we have 
\begin{align}
\tr(H\rho)&= \tr(G)\ \label{eq:traceexpressionhrho}
\end{align}
and
\begin{align}
\tr(H U(A) \rho U(A)^\dagger)&=\tr(U(A)^\dagger H U(A) \rho )\\
&=\tr(S(A)^TS(A)G)\ .\label{eq:operatorxmnm}
\end{align}
Since the operator norm is submultiplicative, we have 
\begin{align}
\|S(A)^T S(A)\|&= \|(e^{AJ})^T e^{AJ}\|\leq \|(e^{AJ})^T\|\cdot \|e^{AJ}\|\leq e^{\|(AJ)^T\|}\cdot e^{\|AJ\|}=e^{2\|AJ\|}\leq e^{2\|A\|}\ ,
\end{align}
where we used that $\|AJ\|\leq \|A\|\cdot\|J\|=\|A\|$ since $\|J\|=1$ 
for any symmetric matrix~$A=A^T\in\mathsf{Mat}_{2n\times 2n}(\mathbb{R})$. In particular, this means that we have the operator inequality
\begin{align}
S(A)^TS(A) \leq e^{2\|A\|}I_{2n\times 2n}\qquad\textrm{ for any symmetric matrix }A=A^T\in\mathsf{Mat}_{2n\times 2n}(\mathbb{R})\ .
\end{align}
Inserting this into~\eqref{eq:operatorxmnm} and combining with~\eqref{eq:traceexpressionhrho} imply the claim~\eqref{eq:momentlimithua}, because $G$ is a Gram matrix and thus positive semidefinite. 
\end{proof}

\begin{lemma}[Moment-limit on qubit-controlled displacements]\label{lem: mom limit qubit ctrl displacment}
Consider a system with Hilbert space~$L^2(\mathbb{R})^{\otimes n}\otimes (\mathbb{C}^2)^{\otimes n'}$.
Consider the result~$(\mathsf{ctrl}_jD(d))\rho(\mathsf{ctrl}_jD(d))^\dagger$ of applying a qubit-$j$-controlled displacement~$D(d)$, $d\in\mathbb{R}^{2n}$ to a state~$\rho\in\cB(L^2(\mathbb{R})\otimes (\mathbb{C}^2)^{\otimes n'})$ with finite first moments.
The  Euclidean norm of the displacement vector of this state is  bounded by
\begin{align}
\left\|s\left((\mathsf{ctrl}_jD(d))\rho(\mathsf{ctrl}_jD(d))^\dagger\right)\right\|&
    \leq \|s(\rho)\|+\|d\|\ .\label{eq:firstclaimcontrolleddisplacement} 
\end{align}
Furthermore, for any state~$\rho\in\cB(L^2(\mathbb{R})^{\otimes n}\otimes (\mathbb{C}^2)^{\otimes n'})$ with finite first and second moments, we have 
\begin{align}
\tr\left(H(\mathsf{ctrl}_jD(d))\rho (\mathsf{ctrl}_jD(d))^\dagger\right) &\leq 
2\tr(H\rho)+2\|d\|\cdot \|s(\rho)\|+\|d\|^2\ .
\end{align}
\end{lemma}
\begin{proof}
    We give the proof for $j=n'=1$ (the general case is analogous). Let $U$ be an arbitrary unitary on~$L^2(\mathbb{R})^{\otimes n}$, i.e., on the bosonic modes. Then, we have 
\begin{align}
(\mathsf{ctrl} U)^\dagger R_k(\mathsf{ctrl} U) &=
\proj{0}\otimes R_k+\proj{1}\otimes U^\dagger R_k U\qquad\textrm{ for }k\in [2n]\ .\label{eq:ctrlurkexpression}
\end{align}
In the case where $U=D(d)$ is a displacement, Eq.~\eqref{eq:ctrlurkexpression} specializes to
\begin{align}
    (\mathsf{ctrl} D(d))^\dagger R_k(\mathsf{ctrl} D(d)) &=\proj{0}\otimes R_k+\proj{1}\otimes (R_k+ d_kI)\\
&=I\otimes R_k+d_k I\otimes I\ ,
\end{align}
and we conclude that
\begin{align}
s\left((\mathsf{ctrl}_jD(d))\rho(\mathsf{ctrl}_jD(d))^\dagger\right)&=s(\rho)+d\ .
\end{align}
The claim~\eqref{eq:firstclaimcontrolleddisplacement} follows from the triangle inequality.
The second claim follows immediately from Lemma~\ref{lem:generalunitarysecondmoments}, which follows, applied to the unitary~$U=D(d)$, together with~ the moment bound~\eqref{eq:secondmomentsHeffect} for the unitary~$D(d)$.
\end{proof}
\begin{lemma}[Second moment limit on controlled-unitaries]\label{lem:generalunitarysecondmoments}
Suppose $U$ is a unitary on $L^2(\mathbb{R})^{\otimes n}$ with the following property: There are constants $a,b,c>0$ such that we have 
\begin{align}
\tr(H U\rho U^\dagger) &\leq a\tr(H\rho)+b \|s(\rho)\|+c \label{eq:assumptionrhobounded}
\end{align}
for any state~$\rho\in\cB(L^2(\mathbb{R})^{\otimes n})$ with finite first and second moments.
Consider a system with Hilbert space~$L^2(\mathbb{R})^{\otimes n}\otimes (\mathbb{C}^2)^{\otimes n'}$, and the controlled unitary
\begin{align}
\mathsf{ctrl}_{j}U&=\proj{0}_j\otimes I_{L^2(\mathbb{R})^{\otimes n}}+\proj{1}_j\otimes U\ 
\end{align}
where the bosonic unitary~$U$ is controlled by the $j$-th qubit, with $j\in [n']$.  (Here, we write $\proj{0}_j$ for the~$n'$-qubit operator~$I^{\otimes j-1}\otimes\proj{0}\otimes I^{\otimes n-j}$, and similarly for~$\proj{1}_j$).
Then, we have
\begin{align}
\tr(H (\mathsf{ctrl}_jU)\rho (\mathsf{ctrl}_jU)^\dagger) &\leq (a+1)\tr(H\rho)+b \|s(\rho)\|+c
\end{align}
for any state~$\rho\in\cB(L^2(\mathbb{R})^{\otimes n}\otimes (\mathbb{C}^2)^{\otimes n'})$ with finite first and second moments.
\end{lemma}
\begin{proof}
We give the proof for $j=n'=1$ (the general case is analogous). Then
we have
\begin{align}
(\mathsf{ctrl} U)^\dagger H(\mathsf{ctrl} U) &=\proj{0}\otimes H
+\proj{1}\otimes U^\dagger HU\\
&\leq I\otimes H+I\otimes U^\dagger HU\ ,
\end{align}
where we used the operator inequalities~$\proj{0}\leq I$ and $\proj{1}\leq I$ for a single qubit, and the fact that $H\geq 0$. 
 It follows that
 \begin{align}
 \tr \left(H(\mathsf{ctrl} U)\rho (\mathsf{ctrl} U)^\dagger)\right)
 &\leq \tr(H\rho')+\tr(U^\dagger H U\rho')\\
 &=\tr(H\rho')+\tr(H U\rho' U^\dagger)\ .
 \end{align}
where $\rho'=\tr_{n'}\rho$ denotes the reduced density operator of the bosonic modes after tracing out the qubits. The claim now follows from
the assumption~\eqref{eq:assumptionrhobounded} applied to the state~$\rho'$.
\end{proof}

Combining~\eqref{eq:displacementsingboundmoment},~\eqref{eq:momentlimitgaussianquadratic}, and Lemma~\ref{lem: mom limit qubit ctrl displacment}
  gives the following corollary.
  \begin{corollary}[Moment limits on gates from~$\cG$]\label{cor:momentlimitggate}
  Let $\rho\in\cB(L^2(\mathbb{R})^{\otimes n}\otimes (\mathbb{C}^2)^{\otimes n'})$ be a state of $n$~bosons and $n'$~qubits. Assume that~$\rho$ has finite first and second moments. Let $s(\rho)\in \mathbb{R}^{2n}$ denote its displacement vector. Then, 
  \begin{align}
\renewcommand*{\arraystretch}{1.225}
  \begin{matrix*}[l]
  \left\|s(U\rho U^\dagger)\right\|&\leq& e^{2\pi}\|s(\rho)\|+ 2\pi  
  \\
  \tr(HU\rho U^\dagger)&\leq & e^{4\pi} \tr(H\rho) + 4\pi \|s(\rho)\| + (2\pi)^2
  \end{matrix*}\qquad\textrm{ for every }U\in\cG\ .
  \end{align}
  \end{corollary}

\subsection{Moment limits on low-complexity states\label{sec:energylowerbound}}
Here, we argue that a state produced by a small circuit with  gates from~$\cG$ has small energy.

\begin{lemma} \label{lem: moment limits unitary state prep}
Consider a circuit~$U=U_T\cdots U_1$ 
 composed of $T$ unitaries $U_1,\ldots,U_T$ from  the set $\cG$
acting on the initial state $\ket{\Psi} = \ket{\vac}\otimes\ket{\vac}^{\otimes m}\otimes\ket{0}^{\otimes m'}$. We define the Hamiltonian 
\begin{align}
    H = \sum_{k=1}^{m+1} \left(Q_k^2 + P_k^2\right)\ .
\end{align}
Then, 
\begin{align}
\bra{\Psi}U^\dagger HU\ket{\Psi}&\leq 
e^{8\pi  T}(m+2):= E_m(T)\ . 
\end{align}
\end{lemma}
\begin{proof}
Let $R=(R_1,..., R_{2(m+1)}):=(Q_1,P_1,\ldots,Q_{(m+1)},P_{(m+1)})$ be the vector of mode operators, and let 
\begin{align}
s(\rho)&=\tr(R\rho)\in \mathbb{R}^{2(m+1)}
\end{align}
the displacement vector of a state $\rho$ on $L^2(\mathbb{R})^{\otimes (m+1)}\otimes (\mathbb{C}^2)^{\otimes m'}$. 
Denoting by~$\|\cdot\|$ the Euclidean norm in 
$\mathbb{R}^{2m}$, we have by Corollary~\ref{cor:momentlimitggate} that
\begin{align}
\|s(U_t\rho U_t^\dagger)\|&\leq \gamma \|s(\rho)\|+ \delta \label{eq:lowerboundm}
\end{align} with $\gamma =e^{2\pi}$ and $\delta = 2\pi$
for any unitary~$U_t\in\cG$ belonging to our gate set~$\cG$
and any state~$\rho$.
Setting $\rho^{(0)}=\proj{\Psi}$
and $\rho^{(t)}=U_t\rho^{(t-1)}U_t^\dagger$ for $t\geq 1$,
we can deduce from~\eqref{eq:lowerboundm} that
\begin{align}
\|s(\rho^{(t)})\|&\leq f(\|s(\rho^{(t-1)}\|)
\end{align}
where $f(s)=\gamma s+\delta$. In particular, defining $u_0:=\|s(\rho^{(0)})\|=0 $ and
$u_t:=f(u_{t-1})$ for $t\geq 1$, we get the upper bound
\begin{align}
\|s(\rho^{(t)})\|&\le u_t\\
&=f^{\circ t}(u_0)\\
&=\gamma^t u_0+\delta (1+\gamma+\cdots +\gamma^{t-1})\\
&= \delta \cdot \frac{\gamma^t-1}{\gamma-1}\qquad\textrm{ since } u_0 = 0\,,\\
&\le \delta \cdot \gamma^t\qquad\textrm{ since }  \gamma>2 \,,\\
&=2\pi \cdot e^{2\pi t}\ .
\label{eq: bound norm s(rho)}
\end{align}
Again using Corollary~\ref{cor:momentlimitggate}, we have for $1\le t \le T$ that 
\begin{align}
    \tr\left(H \rho^{(t)}\right) &\le  e^{4\pi} \tr\left( H \rho^{(t-1)} \right) +4\pi\|s(\rho^{(t)})\| + (2\pi)^2\\ 
    &\le e^{4\pi} \tr\left( H \rho^{(t-1)} \right) + 8\pi^2\cdot e^{2\pi t} +(2\pi)^2\\
    &\le e^{4\pi} \tr\left( H \rho^{(t-1)} \right) + 8\pi^2\cdot e^{2\pi T} +(2\pi)^2\ . 
\end{align}
Therefore, setting $A=e^{4\pi}$, $B= 8\pi^2\cdot e^{2\pi T} +(2\pi)^2$ and $x_t= \tr\left(H \rho^{(t)}\right) $, we can rephrase the previous bound as
\begin{align}
    x_t \le A x_{t-1} + B\ .
\end{align}
Proceeding similarly as in~\eqref{eq: bound norm s(rho)}, we have 
\begin{align}
    x_t &\le A^t x_0 + B(1 + A + \ldots +A^{t-1})\\
        &\le A^t x_0 + B A^t \qquad\textrm{as}\quad A >2\ .
\end{align}
With $t=T$, we have
\begin{align}
    \tr\left(H \rho^{(T)} \right) &\le e^{4\pi T}  \tr\left(H \rho^{(0)} \right) + (8\pi^2\cdot  e^{2\pi T} +(2\pi)^2)\cdot e^{4\pi T}\\
    &\le  e^{8\pi T}  \left(\tr\left(H \rho^{(0)} \right) +1\right)\ ,\label{eq: gate_lower_bound}
\end{align}where we used that $(8\pi^2\cdot  e^{2\pi T} +(2\pi)^2) \le e^{4\pi T}$ for $T\ge 1$. Obviously the bound~\eqref{eq: gate_lower_bound} is also valid in the case $T=0$.
The claim follows from $\bra{\mathsf{vac}}(Q^2 +P^2)\ket{\mathsf{vac}} = 1$, which gives $\tr\left(H \rho^{(0)}\right)= m+1$\ .
\end{proof}

The following lemma shows that an energy-limited state has most of its support on a subspace spanned by  functions with support in a bounded interval, both in position space and momentum space. Given a set $I\subseteq \mathbb{R}$, we define the operator $\Pi_I$ as the projection onto the subspace of $L^2(\mathbb{R)}$ of functions with support contained in $I$ in position space. 
In the following, we define the Fourier transformation as the operator
\begin{align}
    \begin{matrix}
        \cF & L^2(\mathbb{R}) & \rightarrow &  L^2(\mathbb{R}) \\
        & f &\mapsto& \cF(f)
    \end{matrix}
\end{align}
acting on a function $f\in L^1(\mathbb{R})\cap L^2(\mathbb{R})$ as
\begin{align}
    \cF(f)(p) = \widehat{f}(p) = \frac{1}{\sqrt{2\pi}} \int f(x) e^{-ipx} dx\ .\label{eq: def fourier}
\end{align}
Note that we have $\|f\|_2 =\|\widehat{f}\|_2$. Hence, the Fourier transformation is isometric on a dense subspace ($L^1(\mathbb{R})\cap L^2(\mathbb{R}$), and thus it uniquely extends to a unitary operator on $L^2(\mathbb{R})$.
In this context we define the operator $\widehat{\Pi}_{I} := \cF^\dagger \Pi_I \cF$, which is just the orthogonal projection onto the subspace of $L^2(\mathbb{R})$ of functions~$f$ whose Fourier transform~$\cF(f)$ has support contained in $I$.
\begin{lemma}\label{lem:markovinequalityprojection}
Let $\rho\in\cB(L^2(\mathbb{R}))$ be a state. Then, 
\begin{align}
\min\left\{\tr(\Pi_{[-R,R]}\rho),\tr(\widehat{\Pi}_{[-R,R]}\rho)\right\}&\geq 1-\frac{\tr((Q^2+P^2)\rho)}{R^2}
    \end{align}
    for any $R>0$. 
\end{lemma}
\begin{proof}
 Markov's inequality in the form~
 \begin{align}
     \Pr[Q^2>R^2]\leq \frac{\ExpE[Q^2]}{R^2}
 \end{align} implies that
\begin{align}
\tr\left((I-\Pi_{[-R,R]})\rho\right)&\leq \frac{\tr\left(Q^2\rho\right)}{R^2} \le \frac{\tr\left((Q^2+P^2)\rho\right)}{R^2} \, ,
\end{align} 
that is
\begin{align}
   \tr\left(\Pi_{[-R,R]}\rho\right)&\ge 1 - \frac{\tr\left((Q^2+P^2)\rho\right)}{R^2}
\label{eq: bound E(T) psi unitary}\ .
\end{align}
Interchanging the roles of $P$ and $Q$ yields the analogous bound for $\tr(\widehat{\Pi}_{[-R,R]}\rho)$.
\end{proof}

\subsection{Tail bounds for approximate GKP states\label{sec:concentrationboundgkp}}
In this section, we focus on the decay properties of approximate GKP states $\ket{\gkp_{\kappa,\Delta}}$ in position and momentum space.
We will repeatedly use the variational characterization 
\begin{align}
\|\rho-\sigma\|_1&=2\sup_{\Pi} \tr(\Pi(\rho-\sigma))\ ,\label{eq: variational L1}
\end{align}
of the $L^1$-distance of any two states~$\rho,\sigma\in\cB(L^2(\mathbb{R}))$, where the supremum is taken over all orthogonal projections~$\Pi$ 
acting on~$L^2(\mathbb{R})$.
\begin{lemma} \label{lem: proj pos space}
    Let $\kappa\in (0,1/4)$, $\Delta\in (0,1/100)$ and $R>0$. Then, 
    \begin{align}
        &\left\|\Pi_{[-R,R]}\ket{\gkp_{\kappa,\Delta}}\right\|^2 \le  4 \kappa R + 5\sqrt{\kappa
        } +  7\sqrt{\Delta} 
        \label{eq: gkp concentration pos}\\
&\left\| \widehat{\Pi}_{[-R, R]}\ket{\gkp_{\kappa,\Delta}} \right\|^2 \le 2\Delta R + 5\sqrt{\kappa}  + 7\sqrt{\Delta} 
\label{eq: gkp concentration mom}\ .
    \end{align}
\end{lemma}
\begin{proof}    
In the following, we fix $\varepsilon =\sqrt{\Delta}$. Note that $\varepsilon\in (0,1/2)$ by the assumption on~$\Delta$.

We first prove claim~\eqref{eq: gkp concentration pos}. We have
\begin{align}
\left\|\Pi_{[-R,R]}\ket{\gkp_{\kappa,\Delta}}\right\|^2 
 &=\bra{\gkp_{\kappa,\Delta}}\Pi_{[-R,R]}\ket{\gkp_{\kappa,\Delta}}\\
 &\le \bra{\gkp_{\kappa,\Delta}^\varepsilon}\Pi_{[-R,R]}\ket{\gkp_{\kappa,\Delta}^\varepsilon} \\
&\quad+\frac{1}{2}\left\| \proj{\gkp_{ \kappa, \Delta}} - \proj{\gkp_{ \kappa, \Delta}^\varepsilon}\right\|_1\ ,\label{eq: split porj_R epsilon}
\end{align}
where we used~\eqref{eq: variational L1}.
By  Lemma~\ref{lem: overlap gkp vareps proj pos}, we can bound the first term as 
\begin{align}
    \bra{\gkp_{\kappa,\Delta}^\varepsilon}\Pi_{[-R,R]}\ket{\gkp_{\kappa,\Delta}^\varepsilon} \le 4\kappa R + 10 \kappa \ .
\end{align}
Using Corollary~\ref{cor: gkp gkp eps trace distance}, the second term in Eq.~\eqref{eq: split porj_R epsilon}  (by the assumption $\varepsilon= \sqrt{\Delta}$) is bounded by
 \begin{align}
    \left\| \proj{\gkp_{ \kappa, \Delta}} - \proj{\gkp_{ \kappa, \Delta}^\varepsilon}\right\|_1 \leq 6 \sqrt{\Delta} \, .
\end{align}
It follows that
\begin{align}
    \left\|\Pi_{[-R,R]}\ket{\gkp_{\kappa,\Delta}}\right\|^2 &\le 4\kappa R + 10\kappa + 3\sqrt{\Delta}\\
    &\le  4\kappa R + 5\sqrt{\kappa} + 3\sqrt{\Delta}\\
    &\le 4\kappa R + 5\sqrt{\kappa} + 7\sqrt{\Delta}\, ,
\end{align}
where we used that $\kappa < 1/4$ to obtain the second inequality.

To show the claim~\eqref{eq: gkp concentration mom}, we define 
\begin{align}
    \ket{\gkp_{\kappa,\Delta}^{\widehat{\varepsilon}}} := \frac{\widehat{\Pi}_{(2\pi\mathbb{Z})(\varepsilon)} \ket{\gkp_{\kappa,\Delta}}}{\|\widehat{\Pi}_{(2\pi\mathbb{Z})(\varepsilon)} \ket{\gkp_{\kappa,\Delta}}\|}\, , \label{eq: def gkp vareps hat first}
\end{align} where we define $(2\pi\mathbb{Z})(\varepsilon) := 2\pi\mathbb{Z} + [-\varepsilon, \varepsilon]$.
This is to say that in momentum space, we have
\begin{align}
    \gkp_{\kappa,\Delta}^{\widehat{\varepsilon}}(p) =  D_{\kappa,\Delta}^{\varepsilon} \sum_{z\in\mathbb{Z}} \eta_\Delta(p) \chi_\kappa^\varepsilon(2\pi z)(p)\qquad\textrm{for}\qquad p \in \mathbb{R}\ , \label{eq: def gkp varhat sec}
\end{align}
where $D_{\kappa,\Delta}^{\varepsilon}$ is a normalization constant such that $\left\| \ket{\gkp_{\kappa,\Delta}^{\widehat{\varepsilon}}}\right\|=1$. We conclude with~\eqref{eq: variational L1} that
\begin{align}
    \left\| \widehat{\Pi}_{[-R, R]}\ket{\gkp_{\kappa,\Delta}} \right\|^2 &=\langle \gkp_{\kappa,\Delta}|\widehat{\Pi}_{[-R,R]}|\gkp_{\kappa, \Delta}\rangle\\
    &\le \langle \gkp_{\kappa,\Delta}^{\widehat{\varepsilon}}|\widehat{\Pi}_{[-R, R]}|\gkp_{\kappa,\Delta}^{\widehat{\varepsilon}} \rangle \\
    &\quad+ \frac{1}{2} \left\| |\gkp_{\kappa,\Delta}\rangle\langle\gkp_{\kappa,\Delta}|  - |\gkp_{\kappa,\Delta}^{\widehat{\varepsilon}}\rangle\langle\gkp_{\kappa,\Delta}^{\widehat{\varepsilon}}| \right\|_1 \ .\label{eq: proj gkp ineq omega eps}
\end{align}
We bound separately the two terms on the \rhs of~\eqref{eq: proj gkp ineq omega eps}, starting with the first. By Lemma~\ref{lem: Omega proj mom}, we have
\begin{align}
     \langle \gkp_{\kappa,\Delta}^{\widehat{\varepsilon}}|\widehat{\Pi}_{[-R, R]}|\gkp_{\kappa,\Delta}^{\widehat{\varepsilon}}\rangle \le 2\Delta R+ 12\Delta\ .\label{eq: proj gkp eps widehat first}
\end{align}
Moreover, from Corollary~\ref{cor: dist gkp omega eps}, we know that
\begin{align}
    \left\| \proj{\gkp_{\kappa,\Delta}}- \proj{\gkp_{\kappa,\Delta}^{\widehat{\varepsilon}}}  \right\|_1 \le  4\sqrt{\kappa} + 13\sqrt{\Delta}\ . \label{eq: proj gkp wide eps sec}
\end{align}
Combining~\eqref{eq: proj gkp eps widehat first} with~\eqref{eq: proj gkp wide eps sec} and using that by assumption $\sqrt{\Delta}\le 1/10$, we have $\Delta \le \sqrt{\Delta}/10$ yields
\begin{align}
    \left\| \widehat{\Pi}_{[-R, R]}\ket{\gkp_{\kappa,\Delta}} \right\|^2 &\le 2\Delta R + 12\Delta + 2\sqrt{\kappa} + 13\sqrt{\Delta}/2 \\
    &\le 2\Delta R + 2\sqrt{\kappa} + 7\sqrt{\Delta}\\
    &\le 2\Delta R + 5\sqrt{\kappa} + 7\sqrt{\Delta}
\end{align}
This completes the proof.
\end{proof}

Given the established tail bounds for the state~$\ket{\gkp_{\kappa,\Delta}}$, we can  use the projections~$\Pi_{[-R,R]}$ and $\widehat{\Pi}_{[-R,R]}$ to establish lower bounds on the distance of a state~$\rho$ to $\proj{\gkp_{\kappa,\Delta}}$, as follows.

\begin{lemma}\label{lem:distancetogkpprojection}
Let $\rho\in\cB(L^2(\mathbb{R}))$ be a state, $\kappa\in(0,1/4)$ and $\Delta \in (0,1/100)$.
Then
\begin{align}
\renewcommand*{\arraystretch}{1.225}
\begin{matrix*}[l]
\left\|
\rho-\proj{\gkp_{\kappa,\Delta}}\right|_1
&\geq& 2\left(\tr(\Pi_{[-R,R]}\rho)-(4\kappa R+5\sqrt{\kappa}+7\sqrt{\Delta})\right)\\
    \left\|
\rho-\proj{\gkp_{\kappa,\Delta}}\right|_1
&\geq& 2\left(\tr(\widehat{\Pi}_{[-R,R]}\rho)-(2\Delta R + 5\sqrt{\kappa}+7\sqrt{\Delta})\right)
\end{matrix*}\qquad\textrm{ for all }R>0\ .\notag
    \end{align} 
 \end{lemma}
\begin{proof}The claim follows from~\eqref{eq: variational L1} with the choice~$\Pi=\Pi_{[-R,R]}$ respectively~$\Pi=\widehat{\Pi}_{[-R,R]}$ and Lemma~\ref{lem: proj pos space} giving an upper bound on the norms of the projected states~$\Pi_{[-R,R]}\ket{\gkp_{\kappa,\Delta}}$, $\widehat{\Pi}_{[-R,R]}\ket{\gkp_{\kappa,\Delta}}$. 
\end{proof}

\subsection{A lower bound on the unitary complexity of~$\ket{\gkp_{\kappa,\Delta}}$\label{sec:unitarycomplexitygkp}}

In this section, we establish a lower bound on the approximate unitary state complexity  $\cC^*_\varepsilon(\ket{\gkp_{\kappa,\Delta}})$ of the state~$\ket{\gkp_{\kappa,\Delta}}$ as introduced in Section~\ref{sec: unitary state complexity}. We prove the following.
\begin{theorem} \label{thm: unitary state complexity}
    Let $\kappa, \Delta>0$ be such that 
    \begin{align}
20\sqrt{\kappa} + 28\sqrt{\Delta}\le 1\, .\label{eq:kappadeltaassumption}
\end{align}
    Then, we have
    \begin{align}
        \cC_1^*(\ket{\gkp_{\kappa,\Delta}})\geq \frac{1}{8\pi}\left(\log 1/\kappa + \log 1/\Delta\right) -1\, .
    \end{align}
\end{theorem}
In particular, we infer the following from Theorem~\ref{thm: unitary state complexity}. 
\begin{corollary}\label{cor: unitary state compl gkp limit}
    \begin{align}
\cC_{1}^*(\ket{\gkp_{\kappa,\Delta}})
=\Omega(\log 1/\kappa+\log 1/\Delta)  \qquad\textrm{for}\qquad(\kappa,\Delta) \rightarrow (0,0)\ . 
    \end{align}
\end{corollary}

\begin{proof}[Proof of Theorem~\ref{thm: unitary state complexity}]
    Consider a circuit $U$ 
    preparing a distance-$1$ approximation to~$\ket{\gkp_{\kappa,\Delta}}$ with minimal  complexity~$\cC_1^*(\ket{\gkp_{\kappa,\Delta}})$. In other words, the circuit~$U$ uses  $m+1$~bosonic modes and $m'$~qubits,  and is composed of $T$ unitaries from the set~$\cG$. These parameters satisfy
     \begin{align}
    \cC_1^*(\ket{\gkp_{\kappa,\Delta}})= T+(m+1)+m'\ .\label{eq:complexitygkpstatem}
\end{align}
Furthermore, denoting the state  of the first mode  after application of~$U$ to the initial state      $\ket{\Psi}=\ket{\vac}\otimes\ket{\vac}^{\otimes m}\otimes\ket{0}^{\otimes m'}$
     by  $\rho := \tr_{m,m'}U\proj{\Psi}U^\dagger$, we have
     \begin{align}
         \|\rho-\proj{\gkp_{\kappa,\Delta}}\|_1\leq 1\ .
     \end{align}
     
    We proceed in two steps. First, the assumption~\eqref{eq:kappadeltaassumption} implies that~$5\sqrt{\kappa}+7\sqrt{\Delta}<1/4$ and also that $\kappa <1/4$ and $\Delta < 1/100$.
Therefore, by the first inequality of Lemma~\ref{lem:markovinequalityprojection} 
and by
Lemma~\ref{lem:distancetogkpprojection}, we have
\begin{align}
    \|\rho-\proj{\gkp_{\kappa,\Delta}}\|_1 &\geq 2\left(1-\left(\frac{E_m(T)}{R^2}+4\kappa R\right)-1/4 \right)\qquad\textrm{for all}\qquad R>0\ ,
\end{align}
which implies 
\begin{align}
    E_m(T) \ge \left( 1/4 - 4\kappa R\right)R^2\qquad\textrm{for}\qquad R>0 \, .
\end{align}
We maximize the \rhs for $R>0$. The global maximum of the function $g(R) = (1/4 - 4\kappa R)R^2$ restricted to $R>0$ is at $R_0 = 1/(24\kappa)$ with $g(R_0) = 1/(12\cdot 24^2 \kappa^2)$.

Inserting the definition of~$E_m(T)$ from 
 Lemma~\ref{lem: moment limits unitary state prep} into the above gives the lower bound
 \begin{align}
 8\pi T  +\log(m+2) \geq   \log 1/\kappa^2 - \log(12\cdot24^{2}) \ ,
     \end{align}
     hence using that $\left(\log(m+2)\right)/(8\pi) \le m+1$ for all $m\ge 0$, and that $\log(12\cdot24^2)/(8\pi)<1$, we obain
     \begin{align}
     T+(m+1)+m'\geq \frac{1}{4\pi}\cdot \log1/\kappa  - 1 \ .
     \end{align}
     Inserting this into~\eqref{eq:complexitygkpstatem}
     gives the first lower bound
     \begin{align}
     \cC_1^*(\ket{\gkp_{\kappa,\Delta}})\geq  \frac{1}{4\pi}\cdot \log 1/\kappa -1 \label{eq: unitary complexity position}
     \end{align}
     on the complexity~$\cC_1^*(\ket{\gkp_{\kappa,\Delta}})$.

    Second, by a similar calculation using the bound involving~$\widehat{\Pi}_{[-R,R]}$ in Lemma~\ref{lem:markovinequalityprojection}, we have (by the assumption~\eqref{eq:kappadeltaassumption}, which is~$5\sqrt{\kappa} + 7\sqrt{\Delta} <1/4$) 
    \begin{align}
\|\rho-\proj{\gkp_{\kappa,\Delta}}\|_1 &\ge 2\left( 1-\left(\frac{E_m(T)}{R^2} + 2\Delta R \right)-1/4\right)\qquad\textrm{for all}\qquad R>0\, .
        \end{align}
    This implies 
    \begin{align}
        E_m(T) \ge \left(1/4 - 2\Delta R\right) R^2\qquad \textrm{for all}\qquad R>0\, .
    \end{align}
    Proceeding as we did in the calculation involving $\Pi_{[-R,R]}$, we find 
    \begin{align}
        E_m(T) \ge 1/(6\cdot 12^2) 1/\Delta^2\, .
    \end{align}
    
    Again inserting $E_m(T)$ from Lemma~\ref{lem: moment limits unitary state prep} yields 
    \begin{align}
        8\pi T + \log(m+2) \ge \log 1/\Delta^2 - \log(6 \cdot 12^2)
    \end{align}
    hence, by the same argument as before, along with $(\log(6 \cdot 12^2))/(8\pi) <1$, we get
    \begin{align}
        T + (m+1) + m' \ge \frac{1}{4\pi}\cdot \log 1/\Delta-1 \ .
    \end{align}
    Therefore, this second analysis shows that the unitary complexity is bounded as
    \begin{align}
        \cC_1^*(\ket{\gkp_{\kappa,\Delta}})\geq \frac{1}{4\pi}\cdot \log 1/\Delta-1\ . \label{eq: unitary complexity momentum}
    \end{align}
    Combining the bounds from~\eqref{eq: unitary complexity position} and~\eqref{eq: unitary complexity momentum}, we deduce
    \begin{align}
        \cC_1^*(\ket{\gkp_{\kappa,\Delta}})\geq \frac{1}{8\pi}\left(\log 1/\kappa + \log 1/\Delta\right) -1\ .
    \end{align}
    The claim follows.
\end{proof}

\subsection{A lower bound on the heralded complexity of~$\ket{\gkp_{\kappa,\Delta}}$ \label{sec: lower bound heralded state complexity}}

In this section, we prove a lower bound on the complexity of preparing the state $\ket{\gkp_{\kappa,\Delta}}$ with heralding protocols as introduced in Section~\ref{sec: heralded state complexity}.
We will argue that as long as the number of operations (cf. Section~\ref{sec: heralded state complexity}) is bounded, the resulting states will be far from a GKP state except with low probability. This immediately implies a lower bound on the heralded complexity of~$\ket{\gkp_{\kappa,\Delta}}$ states. The procedure is similar to bounding the unitary complexity. 
We first need to lower bound the terms of the form $\tr\left(\Pi_{[-R,R]}\rho_\acc\right)$ and $\tr\left(\widehat{\Pi}_{[-R,R]}\rho_\acc\right)$, where $\rho_\acc$ is the output state conditioned on acceptance. We show the following.

\begin{lemma}\label{lem: proj rho acc upper bound}
    Consider a heralding protocol described by the unitary $U$ on $L^2(\mathbb{R})^{\otimes (m+1)}\otimes (\mathbb{C}^2)^{\otimes m'}$ composed of $T_1$ gates from $\cG$ (before the measurements) and acceptance probability $p_\acc \in (0,1]$. Let $R>0$ and assume that the maximal norm of a displacement vector associated with a displacement operator applied after measurements is bounded by \begin{align}
        d_\acc := \sup_{\alpha\in F^{-1}(\{\acc\})} \|d(\alpha)\| \le R/2\ .
    \end{align}Then, 
    \begin{align}
       \min\left\{ \tr\left(\Pi_{[-R,R]}\rho_\acc\right), \tr\left(\widehat{\Pi}_{[-R,R]}\rho_\acc\right) \right\} &\ge p_\acc - \frac{4E(T_1)}{R^2}\ ,
    \end{align}
    with the function $E_m$ defined in Lemma~\ref{lem: moment limits unitary state prep}. 
\end{lemma}
\begin{proof}
    We only show
    \begin{align}
        \tr\left(\Pi_{[-R,R]}\rho_\acc\right) \ge p_\acc - \frac{4E(T_1)}{R^2}\ ;
    \end{align} the inequality for $\widehat{\Pi}_{[-R,R]}$ follows by analogous arguments.
    As $|d_P(\alpha)| \le \|d(\alpha)\|\leq R/2$ for all $\alpha\in F^{-1}(\{\acc\})$, we have
    \begin{align}
  [-R/2,R/2]\subseteq [-R+d_P(\alpha),R+d_P(\alpha)]\ .
  \end{align}
  Moreover, we note that
  \begin{align}
      D(d(\alpha))^\dagger \Pi_{[-R, R]} D(d(\alpha)) = \Pi_{[-R +d_P(\alpha), R+d_P(\alpha)]}\ .
  \end{align}
Hence, we conclude
\begin{align}
\tr(\Pi_{[-R,R]}\rho_\acc)
&=
\int_{F^{-1}(\{\acc\})} 
\tr\left((\Pi_{[-R,R]}\otimes I)
(D(d(\alpha))\otimes E_\alpha)
U\proj{\Psi}U^\dagger (D(d(\alpha))^\dagger \otimes I)\right)d\alpha\notag\\
&
=
\int_{F^{-1}(\{\acc\})}
\tr\left((D(d(\alpha))^\dagger\Pi_{[-R,R]} D(d(\alpha)))\otimes E_\alpha)
U\proj{\Psi}U^\dagger \right)d\alpha\\
&
=
\int_{F^{-1}(\{\acc\})} 
\tr\left(\Pi_{[-R+d_P(\alpha),R+d_P(\alpha)]}
(I\otimes E_\alpha)
U\proj{\Psi}U^\dagger \right)d\alpha\\
&\geq 
\int_{F^{-1}(\{\acc\})}
\tr\left(\Pi_{[-R/2,R/2]}\otimes E_\alpha)
U\proj{\Psi}U^\dagger \right)d\alpha\\
&=
\langle U\Psi,(\Pi_{[-R/2,R/2]}\otimes E_\acc)U\Psi\rangle \label{eq: bound proj rho acc first}
\end{align}
where we defined $ E_\acc := \int_{F^{-1}(\{\acc\})}E_\alpha d\alpha$. For later reference, we also define the POVM element associated with the output flag $\rej$ by $E_\rej := \int_{F^{-1}(\{\rej\})}E_\alpha d\alpha$.
Let us consider the post-measurement state without shift corrections $D(d(\alpha))$, which we can decompose according to the flag outputs $\acc$ and $\rej$ as
\begin{align}
    \tr_{m,m'}
    U\proj{\Psi}U^\dagger&=
    p_\acc \rho'_{\acc}+ (1- p_\acc)\rho'_{\rej}\ ,
\end{align}
where 
\begin{align}
    \rho'_{\acc} = \frac{1}{p_\acc}\int_{F^{-1}(\{\acc\})} p(\alpha) \ \rho^{(\alpha)} d\alpha\ \qquad \textrm{and}\qquad \rho'_{\rej} = \frac{1}{1-p_\acc}\int_{F^{-1}(\{\rej\})} p(\alpha) \ \rho^{(\alpha)} d\alpha\ ,
\end{align} and where $p(\alpha)\rho^{(\alpha)}= \bra{\Psi}U^\dagger (I \otimes E_\alpha) U\ket{\Psi}$.
Clearly, we have
\begin{align}
\langle U\Psi,(\Pi_{[-R/2,R/2]}\otimes I)U\Psi
\rangle
&=\langle 
U\Psi,(\Pi_{[-R/2,R/2]}\otimes E_\acc) U\Psi\rangle\\
&\quad+
\langle 
U\Psi,(\Pi_{[-R/2,R/2]}\otimes E_\rej) U\Psi\rangle \label{eq: proj I decomposition}
\end{align}
and 
\begin{align}
\langle 
U\Psi,(\Pi_{[-R/2,R/2]}\otimes E_\rej) U\Psi\rangle\leq 
\langle 
U\Psi,(I\otimes E_\rej) U\Psi\rangle=1-p_{\acc}\ . \label{eq: bound proj E second}
\end{align}
Combining~\eqref{eq: bound proj rho acc first} and~\eqref{eq: proj I decomposition}, we obtain
\begin{align}
 \tr(\Pi_{[-R,R]}\rho_\acc) &\ge 
\langle 
U\Psi,(\Pi_{[-R/2,R/2]}\otimes E_\acc) U\Psi\rangle\\
&=\langle U\Psi,(\Pi_{[-R/2,R/2]}\otimes I)U\Psi
\rangle-\langle 
U\Psi,(\Pi_{[-R/2,R/2]}\otimes E_\rej) U\Psi\rangle\\
&\geq \langle U\Psi,(\Pi_{[-R/2,R/2]}\otimes I)U\Psi
\rangle-(1-p_\acc) \qquad \textrm{by~\eqref{eq: bound proj E second}}\\
&= p_\acc - \left(1 - \langle U\Psi,(\Pi_{[-R/2,R/2]}\otimes I)U\Psi
\rangle \right) \ .\label{eq: proj E third}
\end{align}
Moreover, we can bound
\begin{align}
     \langle U\Psi,(\Pi_{[-R/2,R/2]}\otimes I)U\Psi \rangle &= \tr\left(\Pi_{[-R/2,R/2]} \tr_{m,m'}U\proj{\psi}U^\dagger \right)\\
     &\ge 1- \frac{\tr\left( (Q^2 + P^2) \tr_{m,m'}U\proj{\psi}U^\dagger\right)}{(R/2)^2}\\
     &= 1 - \frac{\bra{\Psi}U^\dagger(Q_1^2 + P_1^2)U\ket{\Psi}}{(R/2)^2}\\
     &\ge 1 - \frac{\bra{\Psi}U^\dagger(\sum_{i=1}^{m+1} Q_i^2 + P_i^2)U\ket{\Psi}}{(R/2)^2} \, , \label{eq: proj R/2 first}
\end{align}
where we used Lemma~\ref{lem:markovinequalityprojection} applied to the reduced state $\tr_{m,m'}U\proj{\psi}U^\dagger$ to obtain the first inequality.
The claim follows from combining~\eqref{eq: proj E third} and~\eqref{eq: proj R/2 first} with the upper bound on the energy $E_m$ stated in Lemma~\ref{lem: moment limits unitary state prep}.
\end{proof}

\begin{theorem} \label{thm: lower bound heralded complexity}
    Let $\kappa$, $\Delta>0$. Then, for all $p\in(0,1]$ and $\varepsilon>0$ such that  
    \begin{align} 
       20\sqrt{\kappa} + 28 \sqrt{\Delta} \le p \label{eq: assum kappa delta heralded}\qquad\textrm{and}\qquad
       \varepsilon \le p \ ,
    \end{align} we have
    \begin{align}
        \cC_{p,\varepsilon}^{\mathsf{her},*}(\ket{\gkp_{\kappa,\Delta}}) \ge \frac{1}{200}\left(\log1/\kappa + \log1/\Delta\right) -1\ .
    \end{align}
\end{theorem}

Theorem~\ref{thm: lower bound heralded complexity} implies the following corollary.
\begin{corollary} \label{cor: her state complexity}There exists a polynomial $s(\kappa,\Delta)$ with $s(0,0)=0$ such that for all functions $p(\kappa,\Delta)$ and $\varepsilon(\kappa,\Delta)$  satisfying $0\le \varepsilon(\kappa,\Delta) \le p(\kappa,\Delta)$ and $s(\kappa,\Delta) \le p(\kappa,\Delta)\le 1$ for sufficiently small $(\kappa,\Delta)$, we have
\begin{align}
\cC_{p(\kappa,\Delta),\varepsilon(\kappa,\Delta)}^{*,\mathsf{her}}(\ket{\gkp_{\kappa,\Delta}}) \ge \Omega\left(\log1/\kappa + \log 1/\Delta \right) \qquad\textrm{for}\qquad (\kappa,\Delta) \rightarrow (0,0)\ .
    \end{align}
\end{corollary}
\begin{proof}[Proof of Theorem~\ref{thm: lower bound heralded complexity}]
    Consider a heralding protocol, i.e., a unitary $U$ composed of $T_1$ gates from $\cG$, a POVM $\left\{E_\alpha\right\}$, a map $F$ computing a flag, and shift corrections $\left\{D(d(\alpha))\right\}$, preparing with probability at least $p$ a distance-$\varepsilon$ approximation to~$\ket{\gkp_{\kappa,\Delta}}$ with minimal complexity~$\cC_{p,\varepsilon}^{*,\her}(\ket{\gkp_{\kappa,\Delta}})$. In other words, the protocol uses  $m+1$~bosonic modes, $m'$~qubits  and  requires (including shift corrections) $T_1 + T_2$ gates from the set~$\cG$, where $T_2$ accounts for the complexity of shift corrections (cf.~the definition in~\eqref{eq: def T_2}). These parameters satisfy
     \begin{align}
    \cC_{p,\varepsilon}^{*,\mathsf{her}}(\ket{\gkp_{\kappa,\Delta}})= T_1 + T_2 +(2m+1)+2m'\ .\label{eq:complexitygkpstatem heralded}
\end{align}
In addition, the output state conditioned on acceptance $\rho_\acc$ (cf.~\eqref{eq: def rho acc}) satisfies
\begin{align}
    \left\| \rho_\acc - \proj{\gkp_{\kappa,\Delta}}\right\|_1 \le \varepsilon\ .
\end{align}
In the following, we argue similarly as in the proof of Theorem~\ref{thm: unitary state complexity}.

For convenience, we define
\begin{align}
        d_\acc :=\sup_{\alpha\in F^{-1}(\{\acc\})} \|d(\alpha)\| \le R/2\ . \label{eq: constraint d_acc}
    \end{align}
By Lemma~\ref{lem:distancetogkpprojection} and Lemma~\ref{lem: proj rho acc upper bound} (using the assumption $p_\acc \ge p$), and considering the projector $\Pi_{[-R,R]}$, we have
\begin{align}
    \left\| \rho_\acc - \proj{\gkp_{\kappa,\Delta}}\right\|_1  \ge  2\left(p-\left(\frac{4E_m(T_1)}{R^2}+4\kappa R\right)-(7\sqrt{\kappa}+5\sqrt{\Delta})\right)
    \label{eq: dist rho acc E bound kappa first}
\end{align}
for any $R \ge 2d_\acc$, where $d_\acc$ is defined in Lemma~\ref{lem: proj rho acc upper bound}. By the assumption~\eqref{eq: assum kappa delta heralded}, we have $5\sqrt{\kappa} + 7\sqrt{\Delta} < p/4$. Thus, solving for $E_m(T_1)$, we find
\begin{align}
    E_m(T_1) &\ge \frac{1}{4}\left(\left(3p/4 - \varepsilon/2 \right) - 4\kappa R\right) R^2\\
    &\geq \frac{1}{4}\left(p/4 - 4\kappa R\right) R^2
    =:f(R) \qquad\textrm{for all} \qquad R\ge 2d_\acc\ .\label{eq: E(U) lower bound R}
\end{align}
where we used the assumption $\varepsilon\leq p$. 
Since this bound is valid for any $R\ge 2d_\acc$, we can maximize the \rhs of this inequality.
That is, we choose \begin{align}
    R_\textrm{max} =\arg\max_{R\ge2d_\acc} f(R)\ .
\end{align}
Note that for $A,B>0$, the function $x\mapsto (A-Bx)x^2$ restricted to $x\ge 0$ has its global maximum at $x_0 = 2A/(3B)$ and is monotonously decreasing for $x\ge x_0$. Hence, the global maximum of~$f(R)$ on  $[0,\infty)$ is at 
\begin{align}
    R_0  = \frac{p}{24\kappa}\, .
\end{align}We have to distinguish two cases, depending on whether this global maximum is contained in $[2d_\acc,\infty)$ or not.
\begin{enumerate}[(a)]
    \item If $R_0 \geq 2d_\acc$, we conclude that 
    the maximum of $f(R)$ restricted to~$[2d_\acc,\infty)$ is equal to 
    $ R_\textrm{max} = R_0$
and so inserting this into~\eqref{eq: E(U) lower bound R}, we find
\begin{align}
    E_m(T_1) \ge \frac{p^3}{432\cdot 4^3} \cdot 1/\kappa^2\ .
\end{align}

Inserting the definition of $E_m$ from Lemma~\ref{lem: moment limits unitary state prep} and solving for $T_1$, we obtain
\begin{align}
    8\pi T_1 + \log(m+2) \ge \log1/\kappa^2 + 3\log p -\log (432\cdot 4^3)\ .
\end{align}
Since $T_1\le T_1 + T_2$, $(\log(m+2))/(8\pi) \le m+1$ for all $m\ge 0$ and $(\log(432\cdot 4^3))/(8\pi)<1$ this implies
\begin{align}
    T_1 + T_2+ (2m+1) + 2m' \ge \frac{1}{4\pi} \log1/\kappa + \frac{3}{8\pi} \log p-1
\end{align}
hence with~\eqref{eq:complexitygkpstatem heralded}
\begin{align}
    \cC_{p,\varepsilon}^{*,\mathsf{her}}(\ket{\gkp_{\kappa,\Delta}}) \ge \frac{1}{4\pi}\log1/\kappa + \frac{3}{8\pi}\log p-1 \label{eq: C kappa case 1}\ .
\end{align}

\item If $R_0<2d_\acc$, we have
\begin{align}
T_2&=\max_{\alpha\in F^{-1}(\{\acc\})} \cC_{\cG}(D(d(\alpha)))\\
&\geq \max_{\alpha\in F^{-1}(\{\acc\})} \frac{1}{4\pi}\left(\log \|d(\alpha)\|\right)-1\\
&=\frac{1}{4\pi}(\log d_\acc)-1
\end{align}
where we used the lower bound~\eqref{eq: lower bound complexiry displacement complexityd} 
on the complexity of a displacement operator~$D(d)$. We conclude that 
    \begin{align}
        T_1 + T_2 +(2m+1)+2m'&\geq
        T_2+1\\
        &\geq \frac{1}{4\pi}\log d_\acc\\
        &\ge \frac{1}{4\pi} \log R_0/2\\
        &\ge \frac{1}{4\pi} \log1/\kappa + \frac{1}{4\pi} \log \left(p/48\right) \ .
    \end{align}
Hence in this case using $\log(48)/(4\pi)<1$, we have
\begin{align}
    \cC_{p,\varepsilon}^{*,\mathsf{her}}(\ket{\gkp_{\kappa,\Delta}}) &\ge  \frac{1}{4\pi} \log1/\kappa + \frac{1}{4\pi} \log \left(p/48\right) \\
    &\geq \frac{1}{4\pi} \log1/\kappa + \frac{1}{4\pi} \log p-1\ 
    \label{eq: C kappa case 2} \ .
\end{align}
\end{enumerate}
In both cases, we have shown that 
\begin{align}
        \cC_{p,\varepsilon}^{*,\mathsf{her}}(\ket{\gkp_{\kappa,\Delta}}) 
    &\geq \frac{1}{4\pi} \log1/\kappa + \frac{3}{8\pi} \log p-1\ 
\end{align}
Due to the assumption~\eqref{eq: assum kappa delta heralded}, we have $p\geq \sqrt{\kappa}$. This implies that 
\begin{align}
\cC_{p,\varepsilon}^{*,\mathsf{her}}(\ket{\gkp_{\kappa,\Delta}}) &\ge \left(\frac{1}{4\pi}-\frac{3}{16\pi}\right)\log 1/\kappa-1\\
&\geq \frac{1}{200}\log 1/\kappa-1\, .\label{eq: final eq kappa her}
\end{align}

Analogously, we can derive bounds using $\widehat{\Pi}_{[-R,R]}$ instead of $\Pi_{[-R,R]}$. Again by Lemma~\ref{lem: proj rho acc upper bound}, Lemma~\ref{lem:distancetogkpprojection}, the assumption~\eqref{eq: assum kappa delta heralded} and $p_\acc \ge p$, we have
\begin{align}
    \left\| \rho_\acc - \proj{\gkp_{\kappa,\Delta}}\right\|_1  \ge  2\left(3p/4-\left(\frac{4E_m(T_1)}{R^2}+2\Delta R \right)\right) \label{eq: dist rho acc E bound Delta first}
\end{align}
for any $R \ge 2d_\acc$. Solving for $E_m(T_1)$ and using that $\varepsilon \le p $, we find
\begin{align}
    E_m(T_1) \ge   \frac{1}{4}\left(p/4 - 2\Delta R \right) R^2 \qquad\textrm{for all} \qquad R\ge 2d_\acc\ .\label{eq: E(U) lower bound R Delta}
\end{align}
Maximizing the rhs.\ for $R\ge 2d_\acc$  and proceeding as before with a case distinction, we  arrive again at
\begin{align}
    \cC_{p,\varepsilon}^{*,\mathsf{her}}(\ket{\gkp_{\kappa,\Delta}}) \ge \frac{1}{4\pi}\log1/\Delta + \frac{3}{8\pi}(\log p) -1 \ . \label{eq: C delta case 1}
\end{align} 
By the same argument as in~\eqref{eq: final eq kappa her}, we infer
\begin{align}
     \cC_{p,\varepsilon}^{*,\mathsf{her}}(\ket{\gkp_{\kappa,\Delta}}) \geq \frac{1}{100}\log 1/\Delta-1 \label{eq: final delta her}
\end{align}
Combining~\eqref{eq: final eq kappa her} and \eqref{eq: final delta her} it follows
\begin{align}
    \cC_{p,\varepsilon}^{*,\mathsf{her}}(\ket{\gkp_{\kappa,\Delta}}) \geq \frac{1}{200}\left(\log1/\kappa+ \log 1/\Delta\right)-1
\end{align}
\end{proof}

In particular, Theorem~\ref{thm: lower bound heralded complexity} covers all heralding protocols whose acceptance probability $p_\acc$ is lower bounded by a constant for sufficiently small $\kappa$ and $\Delta$, as it is the case in Protocol~\ref{prot: appx GKP state prep} (cf.~Theorem~\ref{thm:main}). We are in the position to combine the upper bound from Corollary~\ref{cor: upper bound her state complexity} and the upper bound on the heralded state complexity stated in Corollary~\ref{cor: her state complexity} to get the following.
\begin{corollary}
    There is a polynomial $r(\kappa,\Delta)$ with $r(0,0)=0$ such that for all functions $p(\kappa,\Delta)$ and $\varepsilon(\kappa,\Delta)$ satisfying $r(\kappa,\Delta) \le \varepsilon(\kappa,\Delta) \le p(\kappa,\Delta) \le 1/10$ for sufficiently small $(\kappa,\Delta)$, we have
    \begin{align}
        \cC_{p(\kappa,\Delta),\varepsilon(\kappa,\Delta)}^{*,\mathsf{her}}(\ket{\gkp_{\kappa,\Delta}}) = \Theta(\log1/\kappa + \log1/\Delta) \qquad\textrm{for}\qquad(\kappa,\Delta) \rightarrow (0,0)\, .
    \end{align}
\end{corollary}

\section{Conclusions}
We demonstrated that the complexities of coherent states and approximate GKP states are essentially determined by their energies. The scaling of their complexities can thus be determined by straightforward back-of-the-envelope calculations. Based on these two qualitatively very different examples, it is natural to ask whether the observed relationship between complexity and energy can be generalized to arbitrary bosonic states. More specifically, can every bosonic state be (approximately) prepared with a number of elementary operations that scales logarithmically with the energy of the state?

While our state-preparation protocols achieve optimal complexity under ideal, noiseless conditions, a compelling open question is whether efficient fault-tolerant implementations can be devised with similar fidelity guarantees. This question is particularly relevant for approximate GKP states, as these are instrumental in universal fault-tolerant quantum computation with continuous-variable systems. Existing work~\cite{Shi_2019} in this direction provides first examples of such protocols but only cover highly restricted error models.

\section*{Acknowledgements}
LB, LC and RK gratefully acknowledge support by the European Research Council
under grant agreement no. 101001976 (project EQUIPTNT), as well as the Munich Quantum Valley, which is supported by the Bavarian state government through the Hightech Agenda Bayern Plus. XCR acknowledges funding from the BMW endowment fund. XCR thanks the Swiss National Science Foundation (SNSF) for their financial support.

\newpage
\appendix
\section{Basic properties of approximate GKP states}
In this appendix, we prove closeness of various types of approximate GKP states that we use.

\subsection{Bounds on sums of Gaussians}

\begin{lemma} \label{lem: gaussian sum integral bound} Let $c>0$.  Then, 
\begin{align} \label{eq: bound 1 lem gaussian bound}
   \sqrt{\frac{\pi}{c}} - 1 
   \le &\sum_{z \in \mathbb{Z}} e^{-cz^2} 
   \le 
   \sqrt{\frac{\pi}{c}} + 1\ .
\end{align}
In particular,\begin{align}
     \label{eq: bound 2 lem gaussian bound}
    \sum_{z \in \mathbb{Z}\setminus\{0\}} e^{-cz^2} \le  
    \sqrt{\frac{\pi}{c}}\ .
\end{align}
Moreover, we can bound
\begin{align}
    \sum_{z \in \mathbb{Z}} e^{-c(z-1/2)^2} \ge \sqrt{\frac{\pi}{c}} - 1\, .\label{eq: bound gaussian shift 1/2}
\end{align}
We also have for $\varepsilon\in(0,1/2)$ that
\begin{align}
    \sum_{z \in \mathbb{Z}} e^{-c(|z|+\varepsilon)^2} & \ge  \sqrt{\frac{\pi}{c}}-2(1+\varepsilon)\label{eq: abs z gauss shift lower}\\
    \sum_{z \in \mathbb{Z}\setminus\{0\}} e^{-c(|z|-\varepsilon)^2} & \le \sqrt{\frac{\pi}{c}}+2  \ .\label{eq: abs z gauss shift upper}
\end{align}

\end{lemma}
\begin{proof}
    We obtain the bounds in Eq.~\eqref{eq: bound 1 lem gaussian bound} and~\eqref{eq: bound 2 lem gaussian bound} as follows. Notice that $x\mapsto e^{-cx^2}$ is monotonously decreasing function for $c>0$ and $x\ge 0$. Thus,
    \begin{align}
    e^{-cz^2} &\le \int_{z-1}^z e^{-cx^2} dx \qquad\textrm{ for any $z\ge 1$}\\
    e^{-cz^2} &\ge  \int_{z}^{z+1} e^{-cx^2} dx \qquad\textrm{ for any $z\ge 0$} \ . \label{eq: gauss one step int bound}
    \end{align}
    Thus, we obtain the upper bound in~\eqref{eq: bound 1 lem gaussian bound}:
    \begin{align}
        \sum_{z \in \mathbb{Z}} e^{-cz^2} = 1+ \sum_{z \in \mathbb{Z}\setminus\{0\}} e^{-cz^2} \le 1+ \int e^{-cx^2} dx = \sqrt{\frac{\pi}{c}} + 1\ .
    \end{align}
    This also implies the upper bound~\eqref{eq: bound 2 lem gaussian bound}.
    By~\eqref{eq: gauss one step int bound}, we have
    \begin{align}
        \sum_{z \in \mathbb{Z}} e^{-cz^2}\ge 2\int_1^\infty e^{-cx^2} dx + 1  \ .\label{eq:sum gauss int lowerbound}
    \end{align}
    Since $e^{-cx^2} \le 1$ for all $x \in \mathbb{R}$, we infer $\int_{0}^1  e^{-cx^2} dx \le 1$. This with~\eqref{eq:sum gauss int lowerbound} gives the lower bound in~\eqref{eq: bound 1 lem gaussian bound}
    \begin{align}
        \sum_{z \in \mathbb{Z}} e^{-cz^2}\ge \int e^{-cx^2} dx - 1=\sqrt{\frac{\pi}{c}} -1 \ .
    \end{align}
    
To show the bound~\eqref{eq: bound gaussian shift 1/2}, we note that
\begin{align}
     \sum_{z \in \mathbb{Z}} e^{-c(z-1/2)^2} = 2 \sum_{z \in \mathbb{N}} e^{-c(z-1/2)^2} \ge 2\int_{1/2}^\infty e^{-cx^2}dx  \ge \int e^{-cx^2}dx - 1 = \sqrt{\frac{\pi}{c}} - 1\, ,
\end{align}
where we used~\eqref{eq: gauss one step int bound} to obtain the first inequality and the fact that $e^{-cx^2} \le 1$ for all $x$ in the second inequality.

Next, we show Eq.~\eqref{eq: abs z gauss shift lower}. Observe that
\begin{align}
    \sum_{z \in \mathbb{Z}} e^{-c(|z|+\varepsilon)^2} =  \sum_{z > 0} e^{-c(z+\varepsilon)^2} + \sum_{z \in\mathbb{N}} e^{-c(z+\varepsilon)^2} \ge 2 \sum_{z \in\mathbb{N}} e^{-c(z+\varepsilon)^2}\ . \label{eq: shifted gauss sum bound first}
\end{align}
By \eqref{eq: gauss one step int bound} and by the assumption $\varepsilon\in(0,1/2)$, we have
\begin{align}
    2 \sum_{z \in\mathbb{N}} e^{-c(z+\varepsilon)^2} &\ge 2\int_{1+\varepsilon}^\infty e^{-cx^2} dx\\
    &\ge \int e^{-cx^2} dx - 2(1+\varepsilon) \ ,
\end{align}
where we used that $e^{-cx^2}\le 1$ for all $x>0$ and $c>0$. Combining this with Eq.~\eqref{eq: shifted gauss sum bound first} implies the claim in Eq.~\eqref{eq: abs z gauss shift lower}
\begin{align}
    \sum_{z \in \mathbb{Z}} e^{-c(|z|+\varepsilon)^2} \ge \sqrt{\frac{\pi}{c}}-2(1+\varepsilon)\ .
\end{align}

Finally, we show~\eqref{eq: abs z gauss shift upper}. By $e^{-c x^2}\le 1$ for all $x\in\mathbb{R}$ and $c>0$, we have
\begin{align}
    \sum_{z \in \mathbb{Z}\setminus\{0\}} e^{-c(|z|-\varepsilon)^2} 
    &= 2 \sum_{z\in \mathbb{N}} e^{-c(z-\varepsilon)^2}
    \le 2 + 2 \sum_{z\ge2} e^{-c(z-\varepsilon)^2} \ .
\end{align}
This, together with the upper bound from Eq.~\eqref{eq: gauss one step int bound}, implies the claim
\begin{align}
    \sum_{z \in \mathbb{Z}\setminus\{0\}} e^{-c(|z|-\varepsilon)^2}
    &\le 2+2\int_{1-\varepsilon}^{\infty} e^{-cx^2} dx\\
    &\le 2+\int e^{-cx^2} dx\\
    &=2+\sqrt{\frac{\pi}{c}}\ .
\end{align}
\end{proof}

\subsection{Distance bounds between approximate comb and GKP states}

In this section, we prove several bounds on the closeness of Gaussian states, approximate comb states, and approximate GKP states.

\subsubsection{Bounds on Gaussian states}
Let us recall the definition of Gaussian state with parameter $\Delta$ (cf. Eq.~\eqref{eq:chiDeltadefinition})
\begin{align}
\Psi_\Delta(x)=\frac{1}{(\pi \Delta^2)^{1/4}}e^{-x^2/(2\Delta^2)}\,
\ \label{eq:appx Psi Delta}
\end{align}
and its truncated variant (cf. Eq.~\eqref{eq:chiDelta ep definition})
\begin{align}
    \Psi^\varepsilon_\Delta= \frac{\Pi_{[-\varepsilon,\varepsilon]}\Psi_{\Delta}}{\|\Pi_{[-\varepsilon,\varepsilon]}\Psi_{\Delta}\|}\ , \label{eq:appx Psi Delta Ep}
\end{align}
where $\|\cdot\|$ is the Euclidean norm.

We show that, for a certain choice of parameters, these states are close to each other.

\begin{lemma}\label{lem: gaussian gaussian epsilon} Let $\Delta>0$ and $\varepsilon\in(0,1/2)$. Then
\begin{align}
    \abs{\langle \Psi_\Delta , \Psi_\Delta^{\varepsilon}\rangle}^2
    \ge 1-2e^{-(\varepsilon/\Delta)^2}\ . \label{eq: Psi Psi eps overlap}
\end{align}
Furthermore, if $\varepsilon\in [\sqrt{\Delta},1/2)$, we have
\begin{align}
    \left|\langle \Psi_\Delta, \Psi_\Delta^\varepsilon\rangle \right|^2 
    &\ge 1 - 2e^{-1/\Delta} \ge 1 - 2 \Delta\ . \label{eq: Psi Psi eps overlap restricted ep}
\end{align}
\end{lemma}

\begin{proof}
By definition of $\Psi_\Delta^\varepsilon$ (cf.~\eqref{eq:appx Psi Delta Ep}), we have
\begin{align}
    \langle \Psi_\Delta, \Psi_\Delta^\varepsilon\rangle
    &= \|\Pi_{[-\varepsilon,\varepsilon]}\Psi_{\Delta}\|^{-1}\langle \Psi_\Delta, \Pi_{[-\varepsilon,\varepsilon]}\Psi_\Delta\rangle\\
    &=\|\Pi_{[-\varepsilon,\varepsilon]}\Psi_{\Delta}\| \ .\label{eq: overlap psi psi ep}
\end{align}
By definition of $\Psi_{\Delta}$ (cf.~\eqref{eq:appx Psi Delta}) we have
\begin{align}
    \|\Pi_{[-\varepsilon,\varepsilon]}\Psi_{\Delta}\|^2 =\int_{-\varepsilon}^{\varepsilon}  \frac{1}{\sqrt{\pi}\Delta} e^{-(x/\Delta)^2} dx = \Pr[\abs{X}\le \varepsilon]\ ,
\end{align}
where we used that the integrand $x\mapsto 1/(\sqrt{\pi}\Delta)e^{-(x/\Delta)^2}$ is the probability density function of the random variable $X$ sampled from Gaussian distribution with mean $0$ and variance $\Delta^2/2$, i.e., $X\sim\cN(0,\Delta^2/2)$.  Thus, we can use the Chernoff bound (see e.g. Ref.~\cite{Vershynin_2018}) to obtain
\begin{align}
    \|\Pi_{[-\varepsilon,\varepsilon]}\Psi_{\Delta}\|^2
    &=1-2\Pr[X\ge \varepsilon]\\
    &\ge 1-2e^{-(\varepsilon/\Delta)^2}\ .
\end{align}
Inserting this into the square of Eq.~\eqref{eq: overlap psi psi ep} implies \eqref{eq: Psi Psi eps overlap}.

Eq.~\eqref{eq: Psi Psi eps overlap restricted ep} follows from \eqref{eq: Psi Psi eps overlap}, from the assumption $\varepsilon\in [\sqrt{\Delta},1/2)$, and from the inequality $e^{-x}\le 1/x$ for all $x \ge 0$. We have
\begin{align}
    \left|\langle \Psi_\Delta, \Psi_\Delta^\varepsilon\rangle \right|^2 &\ge 1 - 2e^{-(\varepsilon/\Delta)^2}\\
    &\ge 1 - 2e^{-1/\Delta}\\
    &\ge 1 - 2 \Delta\ .
\end{align}
\end{proof}

\begin{corollary}\label{cor: gaussian - gaussian epsilon trace distance}
    Assume $\Delta>0$ and $\varepsilon\in [\sqrt{\Delta},1/2)$. Then,    
    \begin{align}
         \left\| \proj{\Psi_\Delta} - \proj{\Psi_\Delta^{\varepsilon}}\right\|_1
         &\le 3 \sqrt{\Delta}\ .
    \end{align}
\end{corollary}
\begin{proof}

By Lemma~\ref{lem: gaussian gaussian epsilon} and by the relation between the trace distance and the overlap (cf.~Eq.~\eqref{eq: trace distance overlap}), we have
\begin{align}
    \left\|\proj{\Psi_\Delta} - \proj{\Psi_\Delta^\varepsilon} \right\|_1
    &\le 2\sqrt{2 \Delta} \le 3\sqrt{\Delta}\ .
\end{align}
\end{proof}

Let us recall the definition of translated Gaussian states~\eqref{eq:chiDeltadefinition} and
translated truncated Gaussian states \eqref{eq:chiDelta ep definition}:
\begin{align}
    (\chi_\Delta(z))(x)&:=\Psi_\Delta(x-z) \qquad \textrm{ and }\qquad 
    (\chi^\varepsilon_\Delta(z))(x):=\Psi^\varepsilon_\Delta(x-z)\ .
\end{align}

They satisfy the following statements.
\begin{lemma}\label{lem: overlap chi z chi z'} Let $z,z'\in\mathbb{Z}$ and $\Delta\in(0,1/4)$. We have
\begin{align}
    \langle \chi_\Delta(z), \chi_\Delta(z')\rangle &=e^{-(z-z')^2/(4 \Delta^2)} \ .
\end{align}
\end{lemma}
\begin{proof}
We have by definition that
\begin{align}
    \langle \chi_\Delta(z), \chi_\Delta(z')\rangle  &= \frac{1}{\sqrt{\pi}\Delta} \int_{-\infty}^{\infty} e^{-(x-z)^2/(2\Delta^2)} e^{-(x-z')^2/(2\Delta^2)} dx \\
    &=  \frac{1}{\sqrt{\pi}\Delta} \int_{-\infty}^{\infty} e^{-((x-z)^2- (x-z')^2)/(2\Delta^2)}dx\ .
\end{align}
Note that 
\begin{align}
    (x-z)^2 + (x-z')^2 = 2(x- (z+z')/2)^2 + (z-z')^2/2 \ .
\end{align}
Therefore,
\begin{align}
    \langle \chi_\Delta(z), \chi_\Delta(z')\rangle &=  \frac{1}{\sqrt{\pi}\Delta} \left(\int_{-\infty}^{\infty} e^{-(x-(z+z')/2)^2/(2\Delta^2)} dx\right) e^{-(z-z')^2/{4 \Delta^2)}}\\
    &= e^{-(z-z')^2/(4 \Delta^2)} \ , \label{eq: overlap chi z chi z'}
\end{align}
where we used that $\|\Psi_\Delta\|=1$.
\end{proof}

\begin{lemma}\label{lem: e^iQ vs e^iz} Let $\varepsilon\in(0,1/2)$ and $\Delta>0$. Then,
\begin{align}
    \delta_{z,z'} \left(1 - 5\varepsilon^2\right) \le \left\langle \chi_\Delta^\varepsilon(z'), e^{-i\pi z} e^{i\pi Q} \chi_\Delta^\varepsilon(z)\right\rangle \le \delta_{z,z'} \qquad\textrm{for all }\qquad z,z' \in \mathbb{Z}\ .
\end{align}
In particular, the overlap $\left\langle \chi_\Delta^\varepsilon(z'), e^{-i\pi z} e^{i\pi Q} \chi_\Delta^\varepsilon(z)\right\rangle $ is real.
\end{lemma}
\begin{proof}
First, we consider the case $z\neq z'$. By definition, the unitary $e^{i\pi Q}$ acts as a multiplication operator on $L^2(\mathbb{R})$. As the functions $\chi_\Delta^\varepsilon(z)$ and $\chi_\Delta^\varepsilon(z')$ have disjoint support, we conclude 
\begin{align}
    \left\langle \chi_\Delta^\varepsilon(z), e^{-i\pi z} e^{i\pi Q} \chi_\Delta^\varepsilon(z)\right\rangle = 0\qquad\textrm{if}\qquad z \neq z' \, .\label{eq: e^iQ z neq z'}
    \end{align}
Next, we consider the case $z=z'$.
By definition of the unitary $e^{i\pi Q}$, we have
\begin{align}
    \left\langle \chi_\Delta^\varepsilon(z), e^{-i\pi z} e^{i\pi Q} \chi_\Delta^\varepsilon(z)\right\rangle &=  \int_\mathbb{R}   \overline{\chi_\Delta^\varepsilon(z)(x)}e^{-i\pi z} e^{i\pi x} \chi_\Delta^\varepsilon(z)(x) dx \\
    &=\int_{z-\varepsilon}^{z+\varepsilon} |\chi_\Delta^\varepsilon(z)(x)|^2 e^{-i\pi z} e^{i\pi x}  dx \label{Eq: e^iQ vs e^iz 2}\\
    &= \int_{z-\varepsilon}^{z+\varepsilon} |\Psi_\Delta^\varepsilon(x-z)|^2 e^{i\pi (x-z)}  dx \\
    &=  \int_{-\varepsilon}^{\varepsilon} |\Psi_\Delta^\varepsilon(x)|^2 e^{i\pi x} dx \label{Eq: e^iQ vs e^iz 4}\\
    &= \int_{-\varepsilon}^{\varepsilon} |\Psi_\Delta^\varepsilon(x)|^2 \cos(\pi x) dx \label{Eq: e^iQ vs e^iz 5}
\end{align}
We obtained Eq.~\eqref{Eq: e^iQ vs e^iz 2} from the fact that the support of $\chi_\Delta^\varepsilon(z)$ is contained in the interval $[z-\varepsilon, z+\varepsilon]$. Eq.~\eqref{Eq: e^iQ vs e^iz 4} follows from a variable substitution and using that the support of $\Psi_\Delta^\varepsilon$ is contained in $[-\varepsilon,\varepsilon]$. Eq.~\eqref{Eq: e^iQ vs e^iz 5} follows from the fact that $|\Psi_\Delta^\varepsilon|^2$ is even and sinus is odd. We can bound~\eqref{Eq: e^iQ vs e^iz 5} as follows
\begin{align}
    \int_{-\varepsilon}^{\varepsilon} |\Psi_\Delta^\varepsilon(x)|^2 \cos(\pi x) dx &\ge  \cos(\pi \varepsilon) \int_{-\varepsilon}^{\varepsilon} |\Psi_\Delta^\varepsilon(x)|^2  dx \label{Eq: e^iQ vs e^iz 6}\\
    &=  \cos(\pi \varepsilon) \label{Eq: e^iQ vs e^iz 7}\\
    &\ge  1 - 5\varepsilon^2 \ . \label{Eq: e^iQ vs e^iz 8}\ .
\end{align}
Inequality~\eqref{Eq: e^iQ vs e^iz 6} follows from the fact that cosine is an even function monotonously decreasing on the interval $[0, \pi/2]$.  Eq.~\eqref{Eq: e^iQ vs e^iz 8} is a consequence of the bound $\cos x \ge 1 - x^2/2$ for all $x\in\mathbb{R}$. 
Finally, from Eq.~\eqref{eq: e^iQ z neq z'} and \eqref{Eq: e^iQ vs e^iz 5}, we conclude that $ \left\langle \chi_\Delta^\varepsilon(z), e^{-i\pi z} e^{i\pi Q} \chi_\Delta^\varepsilon(z)\right\rangle \in \mathbb{R}$ for all $z,z' \in \mathbb{Z}$.
\end{proof}

\subsubsection{Bounds on approximate comb states}
Let us recall the definitions of the approximate comb states (cf.~\eqref{eq:sha L}) and of the truncated approximate comb states (cf.~\eqref{eq:Sha L ep} and~\eqref{eq:Sha L ep sum}). Those are
\begin{align}
\ket{\Sha_{L,\Delta}} &=\frac{D_{L,\Delta}}{\sqrt{L}}\sum_{z=-L/2}^{L/2-1} \ket{\chi_\Delta(z)}\ , \label{eq:appx Sha L}
\end{align}
where $D_{L,\Delta}$ is a normalization factor and $\chi_{\Delta}$ are translated Gaussians as in~\eqref{eq:appx Psi Delta}, and
\begin{align}
        \ket{\Sha^{\varepsilon}_{L,\Delta}}&=\frac{1}{\sqrt{L}}\sum_{z=-L/2}^{L/2-1}\ket{\chi^\varepsilon_\Delta(z)}\ ,\label{eq:appx Sha L ep sum}
\end{align}
where $\chi_{\Delta}$ are truncated translated Gaussians as in~\eqref{eq:appx Psi Delta Ep}.

\begin{lemma} \label{lem: Sha - Sha epsilon overlap}
Let $\varepsilon\in(0, 1/2)$, $\Delta\in(0,1/4)$, and $L \in 2\mathbb{N}$.
Then
\begin{align}
    \left|\left\langle \Sha_{L, \Delta}, \Sha_{L, \Delta}^\varepsilon \right\rangle \right|^2 \ge 1 - 16\Delta^2 - 2e^{-(\varepsilon/\Delta)^2}\ .
\end{align}
\end{lemma}
\begin{proof}
By~\eqref{eq:appx Sha L} and~\eqref{eq:appx Sha L ep sum}, we have
\begin{align}
        \left\langle \Sha_{L, \Delta}, \Sha_{L, \Delta}^\varepsilon \right\rangle  
        &= \frac{ D_{L, \Delta}}{L}\sum_{z = -L/2}^{L/2-1}\sum_{z' = -L/2}^{L/2-1}  \langle \chi_\Delta(z), \chi_\Delta^\varepsilon(z') \rangle \\
        &\ge \frac{D_{L, \Delta}}{L}  \sum_{z = -L/2}^{L/2-1} \langle \Psi_\Delta , \Psi_\Delta^\varepsilon\rangle  \label{eq: first sha sha}\\
        & = D_{L, \Delta}  \langle \Psi_\Delta , \Psi_\Delta^\varepsilon\rangle \\
        & \ge D_{L, \Delta} \left( 1 - 2e^{-(\varepsilon/\Delta)^2}\right)^{1/2} \label{eq: second sha sha} \ ,
    \end{align}
    where inequality~\eqref{eq: first sha sha} follows from the non-negativity of $\chi_\Delta(z)(\cdot)$ and $\chi_\Delta^\varepsilon(z)(\cdot)$ and the equality $\langle \chi_\Delta(z), \chi_\Delta^\varepsilon(z) \rangle=\langle\Psi_\Delta,\Psi^\varepsilon_\Delta\rangle$ for all $z \in \mathbb{Z}$. The inequality~\eqref{eq: second sha sha} follows from Lemma~\ref{lem: gaussian gaussian epsilon}.
    
    We have
    \begin{align}
        D_{L,\Delta}^{-2} 
        = \frac{1}{L}\sum_{z = -L/2}^{L/2-1}\sum_{z' = -L/2}^{L/2-1}  \langle \chi_\Delta(z), \chi_\Delta(z') \rangle
        = \frac{1}{L}\sum_{z = -L/2}^{L/2-1}\sum_{z' = -L/2}^{L/2-1}  e^{-\frac{(z-z')^2}{4 \Delta^2}} \ ,\label{eq:DLDelta-2}
    \end{align}
    where we used that $\langle \chi_\Delta(z), \chi_\Delta(z')\rangle = e^{-\frac{(z-z')^2}{4 \Delta^2}}$ by Lemma~\ref{lem: overlap chi z chi z'}.
    We split the double sum into diagonal and off-diagonal parts, $S$ and $J$, respectively. This gives
    \begin{align}
        \sum_{z = -L/2}^{L/2-1}\sum_{z' = -L/2}^{L/2-1}  e^{-\frac{(z-z')^2}{4 \Delta^2}} = S+J\ ,\label{eq: comb double sum split}
    \end{align}
    where
    \begin{align}
        S &= \sum_{z = -L/2}^{L/2-1}  e^{-\frac{(z-z)^2}{4 \Delta^2}} = \sum_{z = -L/2}^{L/2-1} 1 =  L \label{eq: sh sha D is L}\\
        J &= \sum_{\substack{z,z'\in\{-L/2,\ldots, L/2-1\}\\ z\not=z'}} e^{-\frac{(z-z')^2}{4 \Delta^2}}
        \ .
    \end{align}
    Let us rewrite $J$ as
    \begin{align}
        J 
        &= 2 \sum_{z = -L/2}^{L/2-1}\sum_{z' = z+1}^{L/2-1} e^{-\frac{(z-z')^2}{4 \Delta^2}} \\
        &= 2  \sum_{z = -L/2}^{L/2-1} \sum_{ y=1}^{L/2-z-1} e^{-\frac{y^2}{4 \Delta^2}}\\
        &\le 2L \sum_{y=1}^{L-1} e^{-\frac{y^2}{4 \Delta^2}} \ .
        \label{eq: diag sum}
    \end{align}
    We will use that $y \le y^2$ for $y\ge 1$ and bound $J$ as follows.
    \begin{align}
        J 
        &\le 2L  \sum_{y=1}^{L-1} \left(e^{-\frac{1}{4 \Delta^2}}\right)^y \\
        &= 2L e^{-\frac{1}{4\Delta^2}} \frac{1 -  e^{-\frac{L-1}{4 \Delta^2}}}{1 - e^{-\frac{1}{4\Delta^2}}} \qquad\textrm{ by sum of geometric series}\\
        &\le 2L\frac{e^{-\frac{1}{4\Delta^2}}}{1 - e^{-\frac{1}{4\Delta^2}}}\ .
    \end{align}
    Notice that $e^{-1/(4\Delta^2)} \le 1/4$ because $0<\Delta<1/4$. Therefore, we can use the inequality $x/(1-x) \le 2x$ that holds for all $0\le x \le  1/2$, and obtain
    \begin{align}
        J &\le 2L\cdot 2e^{-\frac{1}{4\Delta^2}}\\
        &\le 16L \Delta^2 \qquad \textrm{ by $e^{-x} \le x^{-1}$ for all $x\ge 0$}\ .
    \end{align}
    Combining this bound with Eqs.~\eqref{eq: sh sha D is L},~\eqref{eq: comb double sum split}, and~\eqref{eq:DLDelta-2} gives
    \begin{align}
        D^{-2}_{L,\Delta} \le 1+16\Delta^2\ .
    \end{align}
    Finally, Eq.~\eqref{eq: second sha sha} together with the bound on $D^{-2}_{L,\Delta}$ gives
    \begin{align}
        \left|\left\langle \Sha_{L, \Delta}, \Sha_{L, \Delta}^\varepsilon \right\rangle \right|^2 &\ge \frac{1}{1 + 16\Delta^2} \left( 1 - 2e^{-(\varepsilon/\Delta)^2}\right) \\
        &\ge \left(1 - 16 \Delta^2\right) \left( 1 - 2e^{-(\varepsilon/\Delta)^2}\right) \\
        & \ge 1 - 16\Delta^2 - 2e^{-(\varepsilon/\Delta)^2}\ ,
    \end{align}
    where we used $(1+x)^{-1}\ge 1-x$ for all $x\ge 0$ on the first line and $(1-x)(1-y)\ge 1-x-y$ for all $x,y\ge 0$ on the last line.
\end{proof}

\begin{corollary}\label{cor: Sha - Sha epsilon trace distance}
    Let $\Delta\in(0,1/4)$, $\varepsilon\in[\sqrt{\Delta}, 1/2)$ and $L\in 2N$. Then
    \begin{align}
        \left\| \proj{ \Sha_{L,\Delta}} - \proj{\Sha_{L,\Delta}^{\varepsilon}}\right\|_1
        &\le 5 \sqrt{\Delta}\ .
    \end{align}
\end{corollary}
\begin{proof}
By Lemma~\ref{lem: Sha - Sha epsilon overlap}, we have
        \begin{align}
            \left| \langle \Sha_{L,\Delta}, \Sha_{L, \Delta}^\varepsilon\rangle \right|^2 &\ge 1 - 16\Delta^2 - 2e^{-(\varepsilon/\Delta)^2} \\
            &\ge 1 - 16 \Delta^2 - 2e^{-1/\Delta} \qquad \textrm{ by $\varepsilon \ge \sqrt{\Delta}$}\\
            &\ge 1 - 16\Delta^2 - 2\Delta \qquad   \textrm{ by $e^{-x} \le x^{-1}$  for $x>0$}  \\
            &\ge 1 - 6\Delta \qquad  \textrm{ by $\Delta<1/4$} \label{eq: overlap sha sha first}\ . 
        \end{align}
        Using the relation between the trace distance and the overlap (cf.~Eq.~\eqref{eq: trace distance overlap}) gives the claim
        \begin{align}
        \left\lVert  \proj{ \Sha_{L,\Delta}} - \proj{\Sha_{L,  \Delta}^{\varepsilon}}\right\rVert_1
         &\le 2\sqrt{6\Delta} \le 5 \sqrt{\Delta}\ .
        \end{align}
\end{proof}

\subsubsection{Bounds on approximate GKP states}
Let us recall the definition of the approximate $\ket{\gkp_{\kappa,\Delta}}$ state with parameters $\kappa$ and $\Delta$ (cf.\ Eq.~\eqref{eq:approximategkpstate eta chi} and Eq.~\eqref{eq:gkpkappdeltageneraldefinition}):
\begin{align}
    \ket{\gkp_{\kappa, \Delta}}&:= C_{\kappa,\Delta} \sum_{z\in\mathbb{Z}}\eta_\kappa(z)\ket{\chi_\Delta (z)} \ ,\label{eq:appx GKP eta chi}
\end{align}
where $C_{\kappa,\Delta}$ is the normalization factor.

It is convenient to define approximate GKP states with truncated peaks $\ket{\gkp^\varepsilon_{\kappa,\Delta}}$ and approximate GKP states with truncated peaks and compactly supported envelope $\ket{\gkp^\varepsilon_{L,\kappa,\Delta}}$. We define these (for $\varepsilon\in(0,1/2)$ and $\kappa,\Delta>0$) as
\begin{align}
\ket{\gkp_{\kappa, \Delta}^{\varepsilon}}&:= C_\kappa \sum_{z\in\mathbb{Z}}\eta_\kappa(z)\ket{\chi_\Delta^\varepsilon(z)} \ , \\
\ket{\gkp_{L,\kappa, \Delta}^{\varepsilon}}&:= C_{L,\kappa} \sum_{z=-L/2}^{L/2-1}\eta_\kappa(z)\ket{\chi_\Delta^\varepsilon(z)} \ ,
\label{eq:approximategkpstatedef}
\end{align}
where $C_\kappa$ and $C_{L,\kappa}$ are the respective normalization factors and
where we used 
\begin{align}
(\chi^\varepsilon_\Delta(z))(x)&:=\Psi^\varepsilon_\Delta(x-z)
\qquad \textrm{ and }\qquad \Psi^\varepsilon_\Delta:=
\frac{\Pi_{[-\varepsilon,\varepsilon]}\Psi_{\Delta}}{\|\Pi_{[-\varepsilon,\varepsilon]}\Psi_{\Delta}\|}\ 
\end{align}
as in Eq.~\eqref{eq:chiDelta ep definition}.

We display these approximate GKP states in Fig.~\ref{fig:GKP trunc bounded} and show by the following lemmas  and subsequent corollaries that they are close to state $\gkp_{\kappa,\Delta}$ for suitably chosen $\varepsilon$ and $L$.

\begin{figure}[!htb]
    \centering
    \begin{subfigure}[b]{0.45\textwidth}
        \centering
        \includegraphics[width= \textwidth]{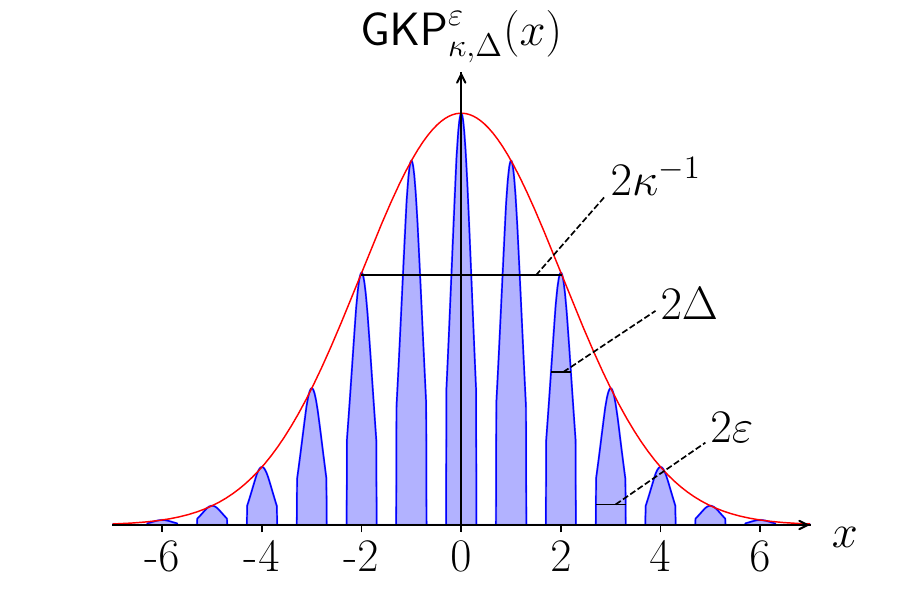}
    \caption{$\ket{\gkp_{\kappa, \Delta}^\varepsilon}$}
    \end{subfigure}
    \hfill
    \centering
    \begin{subfigure}[b]{0.45\textwidth}
        \centering
        \includegraphics[width= \textwidth]{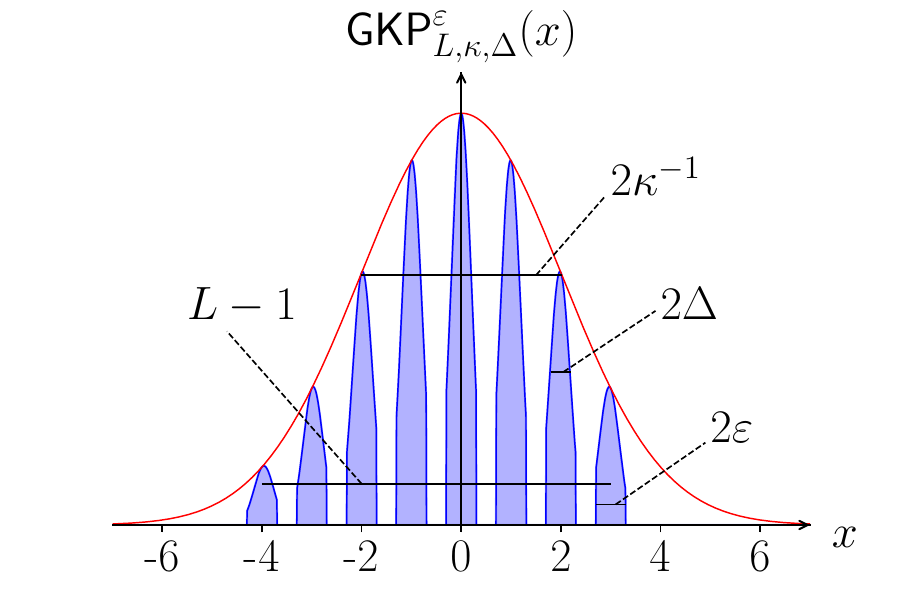}
    \caption{$\ket{\gkp_{L, \kappa,\Delta}^\varepsilon}$}
    \end{subfigure}
    \hfill
    \caption{Illustration of GKP states with truncated peaks $\ket{\gkp_{\kappa,\Delta}^\varepsilon}$ and with truncated peaks and compactly supported envelope $\ket{\gkp^\varepsilon_{L,\kappa,\Delta}}$. 
    }\label{fig:GKP trunc bounded}
\end{figure}

For future reference, observe that (by the orthogonality of the states~$\{\ket{\chi^\varepsilon_\Delta(z)}\}_{z\in\mathbb{Z}}$) the normalization constants satisfy
\begin{align}
    C_{\kappa}^{-2}&=\sum_{k\in\mathbb{Z}} \eta_\kappa(k)^2
    \label{eq:cknormbound}\\
    C_{L,\kappa}^{-2}&=\sum_{k=-L/2}^{L/2-1} \eta_\kappa(k)^2
    \label{eq:cknormboundlimit}
\end{align}

We can bound the quantities $C_{\kappa}^{-2},C_{\kappa,\Delta}^{-2}$ and $C_{L,\kappa}^{-2}$ related to normalization constants by the following technical lemma, which we will use repeatedly in our proofs.
\begin{lemma}\label{lem: technical lemma Ck CkD CkL}
    Let $\kappa,\Delta>0$. Then
    \begin{align}
        1 - \frac{\kappa}{\sqrt{\pi}}\le C_{\kappa}^{-2} &\le 1 + \frac{\kappa}{\sqrt{\pi}}\ \label{eq: Ck-2 bounds}\\
        C_{\kappa,\Delta}^{-2} &\le 1+\frac{\kappa}{\sqrt{\pi}}+2(\sqrt{2\pi}+\kappa)\Delta   \label{eq: CkD-2 bound}\\
        \left(1 - 2e^{-(\kappa L/2)^2}\right)\cdot\left(1 + \frac{\kappa}{\sqrt{\pi}}\right) \le C_{L,\kappa}^{-2}& \ .\label{eq: CLk-2 bound}
    \end{align}
    Furthermore we have
    \begin{align}
         1-\frac{2(\sqrt{2\pi}+\kappa)}{1-\kappa/\sqrt{\pi}}\Delta \ \leq \frac{C^{2}_{\kappa,\Delta}}{C_\kappa^2}\leq 1\ .\label{eq:comparisoninequalityckappakappadelta}
    \end{align}
\end{lemma}
\begin{proof}
We obtain the lower bound in Eq.~\eqref{eq: Ck-2 bounds} by using Lemma~\ref{lem: gaussian sum integral bound} (with $c=\kappa^2$):
\begin{align}
    C_{\kappa}^{-2} &= \sum_{z \in \mathbb{Z}} \eta_\kappa(z)^2 \\
    & = \frac{\kappa}{\sqrt{\pi}} \sum_{z\in\mathbb{Z}} e^{-\kappa^2 z^2}\\
    &\ge \frac{\kappa}{\sqrt{\pi}} \left( \frac{\sqrt{\pi}}{\kappa} - 1 \right) \\
    &= 1 - \frac{\kappa}{\sqrt{\pi}}\, . \label{eq: D lower bound}
\end{align}
We obtain the upper bound in Eq.~\eqref{eq: Ck-2 bounds} analogously (by using Lemma~\ref{lem: gaussian sum integral bound} (with $c=\kappa^2$)):
\begin{align}
    C_{\kappa}^{-2} &\le 1 + \frac{\kappa}{\sqrt{\pi}}\, . \label{eq: D upper bound}
\end{align}

We get the upper bound on $C_{\kappa,\Delta}^{-2}$ in Eq.~\eqref{eq: CkD-2 bound} by first noting that
\begin{align}
    C_{\kappa,\Delta}^{-2}
    &=\sum_{z,z' \in \mathbb{Z}}\eta_\kappa(z) \eta_\kappa(z') \langle \chi_\Delta(z), \chi_\Delta(z') \rangle\\
    &=\sum_{z \in \mathbb{Z}} \eta_\kappa(z)^2 + \sum_{z \in \mathbb{Z}} \sum_{\substack{z'\in\mathbb{Z}\\z'\neq z}} \eta_\kappa(z) \eta_\kappa(z') \langle \chi_\Delta(z), \chi_\Delta(z') \rangle\\
    &=C_{\kappa}^{-2}+J \ , \label{eq: CkL-2 = Ck-2 + J}
\end{align}
where we defined $J$ as the second term in~\eqref{eq: CkL-2 = Ck-2 + J}.
Because we  have $J\geq 0$  
(because each function~$\chi_\Delta(z)$ is non-negative), 
the second inequality in Eq.~\eqref{eq:comparisoninequalityckappakappadelta} follows.

Since $\langle \chi_\Delta(z), \chi_\Delta(z')\rangle = e^{-\frac{(z-z')^2}{4 \Delta^2}}$,  we have
\begin{align}
    J &= \sum_{z \in \mathbb{Z}} \eta_\kappa(z) \sum_{k\in\mathbb{Z}\setminus\{ 0\}} \eta_\kappa(z-k) e^{- k^2/(4\Delta^2)} \ \label{eq: offdiag terms}
\end{align}
by Lemma~\ref{lem: overlap chi z chi z'}. 
Since $\eta_\kappa(z') \le \sqrt{\kappa}/\pi^{1/4}$ for any $z'\in\mathbb{Z}$, we obtain 
\begin{align}
    J &\le \sum_{z \in \mathbb{Z}} \eta_\kappa(z) \sum_{k\in\mathbb{Z}\setminus\{0\}} \frac{\sqrt{\kappa}}{\pi^{1/4}} e^{-k^2/(2\Delta)^2} \\
    &\le  \left(\sum_{z \in \mathbb{Z}} \eta_\kappa(z) \right) \left( \frac{\sqrt{\kappa}}{\pi^{1/4}} 2\sqrt{\pi} \Delta \right) \qquad\textrm{ by Lemma~\ref{lem: gaussian sum integral bound} with $c=1/(2\Delta)^2$ \,,}\label{eq: overlap first bound}\\
    &\le \frac{\kappa}{\sqrt{\pi}}  \left(\frac{\sqrt{2 \pi}}{\kappa}+1\right) 2\sqrt{\pi} \Delta \qquad\textrm{ by Lemma~\ref{lem: gaussian sum integral bound} with $c=\kappa^2/2$\,, } \label{eq: overlap second bound}\\
    &= 2\left(\sqrt{2\pi} + \kappa\right)\Delta
    \, .\label{eq: O upper bound} 
\end{align}
We obtain the claim~\eqref{eq: CkD-2 bound} by inserting this and~\eqref{eq:cknormboundlimit} in~\eqref{eq: CkL-2 = Ck-2 + J}.
This further gives
\begin{align}
    \frac{C_{\kappa, \Delta}^{2}}{C_\kappa^{2}}=\frac{C_\kappa^{-2}}{C_{\kappa, \Delta}^{-2}}=\frac{C_\kappa^{-2}}{C_\kappa^{-2}+J}=(1+J/C_\kappa^{2})^{-1}\ge 1-J/C_\kappa^{-2}
\end{align}
by inequality $(1+x)^{-1}\ge 1-x$ for all $x>0$. Thus,
\begin{align}
    \frac{C_{\kappa, \Delta}^{2}}{C_\kappa^{2}}\ge 1-\frac{2(\sqrt{2\pi}+\kappa)\Delta}{1-\kappa/\sqrt{\pi}}\ .
\end{align}

We get the upper bound on $C_{L,\kappa}^{-2}$ in~\eqref{eq: CLk-2 bound} as follows.
By~\eqref{eq:cknormbound} and~\eqref{eq:cknormboundlimit}, we can write 
\begin{align}
    C_{L,\kappa}^{-2}=C_\kappa^{-2}-K \label{eq: CLk-2 Ck-2 K}\ ,
\end{align}
where
\begin{align}
    K&=\sum_{k=L/2+1}^\infty \eta_{\kappa}(-k)^2+\sum_{k=L/2}^\infty \eta_{\kappa}(k)^2 \ .
\end{align}
We have
\begin{align}
K&\le 2\sum_{k=L/2}^\infty \eta_\kappa(k)^2\qquad\textrm{ using $\eta_\kappa(-k)=\eta_\kappa(k)$\,,}\\
&= \frac{2\kappa}{\sqrt{\pi}}\sum_{k=L/2}^\infty e^{-\kappa^2 k^2}\\
&=\frac{2\kappa}{\sqrt{\pi}} e^{-(\kappa L/2)^2}\sum_{k=L/2}^\infty e^{-\kappa^2(k^2 - (L/2)^2)}\\
&\le \frac{2\kappa}{\sqrt{\pi}} e^{-(\kappa L/2)^2}\sum_{k=L/2}^\infty e^{-\kappa^2(k-L/2)^2} \qquad \textrm{ since $(k-a)^2\le k^2-a^2$ for $0<a\le k$\,,}\\
&= \frac{2\kappa}{\sqrt{\pi}} e^{-(\kappa L/2)^2}\sum_{r=0}^\infty e^{-\kappa^2 r^2}\label{eq: J upper} \ .
\end{align}
By Lemma~\ref{lem: gaussian sum integral bound} with $c=\kappa^2$, we have
\begin{align}
    \sum_{r=0}^\infty e^{-\kappa^2 r^2}
    &=
    \frac12 \left(1+\sum_{r=-\infty}^\infty e^{-\kappa^2 r^2}\right) \le \left(1+\frac{\sqrt{\pi}}{2\kappa}\right) \ .
\end{align}
By inserting this into Eq.~\eqref{eq: J upper}, we obtain
\begin{align}
    K &\le \frac{2\kappa}{\sqrt{\pi}} e^{-(\kappa L/2)^2} \left(1+\frac{\sqrt{\pi}}{2\kappa}\right) =e^{-(\kappa L/2)^2}\left(1+\frac{2\kappa}{\sqrt{\pi}}\right) \le 2e^{-(\kappa L/2)^2}\left(1+\frac{\kappa}{\sqrt{\pi}}\right)\ .
\end{align}
Inserting this bound and bound from~\eqref{eq: Ck-2 bounds} into Eq.~\eqref{eq: CLk-2 Ck-2 K} gives the claim~\eqref{eq: CLk-2 bound}.
\end{proof}

\begin{lemma}\label{lemma: overlap truncated gkp}
    Let $\kappa\in(0,1/4)$, $\Delta>0$ and $\varepsilon\in(0,1/2)$. Then
    \begin{align}
         \left|\left\langle \gkp_{\kappa, \Delta}, \gkp_{\kappa, \Delta}^\varepsilon \right\rangle \right|^2 \ge 1-7\Delta - 2e^{-(\varepsilon/\Delta)^2} \, .
    \end{align}
\end{lemma}

\begin{proof}
We have
\begin{align}
    \left\langle \gkp_{\kappa, \Delta}, \gkp_{\kappa, \Delta}^\varepsilon\right 
    \rangle &= C_{\kappa, \Delta} C_\kappa \sum_{z,z' \in \mathbb{Z}} \eta_\kappa(z) \eta_\kappa(z') \langle \chi_\Delta(z), \chi_\Delta^\varepsilon(z') \rangle\\
    &\ge C_{\kappa, \Delta} C_\kappa \sum_{z \in \mathbb{Z}} \eta_\kappa(z)^2 \langle \Psi_\Delta , \Psi_\Delta^\varepsilon\rangle  \label{eq: bound overlaps trunc non trunc}\\
    &= \frac{C_{\kappa, \Delta}}{C_\kappa} \langle \Psi_\Delta , \Psi_\Delta^\varepsilon\rangle\\
    &\ge \frac{C_{\kappa, \Delta}}{C_\kappa}\left(1-2e^{-(\varepsilon/\Delta)^2}\right)^{1/2}\ ,
    \label{eq:bound overlap gkp gkp eps first}
\end{align}
where the first inequality follows from non-negativity of $\chi_\Delta(z)(\cdot)$ and $\chi_\Delta^\varepsilon(z)(\cdot)$ and $\langle \chi_\Delta(z), \chi_\Delta^\varepsilon(z) \rangle=\langle \Psi_\Delta , \Psi_\Delta^\varepsilon\rangle$ for all $z \in \mathbb{Z}$. The last inequality is from Lemma~\ref{lem: gaussian gaussian epsilon}.

We will use Lemma~\ref{lem: technical lemma Ck CkD CkL} to obtain
\begin{align}
    \frac{C_{\kappa, \Delta}^{2}}{C_\kappa^{2}}
    &\ge 
    1-\frac{2(\sqrt{2\pi}+\kappa)}{1-\kappa/\sqrt{\pi}}\Delta \ .
\end{align}
Thus, by the assumption $\kappa<1/4$, we get
\begin{align}
    \frac{C_{\kappa, \Delta}^{2}}{C_\kappa^{2}}\ge 1-7\Delta\ .
\end{align}
Inserting this into the square of~\eqref{eq:bound overlap gkp gkp eps first} and using the inequality $(1-x)(1-y) \ge 1 - x - y$ for $x,y \ge 0$ implies the claim.
\end{proof}

\begin{corollary}\label{cor: gkp gkp eps trace distance}
    Let $\kappa \in (0,1/4)$, $\Delta>0$, and $\varepsilon\in [\sqrt{\Delta}, 1/2)$. Then
    \begin{align}
            \left\| \proj{\gkp_{\kappa,\Delta}} - \proj{\gkp_{\kappa,\Delta}^{\varepsilon}} \right\|_1
            &\le 6\sqrt{\Delta}\, .
        \end{align}
\end{corollary}
\begin{proof}
    By Lemma~\ref{lemma: overlap truncated gkp} and the assumption $\varepsilon \ge \sqrt{\Delta}$, we have 
        \begin{align}
            \left|\left\langle \gkp_{\kappa, \Delta}, \gkp_{\kappa, \Delta}^\varepsilon \right\rangle \right|^2
            &\ge 1-7\Delta -2e^{-(\varepsilon/\Delta)^2}\\
            &\ge 1-7\Delta - 2e^{-1/\Delta}\\
            &\ge 1 - 9\Delta \ ,
        \end{align}
        where the last inequality follows from the inequality $e^{-x} \le x^{-1}$ for $x>0$.
        Using the relation between the trace distance and the overlap (cf.~Eq.~\eqref{eq: trace distance overlap}) gives the claim
        \begin{align}
            \left\| \proj{\gkp_{\kappa,\Delta}} - \proj{\gkp_{\kappa,\Delta}^{\varepsilon}} \right\|_1
            &\le 2\sqrt{9\Delta } = 6\sqrt{\Delta}\ . 
        \end{align}
\end{proof}

\begin{lemma} \label{lem: overlap gkp L varepsilon varepsilon}
Let $\kappa\in(0,1/4)$, $\Delta>0$, $\varepsilon \in(0,1/2)$ and $L \in 2\mathbb{N}$.
Then
\begin{align}
\abs{\langle \gkp_{L,\kappa,\Delta}^{\varepsilon},\gkp_{\kappa,\Delta}^\varepsilon\rangle}^2 &\ge 1-2e^{-\kappa^2 L^2/4}\ .
\end{align}
\end{lemma}
\begin{proof}
By the orthogonality of the states~$\{\ket{\chi^\varepsilon_\Delta(z)}\}_{z\in\mathbb{Z}}$ for $\varepsilon<1/2$, and by~\eqref{eq:cknormboundlimit}, we have
\begin{align}
\langle \gkp_{L,\kappa,\Delta}^{\varepsilon},\gkp_{\kappa,\Delta}^\varepsilon\rangle
&=C_{\kappa}C_{L,\kappa}\sum_{k=-L/2}^{L/2-1}
\eta_{\kappa}(k)^2 = \frac{C_\kappa}{C_{L,\kappa}}\ .
\end{align}
Thus, Lemma~\ref{lem: technical lemma Ck CkD CkL} gives the claim
\begin{align}
    \frac{C_\kappa^{2}}{C_{L,\kappa}^{2}}=\frac{C_{L,\kappa}^{-2}}{C_\kappa^{-2}}\ge\frac{\left(1 - 2e^{-(\kappa L/2)^2}\right)\cdot\left(1 + \kappa/\sqrt{\pi}\right)}{1 + \kappa/\sqrt{\pi}}=1 - 2e^{-(\kappa L/2)^2}\ .
\end{align}
\end{proof}

\begin{corollary} \label{cor: overlap gkp L varepsilon varepsilon tr dist}
Assume $\kappa\in(0,1/4)$ and $L \in 2\mathbb{N}$.
We have
\begin{align}
    \left\|\proj{\gkp_{L,\kappa,\Delta}^{\varepsilon}}-\proj{\gkp_{\kappa,\Delta}^\varepsilon}\right\|_1 &\le 3 e^{-\kappa^2 L^2/8} \ .
\end{align}
\end{corollary}
\begin{proof}
    By Lemma~\ref{lem: overlap gkp L varepsilon varepsilon} and the relation between the trace distance and the overlap (cf.~Eq.~\eqref{eq: trace distance overlap}), we have
    \begin{align}
        \left\|\proj{\gkp_{L,\kappa,\Delta}^{\varepsilon}}-\proj{\gkp_{\kappa,\Delta}^\varepsilon}\right\|_1 &\le 2\sqrt{2e^{-\kappa^2 L^2/4}}\le 3 e^{-\kappa^2 L^2/8}\ . 
    \end{align}
\end{proof}

\begin{lemma} \label{lem: gkp S_P}
    Let $\kappa \in (0,1/4)$, $\Delta>0$ and $\varepsilon\in(0,1/2)$. Then
\begin{align}
    \left|\langle \gkp_{\kappa, \Delta}^\varepsilon,e^{-iP}\gkp_{\kappa, \Delta}^\varepsilon\rangle\right|^2 \ge 1 - 4\kappa \ .
\end{align} 
\end{lemma}

\begin{proof}
From pairwise orthogonality of the states $\{\ket{\chi_\Delta^\varepsilon(z)}\}_{z\in \mathbb{Z}}$ and $e^{-iP}\ket{\chi_{\Delta}^{\varepsilon}(z)}=\ket{\chi_{\Delta}^{\varepsilon}(z-1)}$, we have
\begin{align}
 \langle \gkp_{\kappa, \Delta}^\varepsilon,e^{-iP}\gkp_{\kappa, \Delta}^\varepsilon\rangle
 &= C_\kappa ^{2} \sum_{z,z'\in\mathbb{Z}}\eta_\kappa(z)\eta_\kappa(z')
 \langle \chi_\Delta^\varepsilon(z),\chi_\Delta^\varepsilon(z'-1)\rangle\\
 &= C_\kappa ^{2} \sum_{z\in\mathbb{Z}}\eta_\kappa(z)\eta_\kappa(z-1)\\
 &= C_\kappa ^{2} e^{-\kappa^2/4} \sum_{z\in\mathbb{Z}}\eta_\kappa(z-1/2)^2 \ ,\label{eq: appx GKP shift}
\end{align}
where the last step is by the definition of $\eta_\kappa$ (cf. Eq.~\eqref{eq:etakappadefinition}) that implies
\begin{align}
    \eta_\kappa(z)\eta_\kappa(z-1)&=\frac{\kappa}{\sqrt{\pi}}e^{-\kappa^2z^2/2}e^{-\kappa^2(z-1)^2/2}\\
    &=\frac{\kappa}{\sqrt{\pi}}e^{-\frac{\kappa^2}{2}\left(2(z-1/2)^2+1/2\right)}\\
    &=e^{-\kappa^2/4}\eta_\kappa(z-1/2)^2\ .
\end{align}
By definition of $\eta_\kappa$ and by using Lemma \ref{lem: gaussian sum integral bound} (with $c=\kappa^2$), we have
\begin{align}
    \sum_{z\in\mathbb{Z}}\eta_\kappa(z-1/2)^2  &\ge \frac{\kappa}{\sqrt{\pi}}\left(\frac{\sqrt{\pi}}{\kappa} -1\right) = 1-\frac{\kappa}{\sqrt{\pi}}\ .
\end{align}
Combining Eq.~\eqref{eq: appx GKP shift}, the upper bound on $C_{\kappa}^{-2}$ from Lemma~\ref{lem: technical lemma Ck CkD CkL} gives
\begin{align} 
    \langle \gkp_{\kappa, \Delta}^\varepsilon,e^{-iP}\gkp_{\kappa, \Delta}^\varepsilon\rangle
    &\ge e^{-\kappa^2/4} \frac{1- \kappa/\sqrt{\pi}}{1 + \kappa/\sqrt{\pi}}\\
    &\ge e^{-\kappa^2/4} \left(1 -2\frac{\kappa}{\sqrt{\pi}}\right) \qquad \textrm{by $\frac{1-x}{1+x} \ge 1-2x$ for $x>0$\,,}\\
    &\ge \left(1-\frac{1}{4}\kappa^2
    \right)\left(1 -\frac{2}{\sqrt{\pi}}\kappa\right)\qquad\textrm{ by $e^{-x} \ge 1 - x$\,,}\\
    &\ge 1-\frac{1}{4}\kappa^2 -\frac{2}{\sqrt{\pi}}\kappa\\
    &\ge 1-2\kappa \ ,
\end{align}
where we used the inequality $(1-x)(1-y) \ge 1- x- y$ for $x,y \ge 0$ on the penultimate line and the assumption $\kappa \le 1/4$ on the last line. The claim follows by $(1-x)^2\ge1-2x$ for all $x\in\mathbb{R}$.
\end{proof}

\begin{lemma} \label{lem: gkp S_Q}
     Let $\varepsilon\in(0,1/2)$ and $\kappa >0$. Then,
    \begin{align}
        \left|\langle \gkp_{\kappa, \Delta}^\varepsilon,e^{2\pi iQ}\gkp_{\kappa, \Delta}^\varepsilon\rangle\right|^2 \ge 1 - 40\varepsilon^2
    \end{align}
\end{lemma}

\begin{proof}
Since for $\varepsilon<1/2$ and $z,z'\in\mathbb{Z}$, $\chi_\Delta^\varepsilon(z)$ and $\chi_\Delta^\varepsilon(z')$, respectively $e^{2\pi i Q}\chi_\Delta^\varepsilon(z')$, have disjoint support for $z\neq z'$, we have
    \begin{align}
    \left\langle \chi_\Delta^\varepsilon(z'),  e^{2\pi i Q} \chi_\Delta^\varepsilon(z)\right\rangle &=\delta_{z,z'} \int_{z-\varepsilon}^{z+\varepsilon} \overline{\chi_\Delta^\varepsilon(z)(x)}e^{2\pi i x}\chi_\Delta^\varepsilon(z)(x) dx \\
    &= \delta_{z,z'} \int_{z-\varepsilon}^{z+\varepsilon}\overline{\Psi_\Delta^\varepsilon(x-z)} e^{2\pi i x} \Psi_\Delta^\varepsilon(x-z) dx \\
    &= \delta_{z,z'} \int_{-\varepsilon}^{\varepsilon} |\Psi_\Delta^\varepsilon(x)|^2 e^{2\pi i x} dx\\
    &= \delta_{z,z'} \int_{-\varepsilon}^{\varepsilon} |\Psi_\Delta^\varepsilon(x)|^2 \cos(2\pi x) dx \label{eq: symmetry cos}\ ,
\end{align}
where Eq.~\eqref{eq: symmetry cos} follows from $\Psi_\Delta^\varepsilon(-x) = \Psi_\Delta^\varepsilon(x)$ and the fact that sinus is an odd function.
Moreover, since cosine is an even function monotonously increasing on the interval $[-\pi, 0]$ and monotonously decreasing on the interval $[0,\pi]$, we get
\begin{align}
    \int_{-\varepsilon}^{\varepsilon} |\Psi_\Delta^\varepsilon(x)|^2 \cos(2\pi x) dx &\ge \cos(2\pi \varepsilon) \int_{-\varepsilon}^{\varepsilon} |\Psi_\Delta^\varepsilon(x)|^2  dx\\
    &=  \cos(2\pi \varepsilon) \\
    &\ge 1 - 20\varepsilon^2 \ ,
\end{align}
where we used the bound $\cos x \ge 1 - x^2/2$ for all $x\in\mathbb{R}$ to obtain the last inequality.
Therefore, we have
\begin{align}
    \langle \gkp_{\kappa, \Delta}^\varepsilon,e^{2\pi iQ}\gkp_{\kappa, \Delta}^\varepsilon\rangle &= C_\kappa^2 \sum_{z,z' \in \mathbb{Z}} \eta_\kappa(z') \eta_\kappa(z)  \left\langle \chi_\Delta^\varepsilon(z'),  e^{2\pi i Q} \chi_\Delta^\varepsilon(z)\right\rangle\\
    &\ge C_\kappa^2 \sum_{z\in\mathbb{Z}} \eta_\kappa(z)^2 (1 - 20 \varepsilon^2)\\
    &=C_\kappa^2 C_\kappa^{-2} (1 - 20 \varepsilon^2)\\
    &= 1 - 20 \varepsilon^2\ .
\end{align}
where we used the identity for $C_\kappa^{-2}$ from Eq.~\eqref{eq:cknormbound} on the penultimate line. The claim follows using $(1-x)^2\ge 1-2x$ for $x\in \mathbb{R}$.
\end{proof}

\begin{lemma} \label{lem: overlap gkp vareps proj pos}
Let $\kappa\in (0,1/4)$, $\Delta\in (0,1/4)$ and $\varepsilon\in (0,1/2)$. Then, 
\begin{align}
\bra{\gkp_{\kappa,\Delta}^\varepsilon}\Pi_{[-R, R]}\ket{\gkp_{\kappa,\Delta}^\varepsilon} \le 4\kappa R + 10\kappa \qquad\textrm{for any}\qquad R>0\ .
\end{align}
\end{lemma}
\begin{proof}
    Recall that the support of $\chi^\varepsilon_\Delta(z)$ is contained in the interval $[z-\varepsilon, z+\varepsilon]$ for all $z\in \mathbb{Z}$. Therefore, by definition of the state $\ket{\gkp_{\kappa,\Delta}^\varepsilon}$, we have
\begin{align}
    \bra{\gkp_{\kappa,\Delta}^\varepsilon}\Pi_{[-R,R]}\ket{\gkp_{\kappa,\Delta}^\varepsilon} &= C_\kappa^2 \sum_{z\in \mathbb{Z}} \eta_\kappa(z)^2 \int_{-R}^{R} \chi_\Delta^\varepsilon(z)^2(x) dx \\
    &\le C_\kappa^2 \sum_{z=-\ceil{R+1/2}}^{\ceil{R+1/2}}\eta_\kappa(z)^2 \\
    &\le C_\kappa^2 \sum_{z=-\ceil{R+1/2}}^{\ceil{R+1/2}} \kappa \qquad\textrm{since $\eta_\kappa(z)^2 \le \kappa/\sqrt{\pi}<\kappa$\,,}\\
    &\le C_\kappa^2 \left(2\kappa R+ 5\kappa\right)\label{eq: gkp hat eps overlap first}\ ,
\end{align}
where we used that $\ceil{R+1/2}\le R + 2 $ to obtain the last inequality.
Combining~\eqref{eq: gkp hat eps overlap first} and the lower bound on $C_{\kappa}^{-2}$ from Lemma~\eqref{lem: technical lemma Ck CkD CkL}, we have
\begin{align}
\bra{\gkp_{\kappa,\Delta}^\varepsilon}\Pi_{[-R,R]}\ket{\gkp_{\kappa,\Delta}^\varepsilon} \le  \left(1 - \frac{\kappa}{\sqrt{\pi}}\right)^{-1}\cdot \left(2\kappa R + 5\kappa\right)\le 3 \kappa R + 6\kappa\, ,
\end{align}
where we used $\kappa<1/4$ in the last step.
\end{proof}

\subsubsection{Bounds on ``point-wise'' approximate GKP states}

Recall the ``point-wise'' GKP state (cf.~\eqref{eq:postmeasurementstateyoutcomedef}):
\begin{align}
\tGKP^\varepsilon_{L,\kappa,\Delta}(x)
    &= D^\varepsilon_{L,\kappa,\Delta}\sum_{z=-L/2}^{L/2-1} 
    \eta_\kappa(x) \chi^\varepsilon_\Delta(z)(x)\ ,
\end{align}
where $D^\varepsilon_{L,\kappa,\Delta}$ is a normalization factor.

It is convenient to define
\begin{align}
    I_k(y)&:=
    \int
    \eta_\kappa(x-y)^2\chi^\varepsilon_\Delta(k)(x)^2 dx \label{eq:Ikdefinition}\\
    I'_k&:=
    \int \eta_\kappa(x)\eta_\kappa(k)\chi^\varepsilon_\Delta(k)(x)^2 dx\label{eq:Iprimekdefinition}\ 
\end{align}
for $k\in\mathbb{N}$ and $y\in\mathbb{R}$.

We prove the following relations between the ``point-wise'' GKP states and the ``peak-wise'' GKP states.
\begin{lemma}\label{lem:approximatetgkpLgkp} Let $\kappa>0$, $\Delta>0$ and $\varepsilon \in (0,1/2)$. Then,
\begin{align}
    \abs{\langle \tGKP^\varepsilon_{L,\kappa,\Delta},\gkp^\varepsilon_{L,\kappa,\Delta}\rangle}^2
    &\ge 1-4\kappa^2\varepsilon L\ .
    \end{align}
\end{lemma}
\begin{proof}
We note that the normalization constant~$D^\varepsilon_{L,\kappa,\Delta}$ for $\ket{\tGKP^\varepsilon_{L,\kappa,\Delta}}$ is 
\begin{align}
D^\varepsilon_{L,\kappa,\Delta}&=
\left(
\sum_{k_1=-L/2}^{L/2-1}
\sum_{k_2=-L/2}^{L/2-1}
M_{k_1,k_2}\right)^{-1/2}\ ,
\end{align}
where 
\begin{align}
M_{k_1,k_2}&:=\int \eta_\kappa(x)^2 \chi^\varepsilon_\Delta(k_1)(x)\chi^\varepsilon_\Delta(k_2)(x)
 dx\ .
\end{align}
Because of the factor 
$\chi^\varepsilon_\Delta(k_1)(x)\chi^\varepsilon_\Delta(k_2)$, $k_1,k_2$ being integers and $\varepsilon<1/2$
we have $M_{k_1,k_2}=0$ unless $k_1=k_2$. Furthermore, we have $M_{k,k}=I_k(0)$ thus 
\begin{align}
D^{\varepsilon}_{L,\kappa,\Delta}&=\left(\sum_{k=-L/2}^{L/2-1}I_k(0)\right)^{-1/2}\ .\label{eq:normexpressiondlkappadleta}
\end{align}
We have
\begin{align}
\langle \tGKP^\varepsilon_{L,\kappa,\Delta},\gkp^\varepsilon_{L,\kappa,\Delta}\rangle&=
C_{L,\kappa} D^{\varepsilon}_{L,\kappa,\Delta}
\sum_{k_1=-L/2}^{L/2-1} \sum_{k_2=-L/2}^{L/2-1}I'_{k_1,k_2}\ ,
\end{align}
where
\begin{align}
I'_{k_1,k_2}&=\int
\eta_\kappa(x)\eta_\kappa(k_2)\chi^\varepsilon_{\Delta}(k_1)(x)\chi^\varepsilon_\Delta(k_2)(x)dx\ .
\end{align}
The integral $I'_{k_1,k_2}$ vanishes unless $k_1=k_2$ because of the term~$\chi^\varepsilon_{\Delta}(k_1)(x)\chi^\varepsilon_\Delta(k_2)(x)$, and it follows that
\begin{align}
\langle \tGKP^\varepsilon_{L,\kappa,\Delta},\gkp^\varepsilon_{L,\kappa,\Delta}\rangle&=
C_{L,\kappa} D^{\varepsilon}_{L,\kappa,\Delta}\sum_{k=-L/2}^{L/2-1}
I'_k\ , \label{eq:innterproductgkpconvgkp}
\end{align}
where $I'_k=I'_{k,k}$ (see the definition~\eqref{eq:Iprimekdefinition}). 

Combining~\eqref{eq:innterproductgkpconvgkp} with~\eqref{eq:normexpressiondlkappadleta} and~\eqref{eq:cknormbound}
gives
\begin{align}
\langle \tGKP^\varepsilon_{L,\kappa,\Delta},\gkp^\varepsilon_{L,\kappa,\Delta}\rangle&=
\left(\sum_{k=-L/2}^{L/2-1} \eta_\kappa(k)^2\right)^{-1/2}
\cdot \left(\sum_{k=-L/2}^{L/2-1}I_k(0)\right)^{-1/2}\cdot \sum_{k=-L/2}^{L/2-1}
I'_k\\
&=
\left(
\frac{\sum_{k=-L/2}^{L/2-1}I_k(0)}{\sum_{k=-L/2}^{L/2-1}\eta_\kappa(k)^2}
\right)^{1/2}
\left(\frac{\sum_{k=-L/2}^{L/2-1}I_k'}{\sum_{k=-L/2}^{L/2-1}I_k(0)}\right)\ ,
\end{align}
By using Lemmas~\ref{lem:ikzerobound} and~\ref{lem:Ikprimeikzerocomparison} , we obtain
\begin{align}
    \langle \tGKP^\varepsilon_{L,\kappa,\Delta},\gkp^\varepsilon_{L,\kappa,\Delta}\rangle
    &\ge (1-2\kappa^2\varepsilon L)^{1/2}(1-\kappa^2\varepsilon L/2)\\
    &\ge 1-2\kappa^2\varepsilon L\ ,
\end{align}
where we used the inequality $(1-2x)^{1/2}(1-x/2)\ge 1-2x$ for $0\le x \le 1/2$.
The claim follows using $(1-x)^2\ge 1-2x$ for all $x\in \mathbb{R}$.
\end{proof}

\begin{corollary} \label{cor: approximate tgkpL GKP trace dist} Let $\kappa\in(0,1/4)$, $\Delta\in (0,1/4)$, and $\varepsilon\in(0,1/2)$. Then,
\begin{align}
    \left\|\proj{\tGKP^\varepsilon_{L,\kappa,\Delta}}-\proj{\gkp^\varepsilon_{L,\kappa,\Delta}}\right\|_1 
    &\le 4\kappa\sqrt{\varepsilon L} \qquad \textrm{ and } \label{eq:L1-dist tgkp GKP 1} \\
    \left\|\proj{\tGKP^\varepsilon_{L,\kappa,\Delta}}~~-~\proj{\gkp^\varepsilon_{\kappa,\Delta}}\right\|_1
    &\le 4\kappa\sqrt{\varepsilon L} + 3e^{-\kappa^2L^2/8}\ . \label{eq:L1-dist tgkp GKP 2} 
\end{align}
Furthermore if $\varepsilon\in[\sqrt{\Delta},1/2)$, we have
\begin{align}
    \left\|\proj{\tGKP^\varepsilon_{L,\kappa,\Delta}}-\proj{\gkp_{\kappa,\Delta}}\right\|_1
    &\le 3\kappa\sqrt{L} +6\sqrt{\Delta} + 3e^{-\kappa^2L^2/8}\ . \label{eq:L1-dist tgkp GKP 3} 
\end{align}
\end{corollary}
\begin{proof}
     Eq.~\eqref{eq:L1-dist tgkp GKP 1} is an immediate corollary of Lemma~\ref{lem:approximatetgkpLgkp} and the relation between the trace distance and the overlap (cf.~Eq.~\eqref{eq: trace distance overlap}).

     Eq.~\eqref{eq:L1-dist tgkp GKP 2} follows by application of the triangle inequality to the Eq.~\eqref{eq:L1-dist tgkp GKP 1} and the bound in Corollary~\ref{cor: overlap gkp L varepsilon varepsilon tr dist}.

     Eq.~\eqref{eq:L1-dist tgkp GKP 3} is obtained by application of the triangle inequality to the Eq.~\eqref{eq:L1-dist tgkp GKP 2} and the bound in Corollary~\ref{cor: gkp gkp eps trace distance}. We have
     \begin{align}
        \left\|\proj{\tGKP^\varepsilon_{L,\kappa,\Delta}}-\proj{\gkp_{\kappa,\Delta}}\right\|_1
        &\le 4\kappa\sqrt{\varepsilon L} +6\sqrt{\Delta} + 4e^{-\kappa^2L^2/8}\\
    &\le 3\kappa\sqrt{L} +6\sqrt{\Delta} + 4e^{-\kappa^2L^2/8}\ , 
    \end{align}
    where we used the assumption $\varepsilon<1/2$ to obtain the last inequality.
\end{proof}

\subsubsection{Bounds on approximate GKP states in momentum space}
In this section, we are concerned with the Fourier-transform (cf. Section~\ref{sec:energylowerbound}) of the wave function~$\ket{\gkp_{\kappa,\Delta}}\in L^2(\mathbb{R})$.
\begin{lemma} \label{lem: hat gkp(p)}
    Let $\kappa>0$ and $\Delta>0$. Then, we have
    \begin{align}
         \widehat{\gkp}_{\kappa,\Delta}(p) = 
         \sqrt{2\pi}C_{\kappa,\Delta} \sum_{z\in\mathbb{Z}} \eta_\Delta(p) \Psi_\kappa(p - 2\pi z) \, .
    \end{align}
\end{lemma}
\begin{proof}
We apply the Fourier transformation (cf. Eq.~\eqref{eq: def fourier}) to the definition of approximate GKP states in Eq.~\eqref{eq:appx GKP eta chi}. By linearity, we have
\begin{align}
    \widehat{\gkp}_{\kappa,\Delta}(p) &= C_{\kappa,\Delta} \widehat{\Psi}_\Delta(p) \sum_{z\in \mathbb{Z}} \eta_\kappa(z) e^{- i zp} \ .
    \end{align}
We will use the Poisson summation  formula (see e.g.~\cite{AntoniZygmund1935}), which states that the following holds for a Schwartz-function~$f$:
\begin{align}
    \sum_{z\in \mathbb{Z}} f(z) = \sqrt{2\pi} \sum_{z\in \mathbb{Z}} \widehat{f}(2\pi z) \ .
\end{align}
By applying it with $f(z) = \eta_\kappa(z)e^{-ipz}$, we get
\begin{align}
         \widehat{\gkp}_{\kappa,\Delta}(p)&= C_{\kappa,\Delta} \widehat{\Psi}_\Delta(p)\left( \sqrt{2\pi}\sum_{z\in \mathbb{Z}}  \widehat{\eta}_\kappa(p + 2\pi z)\right) \qquad\textrm{for}\qquad p\in \mathbb{R}\ .
\end{align}
From the definition of $\Psi_\Delta$ and $\eta_\kappa$ (see~\eqref{eq:chiDeltadefinition}  and~\eqref{eq:etakappadefinition}) along with 
\begin{align}
    \frac{1}{\sqrt{2\pi}}\int e^{-cx^2/2} e^{-ipx}dx = \frac{1}{\sqrt{a}} e^{-p^2/(2a)}\ ,
\end{align}
we obtain (by setting $a$ as $1/\Delta^2$ and $\kappa^2$, respectively) that
\begin{align}
     \widehat{\gkp}_{\kappa,\Delta}(p) &= \sqrt{2\pi} C_{\kappa,\Delta} \frac{\sqrt{\Delta}}{\sqrt{\pi \kappa}} e^{-\Delta^2 p^2/2}  \sum_{z\in \mathbb{Z}} e^{-(p+2\pi z)^2/(2\kappa^2)}\, \\
    &=\sqrt{2\pi} C_{\kappa,\Delta}\sum_{z\in\mathbb{Z}} \eta_\Delta(p) \Psi_\kappa(p - 2\pi z)\, .
\end{align}
\end{proof}
From Lemma~\ref{lem: hat gkp(p)}, we see that, in momentum space, the roles of $\kappa$ and $\Delta$ are interchanged; the spacing of the peaks is altered from $1$ to $2\pi$; and the envelope is acting ``point-wise'', similar to the wavefunction of the $\ket{\tGKP_{\kappa,\Delta}}$ state in position space.

In the remaining part of this section, we show properties of the truncated GKP state in momentum space $\ket{\gkp_{\kappa,\Delta}^{\widehat{\varepsilon}}}$ (cf.~\eqref{eq: def gkp vareps hat first}), which we recall to be defined pointwise in momentum space as
\begin{align}
\gkp_{\kappa,\Delta}^{\widehat{\varepsilon}}(p) =  D_{\kappa,\Delta}^{\varepsilon} \sum_{z\in\mathbb{Z}} \eta_\Delta(p) \chi_\kappa^\varepsilon(2\pi z)(p) \qquad\textrm{for}\qquad p \in \mathbb{R}\ .
\end{align}

It is convenient to first bound the quantity $(D_{\kappa,\Delta}^{\varepsilon})^{-2}$.
\begin{lemma} \label{eq:DkD ep 2 bound} Let $\varepsilon\in(0,1/2)$. Then
\begin{align}
    \frac{1}{2\pi}(1+6\sqrt{\pi}\Delta) \ge (D_{\kappa,\Delta}^{\varepsilon})^{-2} \ge \frac{1}{2\pi}(1-6\sqrt{\pi}\Delta)\ . \label{eq:bounds DkD ep}
\end{align}
\end{lemma}
\begin{proof}
We have that
    \begin{align}
    (D_{\kappa,\Delta}^{\varepsilon})^{-2}&= \sum_{z,z'\in\mathbb{Z}}\int_{-\infty}^\infty  \eta_\Delta(p)^2 \chi_\kappa^\varepsilon(2\pi z)(p) \chi_\kappa^\varepsilon(2\pi z')(p)dp\\
    &=\sum_{z\in\mathbb{Z}}\int_{2\pi z - \varepsilon}^{2\pi z +\varepsilon}  \eta_\Delta(p)^2 \chi_\kappa^\varepsilon(2\pi z)(p)^2dp\, , \label{eq: DkD ep first equality}
    \end{align}
where we used that the functions $\chi^\varepsilon_\kappa(2\pi z)$ and  $\chi^\varepsilon_\kappa(2\pi z')$ with $\varepsilon\in(0,1/2)$ have disjoint support for  $z\neq z'$. Note that 
\begin{align}
    \eta_\Delta(p) \ge \eta_\Delta(2\pi |z|+\varepsilon) \qquad \textrm{for \quad $p \in [2\pi z -\varepsilon, 2\pi z+ \varepsilon]$}\, . \label{eq: eta kappa monotonicity}
\end{align}
By this and the expression $\eta_\Delta(x)^2=\Delta/\sqrt{\pi} \cdot e^{-\Delta^2 x^2}$, we have
\begin{align}
    (D_{\kappa,\Delta}^{\varepsilon})^{-2}& \ge \sum_{z\in\mathbb{Z}} \eta_\Delta(2\pi |z|+\varepsilon)^2  = \frac{1}{2\pi}\sum_{z\in\mathbb{Z}} \eta_{2\pi\Delta}(|z|+\varepsilon/(2\pi))^2 \, .
\end{align}
Lemma~\ref{lem: gaussian sum integral bound} (with $c=(2\pi\Delta)^2$) and the assumption $\varepsilon\in(0,1/2)$ imply the lower bound of the claim~\eqref{eq:bounds DkD ep}:
\begin{align}
    (D_{\kappa,\Delta}^{\varepsilon})^{-2} \ge 
    \frac{1}{2\pi} \frac{2\pi\Delta}{\sqrt{\pi}}
    \left(\frac{\sqrt{\pi}}{2\pi \Delta}-3\right)
    &=\frac{1}{2\pi}\left(1-6\sqrt{\pi}\Delta\right) \ . 
\end{align}

We obtain the upper bound in the claim~\eqref{eq:bounds DkD ep} analogously as follows.
Similarly to~\eqref{eq: eta kappa monotonicity}, we have
\begin{align}
    \eta_\Delta(p)\le \eta_\Delta(2\pi |z|-\varepsilon) \qquad \textrm{for $p\in[2\pi z-\varepsilon,2\pi z+\varepsilon], z\in  \mathbb{Z}\setminus\{0\}$}\ .
\end{align}
By this, by~\eqref{eq: DkD ep first equality}, by the fact that $\eta_\Delta(p)\le\eta_\Delta(0)$ for all $p$ , and by the definition of $\eta_\Delta(\cdot)^2$, we have
\begin{align}
    (D_{\kappa,\Delta}^{\varepsilon})^{-2}
    & \le \eta_{\Delta}(0)^2 +  \sum_{z\in\mathbb{Z}\setminus\{0\}} \eta_\Delta(2\pi |z|-\varepsilon)^2  = \eta_{\Delta}(0)^2 + \frac{1}{2\pi}\sum_{z\in\mathbb{Z}\setminus\{0\}} \eta_{2\pi\Delta}(|z|-\varepsilon/(2\pi))^2 \, .\label{eq: D const bound one}
\end{align}
Therefore, by the bound on this sum from Lemma~\ref{lem: gaussian sum integral bound}
(with $c=(2\pi\Delta)^2$) and by the definition of $\eta_\Delta(\cdot)^2$, we have
\begin{align}
    (D_{\kappa,\Delta}^{\varepsilon})^{-2}
    \le \eta_{\Delta}(0)^2+\frac{1}{2\pi}\frac{2\pi\Delta}{\sqrt{\pi}}\left( \frac{\sqrt{\pi}}{2\pi\Delta}+2 \right) 
    =\frac{1}{2\pi}(1+6\sqrt{\pi}\Delta) \ .
\end{align}
\end{proof}

\begin{lemma} \label{lem: Omega proj mom}
    Let $\kappa\in (0,1/4)$, $\Delta\in (0,1/100)$ and $\varepsilon\in (0,1/2)$. Then,
\begin{align}
\langle \gkp_{\kappa,\Delta}^{\widehat{\varepsilon}}|\widehat{\Pi}_{[-R,R]}|\gkp_{\kappa,\Delta}^{\widehat{\varepsilon}} \rangle \le 2 \Delta R + 12 \Delta
\end{align}
for any $R>0$.
\end{lemma}

\begin{proof}
    The support of $\chi^\varepsilon_\kappa(2\pi z)$ is contained in the interval $[2\pi z-\varepsilon, 2\pi z+\varepsilon]$ for all $z\in \mathbb{Z}$. Hence, by definition of the state $\ket{\gkp_{\kappa,\Delta}^{\widehat{\varepsilon}}}$, we have
\begin{align}
    \langle\gkp_{\kappa,\Delta}^{\widehat{\varepsilon}}|\widehat{\Pi}_{[-R, R]}|\gkp_{\kappa,\Delta}^{\widehat{\varepsilon}} \rangle &= (D_{\kappa,\Delta}^{\varepsilon})^{2} \sum_{z\in\mathbb{Z}}\int_{-R}^{R} \eta_\Delta(p)^2 \chi_\kappa^\varepsilon(2\pi z)(p)^2 dp\\
    &\le (D_{\kappa,\Delta}^{\varepsilon})^{2} \sum_{z\in\mathbb{Z}}\int_{-2\pi  \ceil{R/ (2\pi)}-\varepsilon}^{2\pi  \ceil{R/ (2\pi)}+\varepsilon} \eta_\Delta(p)^2 \chi_\kappa^\varepsilon(2\pi z)(p)^2 dp\\
    &=  (D_{\kappa,\Delta}^{\varepsilon})^{2} \sum_{z=- \ceil{R/ (2\pi)}}^{ \ceil{R/ (2\pi)}} \int_{2\pi z-\varepsilon}^{2\pi z +\varepsilon} \eta_\Delta(p)^2 \chi_\kappa^\varepsilon(2\pi z)(p)^2 dp \ . \label{eq: omega overlap first}
\end{align}
Notice that $\eta_\Delta$ is an even function monotonously decreasing on the interval $[0,
\infty)$, it holds for $p\in [2\pi z - \varepsilon,2\pi z +\varepsilon]$ and $ z \in \mathbb{Z}\setminus \{0\}$ that
\begin{align}
    \eta_\Delta(p)^2 &\le \eta_\Delta(2\pi |z| - \varepsilon)^2 \le \eta_\Delta(2\pi (|z| - 1))^2 \ .
\end{align}
Therefore
\begin{align}
    \langle\gkp_{\kappa,\Delta}^{\widehat{\varepsilon}}|\widehat{\Pi}_{[-R, R]}|\gkp_{\kappa,\Delta}^{\widehat{\varepsilon}} \rangle 
     &\le (D_{\kappa,\Delta}^{\varepsilon})^{2}\left(\eta_\Delta(0)^2 + 2\!\!\!\!\sum_{z=1}^{ \ceil{R/ (2\pi)}} \!\!\!\!\eta_\Delta(2\pi(z-1))^2 \int_{2\pi z-\varepsilon}^{2\pi z +\varepsilon}  \!\!\!\!\!\!\chi_\kappa^\varepsilon(2\pi z)(p)^2 dp \right)\\
     &= (D_{\kappa,\Delta}^{\varepsilon})^{2}\left(\eta_\Delta(0)^2 + 2\sum_{z=1}^{ \ceil{R/ (2\pi)}} \eta_\Delta(2\pi(z-1))^2 \right)\\
     &= (D_{\kappa,\Delta}^{\varepsilon})^{2}\left(\eta_\Delta(0)^2 + 2\sum_{z=0}^{ \ceil{R/ (2\pi)}-1} \eta_\Delta(2\pi z)^2 \right)\ , \label{eq: overlap mom first}
\end{align}
where we used that $\chi^\varepsilon_\kappa(2\pi z)$ are normalized. Clearly, we have
\begin{align}
    \eta_\Delta(p)^2 \le \eta_\Delta(0)^2 = \frac{\Delta}{\sqrt{\pi}}\qquad\textrm{for all}\qquad p \in \mathbb{R} \ .
\end{align}
By using this bound on each term in the sum
in Eq.~\eqref{eq: overlap mom first}, we obtain
\begin{align}
    \langle\gkp_{\kappa,\Delta}^{\widehat{\varepsilon}}|\widehat{\Pi}_{[-R, R]}|\gkp_{\kappa,\Delta}^{\widehat{\varepsilon}} \rangle 
     &\le (D^\varepsilon_{\kappa,\Delta})^2 \frac{\Delta}{\sqrt{\pi}}\left(1+ 2\ceil{R/ (2\pi)}\right) \label{eq: overlap mom sec}
\end{align}
Combining ~\eqref{eq: overlap mom sec}, the fact that $\ceil{R/(2\pi)}\le R/(2\pi)+1$, and the bound from Lemma~\ref{eq:DkD ep 2 bound} gives the claim 
\begin{align}
    \langle \gkp_{\kappa,\Delta}^{\widehat{\varepsilon}}|\widehat{\Pi}_{[-R+s, R+s]}|\gkp_{\kappa,\Delta}^{\widehat{\varepsilon}} \rangle 
    \le 2\pi\left(1-6\sqrt{\pi}\Delta\right)^{-1} \cdot \left( \frac{\Delta R}{\pi^{3/2}} + \frac{3\Delta}{\sqrt{\pi}} \right) 
    \le 2 \Delta R + 12 \Delta \ ,
\end{align}
where we used $\Delta < 1/100$ in the last step.
\end{proof}

\begin{lemma} \label{lem: omgea eps gkp overlap} Let $\kappa \in (0,1/4)$, $\Delta \in (0,1/4)$ and $\varepsilon \in [\sqrt{\Delta},1/2)$. Then, 
\begin{align}
    \left|\langle \gkp_{\kappa,\Delta}, \gkp_{\kappa,\Delta}^{\widehat{\varepsilon}} \rangle \right|^2 \ge 1-3\kappa-39\Delta \ .
\end{align}
\end{lemma}
\begin{proof}
    By Lemma~\ref{lem: hat gkp(p)} and by the fact that the Fourier transformation is unitary, we get
    \begin{align}
         \langle \gkp_{\kappa,\Delta}, \gkp_{\kappa,\Delta}^{\widehat{\varepsilon}} \rangle &= \sqrt{2\pi} C_{\Delta,\kappa} D_{\kappa,\Delta}^{\varepsilon}\sum_{z,z'\in\mathbb{Z}}\int  \eta_\Delta(p)^2 \chi_\kappa(2\pi z)(p) \chi_\kappa^\varepsilon(2\pi z')(p) dp\\
         &\ge \sqrt{2\pi} C_{\Delta,\kappa}  D_{\kappa,\Delta}^{\varepsilon} \sum_{z\in\mathbb{Z}} \int  \eta_\Delta(p)^2 \chi_\kappa(2\pi z)(p)\chi_\kappa^\varepsilon(2\pi z)(p)dp\ ,
         \end{align}
         where we used the non-negativity of the integrands. We have
         \begin{align}
             \int  
             \eta_\Delta(p)^2 \chi_\kappa(2\pi z)(p)\chi_\kappa^\varepsilon(2\pi z) dp  
             &= \int 
             \eta_\Delta(p)^2 \Psi_\kappa(p-2\pi z) \Psi_\kappa^\varepsilon(p-2\pi z) dp\\
             &= \int_{ -\varepsilon}^{ \varepsilon}  \eta_\Delta(u+ 2\pi z)^2 \Psi_\kappa(u)\Psi_\kappa^\varepsilon(u) du \ .
         \end{align}  where we used that the support of $\Psi_\kappa^\varepsilon$ is contained in $[-\varepsilon, \varepsilon]$ and substituted $u:=p-2\pi z$ in the last step. Due to the symmetry and monotonicity of $\eta_\Delta$, we infer that
         \begin{align}
             \eta_\Delta(u + 2\pi z) \ge \eta_\Delta(2\pi|z| + \varepsilon) 
             \qquad\textrm{for}\quad u \in [-\varepsilon, \varepsilon], \ z\in \mathbb{Z}\ .
         \end{align}
         Therefore, we can bound
         \begin{align}
             \langle \gkp_{\kappa,\Delta}, \gkp_{\kappa,\Delta}^{\widehat{\varepsilon}} \rangle  &\ge \sqrt{2\pi} C_{\kappa,\Delta}  D_{\kappa,\Delta}^{\varepsilon} \sum_{z\in\mathbb{Z}} \eta_\Delta(2\pi|z|+\varepsilon)^2 \int_{-\varepsilon}^\varepsilon  \Psi_\kappa(p) \Psi_\kappa^\varepsilon(p) dp\\
             &=  \sqrt{2\pi} C_{\kappa,\Delta} D_{\kappa,\Delta}^{\varepsilon} \sum_{z\in\mathbb{Z}} \eta_\Delta(2\pi |z|+\varepsilon)^2  \langle \Psi_\kappa, \Psi_\kappa^\varepsilon\rangle \ . \label{eq: first overlap omega eps gkp}
         \end{align}
         We bound the squares of the four factors in~\eqref{eq: first overlap omega eps gkp} separately.
         By Lemma~\ref{lem: technical lemma Ck CkD CkL} and by inequality $(1+x)^{-1}\ge 1-x$ for $x>0$ ,
         we have
         \begin{align}
              C_{\kappa,\Delta}^{2} 
              &\ge 1-\frac{\kappa}{\sqrt{\pi}}-2(\sqrt{2\pi}+\kappa)\Delta \\
              &\ge 1 - \kappa - 6\Delta \ , \label{eq: hat C Delta kappa lower}
         \end{align}
         where we used the assumption $\kappa<1/4$ in the last inequality.
         By Lemma~\ref{eq:DkD ep 2 bound},
         we have
         \begin{align}
             (D_{\kappa,\Delta}^{\varepsilon})^2 \ge 2\pi(1+6\sqrt{\pi}\Delta)^{-1}\ge 2\pi(1-6\sqrt{\pi}\Delta) \ge 2\pi(1-11\Delta)\ , \label{eq: D kappa delta eps lower}
         \end{align}
         where we used $(1+x)^{-1}\ge 1-x$ for all $x>0$. By the definition of $\eta_\Delta(\cdot)$, $\varepsilon\in(0,1/2)$ and by Lemma~\ref{lem: gaussian sum integral bound} (with $c=2\pi\Delta$), we have
         \begin{align}
             \sum_{z\in\mathbb{Z}} \eta_\Delta(2\pi |z|+\varepsilon)^2 
             &=\frac{1}{2\pi}\sum_{z\in\mathbb{Z}} \eta_{2\pi\Delta}( |z|+\varepsilon/(2\pi))^2\\
             &\ge \frac{1}{2\pi}\frac{2\pi\Delta}{\sqrt{\pi}}{}\left(\frac{\sqrt{\pi}}{2\pi\Delta}-3\right)\\
             &=\frac{1}{2\pi}\left(1-6\sqrt{\pi}\Delta\right) \\
             &\ge\frac{1}{2\pi}\left(1-11\Delta\right)\ .
         \end{align}
         Thus,
         \begin{align}
             \left(\sum_{z\in\mathbb{Z}} \eta_\Delta(2\pi |z|+\varepsilon)^2 \right)^2\ge \frac{1}{(2\pi)^2}\left(1-12\sqrt{\pi}\Delta\right) \ge \frac{1}{(2\pi)^2}\left(1-22\Delta\right) \label{eq: sum eta Delta lower} \ .
         \end{align}
Finally, by Lemma~\ref{lem: gaussian gaussian epsilon}, we have for $\varepsilon\in [\sqrt{\kappa},1/2)$ that 
\begin{align}
    \left|\langle \Psi_\kappa, \Psi_\kappa^\varepsilon\rangle \right|^2 
    &\ge 1 - 2 \kappa\ .\label{eq: overlap gauss gauss eps}
\end{align}
Combining the bounds from~\eqref{eq: first overlap omega eps gkp} with~\eqref{eq: hat C Delta kappa lower},~\eqref{eq: D kappa delta eps lower},~\eqref{eq: sum eta Delta lower}, and~\eqref{eq: overlap gauss gauss eps} and repeatedly using~$(1-x)(1-y) \ge 1 -x-y$ for $x,y
\ge 0$ gives
\begin{align}
     \left|\langle \gkp_{\kappa,\Delta}, \gkp_{\kappa,\Delta}^{\widehat{\varepsilon}} \rangle 
     \right|^2 
     &\ge (1-\kappa -6\Delta)(1-11\Delta)(1-22\Delta)(1-2\kappa)\\
     &\ge 1-3\kappa-39\Delta \ .
\end{align}
This is the claim.
\end{proof}
         
\begin{corollary}\label{cor: dist gkp omega eps}
     Let $\kappa \in (0,1/4)$, $\Delta \in (0,1/4)$ and $\varepsilon \in [\sqrt{\Delta},1/2)$. Then, 
\begin{align}
    \left\| \proj{\gkp_{\kappa,\Delta}}- \proj{\gkp_{\kappa,\Delta}^{\widehat{\varepsilon}}}  \right\|_1 \le  4\sqrt{\kappa} + 13\sqrt{\Delta}\ .
\end{align}
\end{corollary}
\begin{proof}
    From Lemma~\ref{lem: omgea eps gkp overlap} and from the relation between the trace distance and the overlap (cf.~Eq.~\eqref{eq: trace distance overlap}), we have
    \begin{align}
         \left\| \proj{\gkp_{\kappa,\Delta}}- \proj{\gkp_{\kappa,\Delta}^{\widehat{\varepsilon}}}  \right\|_1 \le 2\sqrt{3\kappa} + 2\sqrt{39\Delta}\le 4\sqrt{\kappa} + 13\sqrt{\Delta} \ ,
    \end{align}
    where we used $\sqrt{x+y} \le \sqrt{x}+\sqrt{y}$ for $x,y\ge 0$.
\end{proof}

\section{Bounds on convolutions}

In this section, we prove bounds on and relations between convolutions of Gaussians with summands (indexed by $k\in\mathbb{Z}$) of the following form (cf.~\eqref{eq:Ikdefinition} and~\eqref{eq:Iprimekdefinition}, where we used the definition of translated Gaussians~\eqref{eq:chiDelta ep definition}). 
\begin{align}
I_k(y)&=\int
\eta_\kappa(x-y)^2\Psi^\varepsilon_\Delta(x-k)^2 dx \qquad \textrm{ where $y\in\mathbb{R}$, and}\label{eq:Ikdefinition2}\\
I'_k&=
\int \eta_\kappa(x)\eta_\kappa(k)\Psi^\varepsilon_\Delta(x-k)^2 dx\label{eq:Iprimekdefinition2}\ .
\end{align}

For future reference, we observe that for any integer~$m\in\mathbb{Z}$ and $\delta\in\mathbb{R}$, it holds by the variable substitution~$z:=x-m$ that
\begin{align}
I_k(m+\delta)&=
\int \eta_\kappa(x-(m+\delta))^2\Psi^\varepsilon_\Delta(x-k)^2 dx\\
&=\int \eta_\kappa(z-\delta)^2\Psi^\varepsilon_\Delta(z-(k-m))^2 dz\\
&=I_{k-m}(\delta)\ .\label{eq:ikmdeltarelation}
\end{align}
Furthermore, by the variable substitution~$z:=x-k$, we obtain the identities
\begin{align}
    I_k(s) 
  &=\int_{-\varepsilon}^\varepsilon  \eta_\kappa((k+z)-s)^2 \Psi^\varepsilon_\Delta(z)^2 dz\ ,  \label{eq:ikdeltaexpressionh}\\
  I_k'
    &=\int_{-\varepsilon}^\varepsilon \eta_\kappa(k+z)\eta_\kappa(k) \Psi^\varepsilon_\Delta(z)^2 dz\ .\label{eq:ikprime}
    \end{align}

In our proofs, we will frequently use the following two identities regarding the product of shifted Gaussians:
\begin{align}
\eta_\kappa(k+z)^2 &=\frac{\kappa}{\sqrt{\pi}}
e^{-\kappa^2 (k+z)^2}\\
&=\frac{\kappa}{\sqrt{\pi}} e^{-\kappa^2 k^2}\cdot 
e^{-2\kappa^2 k z}\cdot e^{-\kappa^2 z^2}\\
&=\eta_\kappa(k)^2\cdot e^{-2\kappa^2 k z}\cdot e^{-\kappa^2 z^2}\ \label{eq:shifted eta2}
\end{align}
and
\begin{align}
\eta_\kappa(k+z)\eta_\kappa(k)&=\frac{\kappa}{\sqrt{\pi}}e^{-\kappa^2(k+z)^2/2}\cdot e^{-\kappa^2k^2/2}\\
&=\frac{\kappa}{\sqrt{\pi}}e^{-\kappa^2 (k+z)^2}\cdot 
e^{\kappa^2(k+z)^2/2-\kappa^2 k^2/2}\\
&=\frac{\kappa}{\sqrt{\pi}} e^{-\kappa^2 (k+z)^2}\cdot
e^{\kappa^2 k z+\kappa^2 z^2/2}\\
&=\eta_\kappa(k+z)^2 \cdot e^{\kappa^2 k z}\cdot  e^{\kappa^2 z^2/2}\ .\label{eq:shifted eta eta}
\end{align}

We prove the following statements.
\begin{lemma}\label{lem:ikzerobound} Let $I_k(0)$ be as in Eq.~\eqref{eq:Ikdefinition} with $\kappa>0$, $\Delta>0$, and $\varepsilon\in(0,1/2)$. Let $L\in 2\mathbb{N}$. Then,
\begin{align}
    \sum_{k=-L/2}^{L/2-1}I_k(0) &\geq \left(1-2\kappa^2\varepsilon L\right)\cdot \sum_{k=-L/2}^{L/2-1}\eta_\kappa(k)^2\ .\label{eq:ikzeroexp}
    \end{align}
\end{lemma}
\begin{proof}

We will prove the claim by bounding integrands in $I_k(0)$ (cf. identity~\eqref{eq:ikdeltaexpressionh} with $s=0$).
We use that
\begin{align}
e^{-2\kappa^2 k z}\cdot e^{-\kappa^2 z^2}&\geq e^{-\kappa^2 L\varepsilon}\cdot e^{-\kappa^2 \varepsilon^2}\ ,
\end{align}
for $z\in (-\varepsilon,\varepsilon)$ and $k\in \{-L/2,\ldots,L/2-1\}$ together with Eq.~\eqref{eq:shifted eta2} to obtain
\begin{align}
    \eta_\kappa(k+z)^2&\geq \eta_\kappa(k)^2\cdot e^{-\kappa^2\varepsilon(L+ \varepsilon)}\ .
\end{align}
By inserting this lower bound into the identity for $I_k(y)$ in Eq.~\eqref{eq:ikdeltaexpressionh}, we obtain (for $k\in\{-L/2,\ldots,L/2-1\}$) 
\begin{align}
I_k(0)&\geq e^{-\kappa^2\varepsilon( L+ \varepsilon)}\eta_\kappa(k)^2\int_{-\varepsilon}^\varepsilon \Psi^\varepsilon_\Delta(x)^2  dz\\
&=e^{-\kappa^2\varepsilon( L+ \varepsilon)}\eta_\kappa(k)^2\ ,
\end{align} 
where we used the fact that  $\Psi^\varepsilon_\Delta\in L^2(\mathbb{R})$ is normalized and has support on~$(-\varepsilon,\varepsilon)$.  Summing over $k\in \{-L/2,\ldots,L/2-1\}$ yields
\begin{align}
\sum_{k=-L/2}^{L/2-1} I_k(0)&\geq e^{-\kappa^2\varepsilon( L+ \varepsilon)} \sum_{k=-L/2}^{L/2-1}\eta_\kappa(k)^2\ .\label{eq:intermediatesumexpre}
\end{align}
By the inequality $e^{-x}\ge 1-x$ for $x\ge 0$ and by the assumption on $\varepsilon$,
we obtain
\begin{align}
e^{-\kappa^2\varepsilon(L+ \varepsilon)}&\geq 1-\kappa^2\varepsilon(L+\varepsilon)\\
&\geq 1-2\kappa^2\varepsilon L\ ,
\end{align}
which together with Eq.~\eqref{eq:intermediatesumexpre} proves the claim.
\end{proof}

\begin{lemma}\label{lem:fractionalboundmnv} Let $I_k(0)$ be as defined in Eq.~\eqref{eq:Ikdefinition} with $\kappa>0$, $\Delta>0$, and $\varepsilon\in(0,1/2)$. Let $L\in2\mathbb{N}$ such that $L\ge 4$.
We have 
\begin{align}
\sum_{k=-L/2}^{L/2-1}I_k(0)&\geq
(1-e^{-\kappa^2 L^2/8}) \sum_{k=-L}^{L-1}I_k(0)\,.
\end{align}
\end{lemma}
\begin{proof}
We will prove the claim by comparing integrands in $I_k(0)$ (cf. identity~\eqref{eq:ikdeltaexpressionh}).
Since both $\eta_\kappa(\cdot)$ and $\Psi_\Delta^\varepsilon(\cdot)$ are even functions, we have from Eq.~\eqref{eq:Ikdefinition} by the variable substitution $z'=-z$ that $I_k(0)$ is even in  $k$. Indeed
\begin{align}
    I_k(0) &=\int_{-\varepsilon}^\varepsilon  \eta_\kappa(k+z)^2 \Psi^\varepsilon_\Delta(z)^2 dz \\
    &= \int_{-\varepsilon}^\varepsilon  \eta_\kappa(-k+z')^2 \Psi^\varepsilon_\Delta(z')^2 dz'\\
    &=I_{-k}(0) \ .\label{eq: Ikzero even}
\end{align}

By Eq.~\eqref{eq:shifted eta2} (with $m+z$ in place of $k$ and $L/2$ in place of $z$), we have for $m\geq 0$ and $z\in (-\varepsilon,\varepsilon)$ that
\begin{align}
\eta_\kappa(L/2+m+z)^2
&=\eta_\kappa(m+z)^2\cdot e^{-\kappa^2mL/2}\cdot e^{-\kappa^2zL/2}\cdot e^{-\kappa^2 L^2/4}\\
&\leq \eta_\kappa(m+z)^2\cdot e^{\kappa^2\varepsilon L/2}\cdot e^{-\kappa^2 L^2/4}\\
&\le \eta_\kappa(m+z)^2\cdot e^{-\kappa^2(L^2/4-\varepsilon L/2)}\\
&\le 
\eta_\kappa(m+z)^2\cdot e^{-\kappa^2L^2/8}\ ,
\label{eq:lasteqmv}\end{align}
where we used that $\kappa^2m L/2\geq 0$ and $z\ge -\varepsilon$ on the second line, and the following inequality
$\varepsilon L/2 \le L/2\le (L/2)\cdot (L/4)=L^2/8$
for $L\ge 4$ on the last line. Inserting the bound from Eq.~\eqref{eq:lasteqmv} into the identity for $I_k(y)$ from Eq.~\eqref{eq:ikdeltaexpressionh} gives
\begin{align}
    I_{L/2+m}(0)&\le e^{-\kappa^2 L^2/8}I_{m}(0)\qquad\textrm{ for $ m\geq 0$ }\  .
\end{align}
Since $I_k(0)$ is even in $k$ (cf.~\eqref{eq: Ikzero even}), we also have
\begin{align}
    I_{-L/2-m}(0)&\leq e^{-\kappa^2 L^2/8}I_{-m}(0)\qquad\textrm{ for $ m\geq 0$ }\  .
\end{align}
In particular, this implies
\begin{align}
\sum_{k=-L}^{L-1}I_k(0)&=\sum_{k=-L/2}^{L/2-1} I_k(0)+
\sum_{k=L/2+1}^{L-1} I_{-k}(0)+
\sum_{k=L/2}^{L-1}I_k(0)\\
&\leq \sum_{k=-L/2}^{L/2} I_k(0)+
e^{-\kappa^2L^2/8}\sum_{m=1}^{L/2-1}I_{-m}(0)+
e^{-\kappa^2L^2/8}\cdot \sum_{m=0}^{L/2-1}I_{m}(0)\\
&=\sum_{k=-L/2}^{L/2} I_k(0)+
e^{-\kappa^2L^2/8}\sum_{k=-L/2+1}^{L/2-1}I_k(0)\\
&\leq \left(1+e^{-\kappa^2 L^2/16}\right)\cdot \sum_{k=-L/2}^{L/2} I_k(0)\ .
\end{align}
In the last inequality,  we  used that $I_k(0)\geq 0$ for every $k\in\mathbb{Z}$. 
Rewriting this as 
\begin{align}
\left(1+e^{-\kappa^2 L^2/8}\right)^{-1}
\sum_{k=-L}^{L-1}I_k(0)&\leq \sum_{k=-L/2}^{L/2} I_k(0)\  ,
\end{align}
and using that $1-x\leq (1+x)^{-1}$  for $x\in (0,1)$ gives the claim.
\end{proof}

\begin{lemma}\label{lem:Ikprimeikzerocomparison} Let $I_k(0)$ be as defined in Eq.~\eqref{eq:Ikdefinition} with $\kappa>0$, $\Delta>0$, and $\varepsilon\in(0,1/2)$. Let $L\ge 0$.
\begin{align}
    \sum_{k=-L/2}^{L/2-1}I_k' &\geq \left(1-\kappa^2\varepsilon L/2\right) \sum_{k=-L/2}^{L/2-1}I_k(0)\ .
    \end{align}
\end{lemma}
\begin{proof}

We will prove the claim by comparing integrands in $I'_k$ and $I_k(0)$ (cf. identities~\eqref{eq:ikprime} and~\eqref{eq:ikdeltaexpressionh}).
By the identity~\eqref{eq:shifted eta eta}, we have 
\begin{align}
\eta_\kappa(k+z)\eta_\kappa(k)
&=\eta_\kappa(k+z)^2 \cdot e^{\kappa^2 k z}\cdot  e^{\kappa^2 z^2/2}\\
&\ge \eta_\kappa(k+z)^2 \cdot e^{\kappa^2 k z}\ , \label{eq:lowerboundetakappakz}
\end{align}
where we used that $e^{\kappa^2 z^2/2}\geq 1$ that holds for $\kappa>0$ and $z\in\mathbb{R}$. Furthermore, for $k\in \{-L/2,\ldots,L/2-1\}$ and  $z\in (-\varepsilon,\varepsilon)$, we have 
\begin{align}
e^{\kappa^2 kz}\geq e^{-\kappa^2 L\varepsilon/2}\ .
\end{align}
It follows that (for this choice of $z$ and $k$)
\begin{align}
\eta_\kappa(k+z)\eta_\kappa(k)&\geq \eta_\kappa(k+z)^2 
e^{-\kappa^2 L\varepsilon/2}\label{eq:lowerboundetakappakzbb} \ .
\end{align}
Inserting Eq.~\eqref{eq:lowerboundetakappakzbb}  into Eq.~\eqref{eq:ikprime}, we get  (for $k\in\{-L/2,\ldots,L/2-1\}$)
\begin{align}
I_k'&\geq e^{-\kappa^2 L\varepsilon/2} \int_{-\varepsilon}^\varepsilon \Psi^\varepsilon_\Delta(z)^2\eta_\kappa(k+z)^2dz\\
&=e^{-\kappa^2 L\varepsilon/2} I_k(0)\ ,
\end{align}
where we used the identity~\eqref{eq:ikdeltaexpressionh} on the last line. Summing over $k\in \{-L/2,\ldots,L/2-1\}$ and using the inequality $e^{-x}\geq 1-x$  for $x\ge 0$ gives the claim.
\end{proof}

\begin{lemma}\label{lem:deltazeroonehalfintegermtwo} Let $\kappa\in\mathbb{R}$,
$\varepsilon\in [0,1/2)$, and $L\in 2\mathbb{N}$. 
Let  $m\in \{-L/4,\ldots,L/4\}$  and $s\in (-1/2,1/2]$. Then
\begin{align}
     \sum_{k=-L/2-m}^{L/2-1-m}I_k(s)\leq  e^{\kappa^2 L} \sum_{k=-L/2-m}^{L/2-1-m}I_k(0)\ .
\end{align}
    \end{lemma}
    \begin{proof}
  We will prove the claim by comparing integrands in $I_k(s)$ (cf. identity~\eqref{eq:ikdeltaexpressionh}).
  By Eq.~\eqref{eq:shifted eta2}, we have
    \begin{align}
        \eta_\kappa((k+z)-s)^2
        &=\eta_\kappa (k+z)^2\cdot  e^{2\kappa^2(k+z)s}\cdot e^{-\kappa^2s^2} \\
        &\le  \eta_{\kappa}(k+z)^2 \cdot e^{2\kappa^2ks}\cdot e^{2\kappa^2 zs} \label{eq:chainofeqshmnp}
    \end{align}
    where we used $e^{-\kappa^2s^2}\leq 1$ (for $s,\kappa\in\mathbb{R}$) on the last line.
    Since for $s\in (-1/2,1/2]$, we have 
    \begin{align}
        \max_{k\in \{-L/2-m,\ldots,L/2-1-m\}}  e^{2\kappa^2ks }
        &\le    e^{\kappa^2(L/2+|m|) }
        \qquad\textrm{ for any }\qquad  m\in\mathbb{Z}\ ,
    \end{align}
        we obtain 
        \begin{align}
            e^{2\kappa^2ks }&\le  e^{3\kappa^2 L/4}
        \end{align}
        for any $m\in \{-L/4,\ldots,L/4\}$, $k\in  \{-L/2-m,\ldots,L/2-1-m\}$.
        Furthermore, we have
        \begin{align}
            e^{2\kappa^2zs}
            \le e^{\kappa^2\varepsilon}\qquad\textrm{ for any }z\in(-\varepsilon,\varepsilon)\textrm{ and }s\in (-1/2, 1/2]\ .
        \end{align}
        Inserting the previous two inequalities  into~\eqref{eq:chainofeqshmnp}, we conclude that for such $m,k,z,s$, we have
        \begin{align}
            \eta_\kappa((k+z)-s)^2&\le \eta_\kappa(k+z)^2\cdot 
            e^{3\kappa^2 L/4}\cdot e^{\kappa^2\varepsilon} \ . \label{eq:uniformlowerboundetakappakz}
        \end{align}
        Inserting~\eqref{eq:uniformlowerboundetakappakz} into~\eqref{eq:ikdeltaexpressionh},
        we obtain 
        \begin{align}
            I_k(s) &\le 
            e^{3\kappa^2 L/4+\kappa^2\varepsilon}\cdot
            \int_{-\varepsilon}^\varepsilon 
            \eta_\kappa(k+z)^2 \Psi^\varepsilon_\Delta(z)^2dz \\
            &=e^{\kappa^2(3L/4+\varepsilon)}
            \cdot I_k(0)\\
            &\leq e^{\kappa^2 L}I_k(0)\ ,\label{eq:kappaLupperboundkom}
            \end{align} 
            whenever $k\in \{-L/2-m,\ldots,L/2-1-m\}$ with $m\in \{-L/4,\ldots,L/4\}$. 
            The claim follows by summing~\eqref{eq:kappaLupperboundkom} over $k\in \{-L/2-m,\ldots,L/2-1-m\}$.
            \end{proof}

\section{Effective squeezing parameters} \label{sec: effective squeezing}
In this section, we bound the effective squeezing parameter, a metric to quantify the quality of GKP states (see~\cite{Weigand_2018, Hastrup_2021}). Let us define the unitaries $S_P = e^{-iP}$ and $S_Q= e^{2\pi i Q}$.
We want to quantify how invariant a state is under these unitaries. To do so, the effective squeezing parameters  $\Delta_P$ and $\Delta_Q$ are introduced as
    \begin{align}
        \Delta_P(\rho) := \sqrt{\log\left( \textfrac{1}{|\langle S_P\rangle_\rho|^2}\right)}\qquad \textrm{ and }\qquad
        \Delta_Q(\rho) := \sqrt{\log\left( \textfrac{1}{|\langle S_Q\rangle_\rho|^2}\right)}\ ,
    \end{align}
where we define $\langle U\rangle_\rho := \tr(U\rho)$ for a unitary $U$.
The motivation of these quantities is the case of the ideal GKP state
\begin{align}
    \ket{\gkp} \propto \sum_{z\in \mathbb{Z}} \ket{z}\ ,
\end{align}
which has stabilizers $S_P$ and $S_Q$ and for which we thus have $\Delta_P(\gkp) = \Delta_Q(\gkp) = 0$. We show the following.
\begin{lemma}
    Let $\rho\in\cB(L^2(\mathbb{R}))$ be a quantum state prepared by Protocol~\ref{prot: appx GKP state prep} (cf. Theorem~\ref{thm:main}). For $\kappa, \Delta$ sufficiently small, it holds that
    \begin{align}
        \Delta_P &\le 39 \Delta^{1/4} + 16\kappa^{1/6}\ ,\\
        \Delta_Q &\le    41 \Delta^{1/4} +16\kappa^{1/6}\ .
    \end{align}
\end{lemma}
\begin{proof}
    Note that for a unitary $U$ and two quantum state $\rho$ and $\sigma$, we have
    \begin{align}\label{eq: exp value trace distance}
         \left|\tr\left[U \rho\right]\right| \ge \left|\tr\left[U\sigma\right]\right| - \left|\tr\left[U(\rho - \sigma)\right]\right|\ge \left|\tr\left[U\sigma\right]\right| - \left\| \rho - \sigma \right\|_1\ .
    \end{align}
    The first inequality uses the triangle inequality; the second uses Hölder's inequality implying that for any trace class operator $X$ and any unitary $U$, we have $|\tr[UX]| \le \|X\|_1$. In particular, we conclude
    \begin{align} \label{eq: squared exp value}
        |\langle U \rangle_\rho|^2 \ge |\langle U \rangle_\sigma|^2 - 2\lVert \rho - \sigma \rVert_1\ .
    \end{align}
    By definition, $S_P = e^{-iP}$ and $S_Q = e^{2\pi i Q}$. Due to Lemma~\ref{lem: gkp S_P}, we have
    \begin{align}
        \left|\langle \gkp_{\kappa, \Delta}^\varepsilon, e^{-iP} \gkp_{\kappa,\Delta}^\varepsilon \rangle\right|^2 \ge 1 - 4\kappa\ . 
    \end{align}
    Moreover, setting $\varepsilon = \sqrt{\Delta}$, we infer from Corollary~\ref{cor: gkp gkp eps trace distance} that
    \begin{align}
        \left\| \proj{\gkp_{\kappa,\Delta}^\varepsilon} - \proj{\gkp_{\kappa,\Delta}} \right\|_1 &\le 6 \sqrt{\Delta}  \label{eq:tri first}   \ .
    \end{align}
    Finally, we have by Theorem~\ref{thm:main} that
    \begin{align}
        \left\|\rho - \proj{\gkp_{\kappa,\Delta}} \right\|_1 \le 190 \sqrt{\Delta} + 24 \kappa^{1/3} \label{eq: tri sec}\ .
    \end{align}
    Thus, by the triangle inequality applied to~\eqref{eq:tri first} and~\eqref{eq: tri sec}, we infer
    \begin{align}
        \left\| \rho - \proj{\gkp_{\kappa,\Delta}^\varepsilon} \right\|_1 \le 190\sqrt{ \Delta} + 30 \kappa^{1/3} \label{eq: tri third}\ .
    \end{align}
    Hence, by Eq.~\eqref{eq: squared exp value} along with~\eqref{eq: tri third}
    \begin{align}
        \left|\langle S_P\rangle_\rho\right|^2 \ge 1 - 380\sqrt{\Delta} - 64 \kappa^{1/3} \ .
    \end{align}
    As $\log((1-x)^{-1}) = -\log(1-x) \le 2x$ for $0\le x\le 1/2$, we deduce for small enough $\Delta$ and $\kappa$ that
    \begin{align}
        \Delta_P(\rho) \le 2\sqrt{380\sqrt{\Delta}+  64 \kappa^{1/3}} \le 39 \Delta^{1/4} + 16\kappa^{1/6}\ .
    \end{align}
    Similarly, by Lemma~\ref{lem: gkp S_Q} we have that
    \begin{align}
        \left|\langle \gkp_{\kappa, \Delta}^\varepsilon, e^{2\pi i Q} \gkp_{\kappa,\Delta}^\varepsilon \rangle\right|^2 \ge 1 - 40\varepsilon^2 = 1 - 40 \Delta\ . 
    \end{align}
     Hence, by Eq.~\eqref{eq: squared exp value} along with~\eqref{eq: tri third} and using $\sqrt{x} \ge x$ for $x \in[0,1]$, we have that
    \begin{align}
        \left|\langle S_Q\rangle_\rho\right|^2 \ge 1  - 40\Delta- 380\sqrt{\Delta} - 60 \kappa^{1/3} \ge 1   - 420 \sqrt{\Delta} - 60\kappa^{1/3}\ .
    \end{align}
    If follows that
    \begin{align}
        \Delta_Q(\rho) \le 2\sqrt{420 \sqrt{\Delta} + 60 \kappa^{1/3}} \le 41\Delta^{1/4} + 16\kappa^{1/6}\ ,
    \end{align}
    where we used $\sqrt{x+y} \le \sqrt{x}+\sqrt{y}$ for $x,y\ge 0$.
\end{proof}

\normalem

\end{document}